\documentclass[11pt,a4paper]{article}

\usepackage{setspace}
\usepackage{mathtools,amsmath,amssymb,amsthm,color,url}
\usepackage{mathrsfs,bm,comment,bbm}
\usepackage{float}
\usepackage{cases}
\usepackage{array}
\usepackage{multirow}

\usepackage[pdftex]{graphicx}
\usepackage {booktabs}

\usepackage{natbib}

\theoremstyle{plain}

\newtheorem{prop}{Proposition}[section]

\theoremstyle{definition}

\definecolor{darkblue}{rgb}{0.0,0.0,0.55}
\RequirePackage[colorlinks,citecolor=darkblue,urlcolor=darkblue,linkcolor=darkblue]{hyperref} 

\makeatletter
 \renewcommand{\theequation}{%
   \thesection.\arabic{equation}}
  \@addtoreset{equation}{section}
\makeatother

\usepackage[top=25mm,bottom=27mm,left=25mm,right=25mm]{geometry}

\usepackage{titlesec}
\titleformat*{\section}{\large\bfseries}
\titleformat*{\subsection}{\it}

  \allowdisplaybreaks
 
\begin{document}

\title{\textbf{Global-local shrinkage priors for modeling random effects in multivariate spatial small area estimation}}
\author{Shushi Nishina$^1$, Takahiro Onizuka$^2$ and Shintaro Hashimoto$^1$\\
$^1$ Department of Mathematics, Hiroshima University\\
$^2$ Graduate School of Social Sciences, Chiba University}
\date{\today}

\maketitle

\begin{abstract}

Small area estimation (SAE) plays a central role in survey statistics and epidemiology, providing reliable estimates for domains with limited sample sizes. The multivariate Fay-Herriot model has been extensively used for this purpose, because it enhances estimation accuracy by borrowing strength across multiple correlated variables. In this paper, we develop a Bayesian extension of the multivariate Fay-Herriot model that enables flexible, component-specific shrinkage of the random effects. The proposed approach employs global-local priors formulated through a sandwich mixture representation, allowing adaptive regularization of each element of the random-effect vectors. This construction yields greater robustness and prevents excessive shrinkage in areas exhibiting strong underlying signals. In addition, we incorporate spatial dependence into the model to account for geographical correlation across small areas. The resulting spatial multivariate framework simultaneously exploits cross-variable relationships and spatial structure, yielding improved estimation efficiency. The utility of the proposed method is demonstrated through simulation studies and an empirical application to real survey data.

\end{abstract}


\section{Introduction}
\label{sec:1}

Small area estimation (SAE) has become an indispensable framework in disciplines such as survey statistics and epidemiology, where reliable estimates are required for domains with limited or no reliable direct data \citep{rao2015small, sugasawa2020small}. Among various SAE approaches, the Fay-Herriot model \citep{fay1979estimates}, a widely used area-level model, has been particularly influential due to its flexibility and computational efficiency. Multivariate extensions of the Fay-Herriot model have attracted growing attention, as they exploit correlations among multiple related variables to enhance estimation accuracy \citep{fay1987application, datta1998multivariate}. Moreover, when spatial information is available, explicitly modeling spatial dependence can further improve estimation by incorporating geographic structure. Several spatially correlated versions of the Fay-Herriot model have been proposed \citep[see, e.g.,][]{porter2014spatial, porter2015small, chung2022bayesian, tang2023global, wang2025variational}; however, the integration of spatial dependence with flexible shrinkage mechanisms remains relatively underexplored.

An important aspect of SAE modeling lies in the specification and assessment of random effects. In many practical applications, certain areas may not require random effects; thus, over-shrinkage or inappropriate inclusion of such effects can lead to biased or inefficient estimators. Consequently, determining the presence or absence of random effects, or, more generally, allowing for flexible shrinkage, is crucial for enhancing model robustness and interpretability.
There has been extensive research on testing for random effects within the univariate small area estimation framework. For example, \citet{datta2011model} proposed a test-based approach under the null hypothesis that the common random effect variance equals zero. Over the past decade, studies addressing the sparsity of random effect parameters in small area estimation have also emerged. Although the Fay-Herriot model assumes a common area-level variance component across all small areas, random effects may be unnecessary in many of them. This assumption can be particularly problematic when the number of small areas is very large. 
To address this issue, \citet{datta2015small} introduced a two-component mixture distribution for the random effect parameters, also known as the spike-and-slab prior \citep[see also][]{chakraborty2016two}. \citet{tang2018modeling} further proposed the use of global-local shrinkage priors for univariate small area estimation and established several theoretical properties of the resulting shrinkage factors. This approach builds upon the expanding literature on Bayesian shrinkage priors such as the horseshoe prior \citep{carvalho2010horseshoe}, normal-gamma prior \citep{griffin2010inference} and Dirichlet-Laplace prior \citep{bhattacharya2015dirichlet} which provide a principled framework for borrowing strength while encouraging sparsity. More recently, \citet{ghosh2022multivariate} extended the work of \citet{tang2018modeling} to a multivariate setting by introducing local parameters for each area, while \citet{tang2023global} considered the application of global-local priors to spatial small area estimation with univariate response variables.

In this paper, we develop a new Bayesian framework for multivariate spatial small area estimation by introducing global-local shrinkage priors within the multivariate spatial Fay-Herriot model. The proposed approach enables adaptive and spatially informed regularization of random effects, providing a flexible mechanism for borrowing strength both across areas and across multiple correlated variables.
Specifically, we construct a novel shrinkage prior that combines a sandwich scale-mixture representation with a multivariate conditional autoregressive (MCAR) structure. This hierarchical formulation introduces local parameters whose number equals the product of the dimension of the multivariate direct estimates, thereby achieving element-wise, spatially correlated shrinkage of the random-effect vectors. For modeling these local parameters, we mainly consider both the normal-gamma and horseshoe priors, representing two complementary regimes of sparsity: the horseshoe prior performs well under strong sparsity, whereas the normal-gamma prior may be more effective in moderately sparse settings. Importantly, the proposed framework is general and can accommodate alternative shrinkage priors according to the needs of the analysis.
Our method extends and generalizes the work of \citet{ghosh2022multivariate} by incorporating spatial dependence and adaptive shrinkage within a unified Bayesian structure. 
To the best of our knowledge, this study represents the first attempt to apply global-local shrinkage priors to multivariate spatial small area estimation. While existing studies have investigated global-local shrinkage priors in multivariate non-spatial or univariate spatial settings, their integration into a multivariate spatial framework remains unexplored. 
The main methodological innovation of this paper lies in proposing a novel modeling framework that combines a sandwich-type covariance structure with the MCAR model through global-local shrinkage priors. Although the sandwich covariance structure and the MCAR model are individually established techniques, our approach uniquely integrates them to achieve element-wise (component-specific) shrinkage of random effect vectors in multivariate spatial small area estimation. This formulation is both novel and practically useful, as it effectively captures heterogeneity across multiple response variables while properly accounting for spatial dependence.
This development fills an important methodological gap, as spatial small area models have traditionally relied on homogeneous random-effect structures that may be overly restrictive in high-dimensional or heterogeneous settings.
We further develop a Markov chain Monte Carlo (MCMC) algorithm for posterior inference and demonstrate, through simulation studies, that the proposed method achieves improved estimation accuracy and robustness compared with existing approaches. Finally, we illustrate its practical utility using data from the West Census Region based on the 2022 five-year American Community Survey (ACS) on median household income and poverty rate.

The remainder of this paper is organized as follows.
Section~\ref{sec:2} provides a brief overview of the multivariate Fay-Herriot model and introduces the modeling of random effects using global-local shrinkage priors.
In Section~\ref{sec:3}, we propose a new multivariate spatial Fay-Herriot model that incorporates shrinkage priors through a sandwich scale-mixture structure.
Section~\ref{sec:4} presents the results of numerical experiments, and Section~\ref{sec:5} illustrates the proposed methodology through an application to real data.
Finally, Section~\ref{sec:6} concludes the paper with a discussion and possible directions for future research. 
R code implementing the proposed methods is available in the GitHub repository \url{https://github.com/Takahiro-Onizuka/GL-MCAR}.

\section{Multivariate FH model and sparsity for random effects}
\label{sec:2}

We consider multivariate small area estimation based on area-level models. The multivariate Fay-Herriot (FH) model \citep{fay1987application} is a multivariate extension of the Fay-Herriot model \citep{fay1979estimates}. The multivariate FH model is defined by
\begin{align}\label{FHmodel}
\begin{split}
\bm{y}_i &= \bm{\theta}_i + \bm{\varepsilon}_i,\quad \bm{\theta}_i =\mathbf{X}_i \bm{\beta} + \bm{u}_i,\quad \bm{\varepsilon}_i\sim \mathcal{N}_k (\bm{0}, \mathbf{V}_i), \\
\bm{u}_i& \sim \mathcal{N}_{k}(\bm{0}, \mathbf{\Psi}), \quad i=1,\dots,m,
\end{split}
\end{align}
where $\bm{y}_i$ is a $k\times 1$ vector, $\mathbf{X}_i$ is a $k\times s$ matrix, and $\bm{\beta}$ is a $s\times 1$ vector. $\mathbf{V}_i$ is a known covariance matrix and $\mathbf{\Psi}$ is a $k\times k$ covariance matrix of the random effect vector $\bm{u}_i$. Let $\bm{y}=(\bm{y}_1^{\top},\dots,\bm{y}_m^{\top})^{\top} \in \mathbb{R}^{mk}$, $\mathbf{X}=(\mathbf{X}_1^{\top},\dots, \mathbf{X}_m^{\top})^{\top} \in \mathbb{R}^{mk\times s}$, $\bm{u}=(\bm{u}_1^{\top},\dots, \bm{u}_m^{\top})^{\top} \in \mathbb{R}^{mk}$ and $\bm{\varepsilon}=(\bm{\varepsilon}_1^{\top},\dots, \bm{\varepsilon}_m^{\top})^{\top}\in \mathbb{R}^{mk}$. The model is expressed as 
\[\bm{y}=\mathbf{X}\bm{\beta}+\bm{u}+\bm{\varepsilon}, \quad \bm{u}\sim \mathcal{N}_{mk}(\bm{0}, \mathbf{I}_m \otimes \mathbf{\Psi}),\quad \bm{\varepsilon}\sim \mathcal{N}_{mk}(\bm{0},\mathbf{V}),\] 
where $\otimes$ is Kronecker product of matrices  and $\mathbf{V}=\mathrm{block \ diag}(\mathbf{V}_1,\dots, \mathbf{V}_m) \in \mathbb{R}^{mk\times mk}$. We assume that $\mathbf{X}$ is of full rank. For example, we consider the crop data analyzed by \cite{Battese01031988}, who studied the data using the nested error regression model. For the $i$th county, let $y_{i1}$ and $y_{i2}$ be survey data of average areas of corn and soybean, respectively. Also let $x_{i1}$ and $x_{i2}$ be satellite data of average areas of corn and soybean, respectively. In this case, $\bm{y}_i$, $\mathbf{X}_i$ and $\bm{\beta}$ correspond to 
\begin{align*} 
\bm{y}_i=(y_{i1}, y_{i2})^{\top},\quad \mathbf{X}_i=\begin{pmatrix}
1 & x_{i1} & x_{i2} & 0 & 0 & 0 \\
0 & 0 & 0 & 1 & x_{i1} & x_{i2}
\end{pmatrix}, \quad \bm{\beta}=(\beta_1,\dots,\beta_6)^{\top}
\end{align*}
for $k=2$ and $s=6$. In small area estimation, our aim is to predict $\bm{\theta}_i$ for each area. Although many approaches have been proposed for efficiently predicting $\bm{\theta}_i$, we focus on the hierarchical Bayesian approach \citep[e.g.,][]{datta1999hierarchical} in this paper. 

Although the multivariate Fay-Herriot (FH) model has received relatively less attention than its univariate counterpart, several important theoretical and applied studies have been conducted. From a frequentist perspective, \citet{benavent2016multivariate} considered multivariate FH models with autoregressive covariance structures for the random-effect vectors and investigated their mean squared errors. \citet{ITO202112} studied empirical Bayes confidence regions for model~\eqref{FHmodel}. In the Bayesian framework, \citet{porter2015small} proposed multivariate spatial FH models that incorporate multivariate conditional autoregressive (MCAR) priors through a hierarchical Bayesian approach.

The determination of whether multivariate random effects are present remains a critical methodological challenge. As discussed in Section \ref{sec:1}, substantial progress has been made in the univariate setting, where Bayesian shrinkage estimation methods for random effects based on global-local shrinkage priors have been extensively developed. In contrast, methodological advances for the multivariate case are still relatively scarce.
Recently, \cite{ghosh2022multivariate} extended global-local priors to the multivariate setting. Based on model~\eqref{FHmodel}, they modeled the random effect vector as
\begin{align}\label{Ghosh-prior}
\bm{u}_i \sim \mathcal{N}_{k}(\bm{0}, \lambda_i^2 \mathbf{\Sigma}), \quad \lambda_i \sim \pi(\lambda_i), \quad \mathbf{\Sigma}\sim \pi(\mathbf{\Sigma}), \quad i=1,\dots,m,
\end{align}
where $\lambda_i$ is a local shrinkage parameter for area $i$, and $\mathbf{\Sigma}$ is a $k \times k$ covariance matrix. This prior is both simple in form and theoretically tractable, and the authors derived theoretical results on the concentration of the shrinkage factor. The approach performs well in most practical settings.
Nonetheless, the model of \citet{ghosh2022multivariate} has two limitations. First, since each area is assigned a single local parameter $\lambda_i$, element-wise shrinkage within each area is not possible. Second, the model does not account for spatial dependence. Element-wise shrinkage is particularly important for capturing heterogeneity among observed variables within a given area. For instance, in a specific region, either median income or median rent may take extreme values relative to other areas. If such a region represents a hotspot, it is desirable that the model can appropriately shrink each component to accurately reflect the underlying pattern. To the best of our knowledge, element-wise shrinkage of random effect vectors has not yet been incorporated into multivariate small area models.
To address these limitations, the next section proposes a hierarchical Bayesian model for multivariate small area estimation that achieves element-wise shrinkage of the random effect vectors $\bm{u}_i$ while simultaneously incorporating spatial correlation.


\section{Global-local priors based on sandwich-mixture covariance structure for multivariate spatial FH models}
\label{sec:3}

When spatial information is available for each area, predictive accuracy can be further improved by accounting not only for correlations among variables but also for spatial dependence. In this section, we extend the multivariate small area model introduced in the previous section by incorporating spatial correlation, and propose a novel model that allows for flexible, element-wise shrinkage of the random effects.

\subsection{The proposed methods}
\label{sec:3.1}

Following \cite{gelfand2003proper}, the proper multivariate conditional autoregressive (MCAR) prior is defined as
\begin{align}\label{MCAR}
\bm{u} \sim \mathcal{N}_{mk}(\bm{0}, (\mathbf{D}-\rho \mathbf{W})^{-1} \otimes \mathbf{\Sigma}  ),
\end{align}
where $(\mathbf{D}-\rho \mathbf{W})^{-1} $ is a $m\times m$ matrix, $\mathbf{\Sigma}$ is a $k\times k$ matrix, $\mathbf{W}$ is an $m\times m$ adjacency matrix whose $(i,j)$-th element is $w_{ij}$, $\mathbf{D}$ is a diagonal matrix with the diagonal elements being the row sums of $W$, and $\rho \in [0,1)$ is a spatial dependence parameter. Under this setting, the prior distribution \eqref{MCAR} becomes proper.
Assuming the prior \eqref{MCAR} for random effects, \cite{porter2015small} and \cite{wang2025variational} considered spatial small area estimation via the multivariate Fay-Herriot model: 
\begin{align}\label{model-spatial}
\begin{split}
\bm{y}_i &= \bm{\theta}_i + \bm{\varepsilon}_i,\quad \bm{\theta}_i =\mathbf{X}_i \bm{\beta} + \bm{u}_i,\quad \bm{\varepsilon}_i\sim \mathcal{N}_k (\bm{0}, \mathbf{V}_i), \quad \pi(\bm{\beta})=1,\\
\bm{u} &\sim \mathcal{N}_{mk}(\bm{0}, (\mathbf{D}-\rho \mathbf{W})^{-1} \otimes \mathbf{\Sigma}  ), \quad \mathbf{\Sigma} \sim \pi(\mathbf{\Sigma}), \quad \rho\sim \pi(\rho) ,\quad i=1,\dots,m,
\end{split}
\end{align}
where $\bm{u}$ is a $mk\times 1$ vector of spatial random effect, which follows a separable multivariate CAR (MCAR) distribution \citep{carlin2003hierarchical, gelfand2003proper}. $\mathbf{\Sigma}$ is the $k\times k$ correlation matrix. This model is called separable because the correlation across variables is decoupled from the spatial covariance. The parameter $\rho \in [0,1)$ controls the strength of spatial dependence. If $\rho=0$, $\bm{u}_1,\dots,\bm{u}_m$ are independent. When $\rho=0$ and $\mathbf{D}=\mathbf{I}_m$, the model \eqref{model-spatial} corresponds to the usual multivariate FH model. Moreover, when $\rho = 1$, the model reduces to the intrinsic MCAR model \citep{mardia1988multi}, which corresponds to an improper prior distribution.
Although the resulting posterior distribution may still be proper and can therefore be used for posterior inference, model comparison becomes difficult under an improper prior.
For this reason, we adopt a proper MCAR prior in this study.
In particular, model comparison in Section~\ref{sec:5} is conducted using the Deviance Information Criterion (DIC).

As a prior distribution for $\bm{u}$ in \eqref{model-spatial}, we propose the following prior distribution for $\bm{u}$:
\begin{align}\label{proposal}
\begin{split}
\bm{u} \mid \mathbf{\Sigma}, \mathbf{\Lambda}, \rho,\tau & \sim \mathcal{N}_{mk}(\bm{0},  \tau^2\mathbf{\Lambda}((\mathbf{D}-\rho \mathbf{W})^{-1} \otimes \mathbf{\Sigma}) \mathbf{\Lambda} ),\\
\bm{\lambda_i}=(\lambda_{i,1},\dots,\lambda_{i,k})^\top & \sim \prod_{j=1}^k \pi(\lambda_{i,j}), \quad i=1,\dots,m,\\
\tau&\sim \pi(\tau),
\end{split}
\end{align}
where $\mathbf{\Lambda}=\mathrm{blockdiag}(\mathbf{\Lambda}_1,\dots, \mathbf{\Lambda}_m) \in \mathbb{R}^{mk\times mk}$ with $\mathbf{\Lambda}_i=\mathrm{diag}(\lambda_{i,1},\dots,\lambda_{i,k})$. 
Following \cite{hamura2025outlier}, we call the prior \eqref{proposal} {\it sandwich mixture prior}. Such sandwich-type covariance matrices have recently been employed in various models. For instance, \cite{tang2023global} utilized them in univariate spatial small area estimation, while \cite{hamura2025outlier} applied them in multivariate robust Bayesian inference. 

To improve identifiability and avoid overparameterization between the covariance (correlation) between variables $\mathbf{\Sigma}$ and the local scale parameters $\mathbf{\Lambda}$, we adopt the separation strategy \citep{barnard2000modeling}. Specifically, while the prior for $\mathbf{\Sigma}$ is specified via an inverse-Wishart distribution, we sample and implement the correlation matrix by standardizing $\mathbf{\Sigma}$ \citep[see also Chapter 12 of][]{hoff2009first} at each MCMC step. This decouples the correlation dynamics from the scale parameters. That is, the role of variances in $\mathbf{\Sigma}$ is delegated to $\mathbf{\Lambda}$ and $\tau^2$. Therefore, although $\mathbf{\Sigma}$ is defined as a covariance matrix in model \eqref{proposal}, it is treated as a correlation matrix in the computational steps. 
Furthermore, the interplay between the local parameters $\mathbf{\Lambda}$ and the global shrinkage parameter $\tau$ is governed by the standard global-local shrinkage hierarchy, ensuring stable and adaptive shrinkage. It should be noted, however, that employing a separation strategy does not automatically guarantee full identifiability of all model parameters in a strict theoretical sense. Nevertheless, separating the estimation of variance scales and correlation structures effectively prevents them from becoming confounded, thereby dramatically improving the mixing performance and numerical stability of the MCMC algorithm.

From the formulation 
\begin{align*}
\bm{y}\mid \bm{\theta} \sim \mathcal{N}_{mk}(\bm{\theta},\mathbf{V}), \quad \bm{\theta}\mid \bm{\beta},\mathbf{\Sigma},\mathbf{\Lambda}, \rho, \tau\sim \mathcal{N}_{mk}(\mathbf{X}\bm{\beta}, \tau^2\mathbf{\Lambda}((\mathbf{D}-\rho \mathbf{W})^{-1} \otimes \mathbf{\Sigma}) \mathbf{\Lambda}),
\end{align*}
the conditional posterior mean of the small area mean vector is given by
\[\mathrm{E}(\bm{\theta} \mid \bm{\beta},\mathbf{\Sigma}, \mathbf{\Lambda}, \rho, \tau,\bm{y})=(\mathbf{I}_{mk} - \mathbf{Q})\bm{y} +\mathbf{Q}\mathbf{X} \bm{\beta},\]
where $\mathbf{Q}=(\mathbf{V}^{-1} + \mathbf{U}^{-1})^{-1} \mathbf{U}^{-1}$ and $\mathbf{U}=\tau^2\mathbf{\Lambda}((\mathbf{D}-\rho \mathbf{W})^{-1} \otimes \mathbf{\Sigma}) \mathbf{\Lambda}$.
The matrix $\mathbf{Q}$ is called the matrix shrinkage factor. Since the matrix shrinkage factor $\mathbf{Q}$ is not a diagonal matrix unlike theta in the Fay-Herriot model, the interpretation tends to be difficult. In contrast to the prior distribution defined in equation \eqref{Ghosh-prior}, the sandwich-type prior poses theoretical challenges for analyzing the posterior distribution of the shrinkage factor matrix. 
To better understand the shrinkage behavior despite the non-diagonal nature of $\mathbf{Q}$, we can consider scalar summary measures, such as the average shrinkage weight for each small area or variable, computed by examining the diagonal blocks of $\mathbf{Q}$ or its trace properties. In practice, because $\mathbf{Q}$ incorporates both the spatial precision matrix $(\mathbf{D}-\rho \mathbf{W})$ and the cross-variable covariance $\mathbf{\Sigma}$ through the sandwich structure, the shrinkage of a specific small-area estimate is not only pulled toward the regression trend $\mathbf{X}\bm{\beta}$ but also borrows strength simultaneously across neighboring areas and correlated response variables. Consequently, the off-diagonal blocks of $\mathbf{Q}$ dictate how spatial smoothing and cross-variable compromise operate jointly, providing a data-adaptive weighting scheme that cannot be achieved by scalar shrinkage factors in classical univariate models.
Following previous studies \cite[e.g.,][]{ghosh2022multivariate, tang2023global}, we also assume $\pi(\bm{\beta}) \propto 1$. 

\begin{prop}[Posterior propriety]\label{propriety}
Under the prior distribution \eqref{proposal}, if the priors of $\mathbf{\Sigma}$, $\rho$, $\tau$ and $\lambda_{i,j}$ are proper, then the posterior density is proper.
\end{prop}
The proof is given in Appendix \ref{sec:appA}.

\subsection{Prior specification}
\label{sec:3.2}

We assume that $\mathbf{\Sigma}\sim \mathrm{IW}(\nu_0, \mathbf{B})$ for fixed hyper-parameters $\nu_0$ and $\mathbf{B}$, where $\mathrm{IW}(\nu,\mathbf{S})$ is the inverse-Wishart distribution with degrees of freedom $\nu$ and covariance matrix $\mathbf{B}$. We employ $\nu_0=k$ and $\mathbf{B}=\mathbf{I}_k$ as default values of these hyper-parameters. Also, we assume the discrete uniform prior for $\rho$. Following \cite{gelfand2003proper}, we put equal mass on the following 31 values: $0, 0.05, 0.1, \dots , 0.8, 0.82, 0.84, \dots , 0.90, 0.91, 0.92, \dots , 0.99$. 
Regarding the sensitivity analysis for the selection of the grid for $\rho$ and the value of $\nu_0$, we demonstrate in the supplementary material that the results are robust against these choices.

The selection of a prior for local parameters is crucial for global-local shrinkage estimation of $\bm{u}_i$. In the literature, prior distributions for local shrinkage parameter can be divided into two
groups, one with exponential tails and the other with polynomial tails. We consider two priors for global and local parameters.
\begin{itemize}
\item (Normal-gamma prior): If $\tau^2=1$ and $\pi(\lambda_{i,j})=\frac{2 b^a}{\Gamma(a)} \lambda_{i,j}^{2a-1}\exp(-b \lambda_{i,j}^2)$, the corresponding prior is called normal-gamma prior \citep{griffin2010inference, tang2024hierarchical}. When $a=1$, $\lambda_{i,j}^2$ follows the exponential distribution \citep[see also,][]{park2008bayesian}. As the parameter $a$ approaches zero, the contraction effect in the neighborhood of the origin becomes increasingly significant.

\item (Horseshoe prior): If $\pi(\lambda_{i,j})=(2/\pi)(1+\lambda_{i,j}^2)^{-1}$, where $C_+(0,1)$ denotes the standard half-Cauchy distribution and $\tau\sim \mathrm{C}_+(0,1)$, the corresponding prior is called horseshoe prior \citep{carvalho2010horseshoe}. Compared to the normal-gamma prior, the horseshoe prior provides markedly more aggressive shrinkage in the vicinity of zero, yet remains robust in that it applies minimal shrinkage to large signals.
\end{itemize}
Under the prior distribution \eqref{proposal}, the marginal prior density of
$\bm{u}$ evaluated at $\bm{0}$, after integrating out $\bm{\lambda}$, is given by
\begin{align*}
\pi(\bm{u}=\bm{0})
&= C_1 \int
\bigl|
\tau^2 \mathbf{\Lambda}
\bigl((\mathbf{D}-\rho \mathbf{W})^{-1} \otimes \mathbf{\Sigma}\bigr)
\mathbf{\Lambda}
\bigr|^{-1/2}
\prod_{i=1}^m \prod_{j=1}^k \pi(\lambda_{i,j}) \, d\bm{\lambda} \\
&= C_2 \int
|\mathbf{\Lambda}|^{-1}
\prod_{i=1}^m \prod_{j=1}^k \pi(\lambda_{i,j}) \, d\bm{\lambda} = C_2 \int
\prod_{i=1}^m \prod_{j=1}^k
\lambda_{i,j}^{-1} \, \pi(\lambda_{i,j}) \, d\lambda_{i,j},
\end{align*}
where $C_1$ and $C_2$ denote constants that do not depend on
$\bm{\lambda}$.
If $\lambda_{i,j} \sim C_+(0,1)$, that is, a standard half-Cauchy distribution,
the above integral diverges to infinity.
This implies that the marginal prior density of $\bm{u}$ under the half-Cauchy
prior has a pole at $\bm{u}=\bm{0}$, thereby inducing strong shrinkage toward
zero.
The normal-gamma prior does not induce shrinkage as aggressively as the horseshoe prior, and its tails are comparatively lighter. However, it is expected to perform well when the underlying sparsity is moderate. Therefore, in this paper, we employ both of these priors as prior distributions for the local parameters. Although various other types of shrinkage priors exist and they are also conceptually applicable, it has been empirically observed that they tend to exhibit broadly similar performance in practical applications \citep[e.g.,][]{ghosh2022multivariate}.

Although the inverse-Wishart (IW) prior is employed in this study to preserve exact Gibbs sampling conjugacy and computational efficiency, alternative specifications for the correlation matrix, such as the LKJ prior \citep{lewandowski2009generating}, offer greater flexibility by decoupling variance and correlation structures. While adopting an LKJ prior would necessitate computationally intensive sampling steps (e.g., Metropolis-Hastings or Hamiltonian Monte Carlo) and increase the burden in large-scale spatial settings, exploring such flexible priors remains a valuable direction for future research, particularly when prior assumptions on correlation need to be less restrictive.

\subsection{Posterior computation}
\label{sec:3.3}

We present an MCMC algorithm for posterior computation. 
The random effect vector $\bm{u}$ can be expressed as 
\begin{align}\label{transformed-u}
\bm{u}=\mathbf{\Lambda}\tilde{\bm{u}}, \quad \tilde{\bm{u}} \sim \mathcal{N}_{mk}(\bm{0}, \tau^2( (\mathbf{D}-\rho \mathbf{W})^{-1} \otimes \mathbf{\Sigma})).
\end{align}
Furthermore, from the definition of the proper MCAR model, we note that the full conditional prior distribution of $\tilde{\bm{u}}_i$ is given by
\begin{align}\label{MCAR-fullcond}
\tilde{\bm{u}}_i \mid \{\tilde{\bm{u}}_j\}_{j\ne i} \sim \mathcal{N}_k\left(\rho\sum_{j\sim i} \frac{1}{w_{i+}} \tilde{\bm{u}}_j,\frac{\tau^2}{w_{i+}} \mathbf{\Sigma}\right),
\end{align}
where $w_{i+}=\sum_{j=1}^m w_{ij}$.

The joint posterior density is proportional to
\begin{align*}
&\pi(\tilde{\bm{u}}, \bm{\beta}, \mathbf{\Sigma}, \mathbf{\Lambda}, \rho, \tau^2 \mid \bm{y})\\
&\propto \exp\left[-\frac{1}{2}(\bm{y}-\mathbf{X}\bm{\beta}-\mathbf{\Lambda}\tilde{\bm{u}})^{\top} \mathbf{V}^{-1}(\bm{y}-\mathbf{X}\bm{\beta}-\mathbf{\Lambda}\tilde{\bm{u}})\right]\\
&\quad \times |\tau^2((\mathbf{D}-\rho \mathbf{W})^{-1} \otimes \mathbf{\Sigma})|^{-1/2}  \exp\left[-\frac{1}{2}\tilde{\bm{u}}^{\top} \frac{1}{\tau^2}((\mathbf{D}-\rho \mathbf{W}) \otimes \mathbf{\Sigma}^{-1})\tilde{\bm{u}}\right]\\
&\quad \times |\mathbf{\Sigma}|^{-(\nu_0+k+1)/2} \exp\{-\mathrm{tr}(\mathbf{B} \mathbf{\Sigma}^{-1})/2\} \left(\prod_{i=1}^m \prod_{j=1}^k \pi(\lambda_{i,j}) \right) \pi(\rho) \pi(\tau^2),
\end{align*}
where we set $\pi(\tau^2)=\delta_1$ where $\delta_1$ denotes a point mass at $1$ and $\pi(\lambda_{i,j})=\pi(\lambda_{i,j}\mid a,b)\pi(a)\pi(b)$ for the normal-gamma prior.

Under prior distributions presented in Section \ref{sec:3.2}, we summarize the MCMC algorithm as follows. The derivations are provided in Appendix \ref{sec:B1}.

\medskip
\noindent
{\bf MCMC algorithm}

\begin{itemize}
\item[-] Draw $\bm{\beta}$ from $\mathcal{N}_s(\bm{m}, (\bm{X}^{\top}\mathbf{V}^{-1}\mathbf{X})^{-1})$, where $\bm{m}=(\mathbf{X}^{\top}\mathbf{V}^{-1}\mathbf{X})^{-1}\mathbf{X}^{\top}\mathbf{V}^{-1}\bm{z}$, $\bm{z}:=\bm{y}-\mathbf{\Lambda}\tilde{\bm{u}}$ and $\mathbf{V}=\mathrm{block diag}(\mathbf{V}_1,\dots, \mathbf{V}_m) \in \mathbb{R}^{mk\times mk}$.

\item[-] Draw $\tilde{\bm{u}}_i$ from $\mathcal{N}_k(\mathbf{A}^{-1}\bm{b}, \mathbf{A}^{-1})$, where $\mathbf{A}=\mathbf{\Lambda}_i\mathbf{V}_i^{-1}\mathbf{\Lambda}_i +\mathbf{\Omega}_{\mathrm{CAR},i}^{-1}$ and $\bm{b}=\mathbf{\Lambda}_i \mathbf{V}_i^{-1}\bm{\xi}_i + \mathbf{\Omega}_{\mathrm{CAR},i}^{-1} \bm{\mu}_{\mathrm{CAR},i}$ with $\bm{\mu}_{\mathrm{CAR},i}=\rho\sum_{j\sim i} \tilde{\bm{u}}_j/w_{i+}$, $\mathbf{\Omega}_{\mathrm{CAR},i}=(\tau^2/w_{i+}) \mathbf{\Sigma}$, and $\bm{\xi}_i=\bm{y}_i-\mathbf{X}_i\bm{\beta}$.

\item[-] Draw $\mathbf{\Sigma}$ from $\mathrm{IW}(m+\nu_0, \mathbf{B}+\tau^{-2}\tilde{\mathbf{U}}(\mathbf{D}-\rho \mathbf{W}) \tilde{\mathbf{U}}^{\top})$, where $\tilde{\mathbf{U}}:=(\tilde{\bm{u}}_1,\dots, \tilde{\bm{u}}_m) \in \mathbb{R}^{k\times m}$. 
Subsequently, convert $\mathbf{\Sigma}$ into the corresponding correlation matrix to ensure numerical stability.

\item[-] Draw $\rho$ using the full conditional distribution
\begin{align*}
\pi(\rho \mid -) \propto  |\mathbf{D}-\rho \mathbf{W}|^{k/2} \times \exp\left[-\frac{1}{2} \tau^{-2}\tilde{\bm{u}}^{\top}  ((\mathbf{D}-\rho \mathbf{W}) \otimes \mathbf{\Sigma}^{-1})\tilde{\bm{u}}\right].
\end{align*}

\item[-] Draw $\bm{\lambda}_i$ using the full conditional distribution
\begin{align}
\pi(\bm{\lambda}_i \mid -) 
\propto  \exp\left[-\frac{1}{2}\left(\frac{\bm{y}_i-\mathbf{X}_i\bm{\beta}}{\tilde{\bm{u}}_i}-\bm{\lambda}_i\right)^{\top} \tilde{\mathbf{U}}_i\mathbf{V}_i^{-1}\tilde{\mathbf{U}}_i\left(\frac{\bm{y}_i-\mathbf{X}_i\bm{\beta}}{\tilde{\bm{u}}_i}-\bm{\lambda}_i\right)\right] \prod_{j=1}^k\pi(\lambda_{i,j}), \label{full-cond-lam}
\end{align}
where the fraction $\frac{\bm{y}_i-\mathbf{X}_i\bm{\beta}}{\tilde{\bm{u}}_i}$ is taken elementwise and $\tilde{\mathbf{U}}_i=\mathrm{diag}(\tilde{u}_{i1},\dots,\tilde{u}_{ik})$.

\item[-] Sampling of $\tau^2$ (horseshoe prior):
Using 
\[\tau\sim \mathrm{C}_+(0,1) \iff \tau^2\mid \psi \sim \mathrm{IG}(1/2, 1/\psi), \quad \psi\sim \mathrm{IG}(1/2,1),\]
draw $\tau^2$ from $\mathrm{IG}\left( \{(mk+1)/2\},\tilde{\bm{u}}^{\top} ((\mathbf{D}-\rho \mathbf{W}) \otimes \mathbf{\Sigma}^{-1})\tilde{\bm{u}}/2  + (1/\psi) \right)$ and draw $\psi$ from $\mathrm{IG}\left(1/2,(1/\tau^2) + 1 \right)$

\item[-] Sampling of $a$ and $b$ (normal-gamma prior):
Under the priors $a\sim U(0,1)$ and $b\sim \mathrm{Ga}(c,d)$, draw $b\sim \mathrm{Ga}\left(amk + c, \sum_{i=1}^m \sum_{j=1}^k \lambda_{i,j}^2 + d \right)$ and draw $a$ using the full conditional distribution
\[\pi(a\mid -)\propto \prod_{i=1}^m \prod_{j=1}^k \frac{b^a}{\Gamma(a)} \lambda_{i,j}^{2a-1}= \left(\frac{b^a}{\Gamma(a)}\right)^{mk} \left( \prod_{i=1}^m \prod_{j=1}^k \lambda_{i,j} \right)^{2a-1}.\]
\end{itemize}
The first term on the right-hand side of \eqref{full-cond-lam} is proportional to the normal density but with the constraint $\bm{\lambda}_i>\bm{0}$, that is, $\lambda_{i,j}>0$ for all $j$. Following \cite{tang2023global}, we employ constraint relaxation by approximating the indicator function $\prod_{j=1}^k 1_{(0,\infty)}(\lambda_{i,j})$ by a sigmoid function $\prod_{j=1}^k 1/(1+\exp(-\eta \lambda_{i,j}))$ for a large constant $\eta>0$ \citep[see also,][]{ray2020efficient, onizuka2024bayesian}. Then we use the elliptical slice sampler \citep{murray2010elliptical, hahn2019efficient} to draw samples for $\bm{\lambda}_{i}$. The details of the sampling method are provided in Appendix \ref{sec:appB2}. Since the full conditional distributions of $\rho$ and of $a$ (used in the normal-gamma prior) do not correspond to any standard probability distributions, we employ grid-based sampling for each of them in this study. Furthermore, to reduce the computational burden in sampling $\rho$, we define
$\mathbf{S} = \mathbf{D}^{-1/2}\mathbf{W}\mathbf{D}^{-1/2}$ and use the identity
\[
|\mathbf{D}-\rho \mathbf{W}|
= |\mathbf{D}|\,|\mathbf{I}_m-\rho \mathbf{S}|
= |\mathbf{D}| \prod_{i=1}^m (1-\rho \gamma_i),
\]
where $\gamma_i$ denotes the eigenvalues of $\mathbf{S}$.
Note that these eigenvalues need to be computed only once.

To evaluate the computational efficiency and scalability of the proposed algorithm, we briefly analyze its computational complexity per MCMC iteration. Let $m$ be the number of areas and $k$ be the dimension of the response variables. The overall computational complexity per iteration is bounded by $\mathcal{O}(mk^3 + km^2)$, where the term $\mathcal{O}(km^2)$ is driven by spatial weight matrix operations for large $m$, and $\mathcal{O}(mk^3)$ stems from area-specific multivariate updates of dimension $k \times k$. To mitigate bottlenecks associated with $m$, our implementation efficiently utilizes sparse matrix operations for the spatial weight matrix $\mathbf{W}$.

Since this paper does not consider ultra-high-dimensional settings, MCMC is sufficiently computationally feasible. Indeed, MCMC performs well in our simulation studies and real data analysis, and we report the effective sample sizes (ESS) and computational times in the subsequent real data analysis. On the other hand, for extremely large $m$, alternative approaches to MCMC could be considered, which are discussed in Section 6. However, it should be noted that approximation-based methods often suffer from the drawback of underestimating posterior uncertainty compared to MCMC.


\section{Simulation}
\label{sec:4}

We conduct simulation studies to illustrate the performance of the proposed methods.

\subsection{Setting}
\label{sec:4.1}

We investigate the performance of different priors for estimating the small area means on simulated datasets. We generate data from the model:
\[\bm{y}_i =\mathbf{X}_i \bm{\beta} + \bm{u}_i+ \bm{\varepsilon}_i,\quad  \bm{\varepsilon}_i\sim \mathcal{N}_k (\bm{0}, \mathbf{V}_i), \quad i=1,\dots,m,\]
where $\bm{\beta}=(1,1,1,1)^{\top}$, $\bm{x}_i=(1, x_{i,1})^{\top}$, where $x_{i,1}$ is generated from the standard normal distribution. We consider the bivariate case ($k=2$) and $m=100$, $200$ and $500$. Concerning the dispersion matrices $\mathbf{V}_i$ of sampling errors $\bm{\varepsilon}$, we consider two cases: (a) $0.7\mathbf{I}_k$, $0.6\mathbf{I}_k$, $0.5\mathbf{I}_k$, $0.4\mathbf{I}_k$, $0.3\mathbf{I}_k$; (b) $2.0\mathbf{I}_k$, $0.6\mathbf{I}_k$, $0.5\mathbf{I}_k$, $0.4\mathbf{I}_k$, $0.2\mathbf{I}_k$
\citep[see also][]{datta2005measuring, ito2020robust}. These error covariances are randomly assigned to each area in equal proportions. In the observation-error covariance matrix, the correlations are set to zero, although a more general treatment is possible. In practice, however, direct estimates of correlation coefficients between observations are seldom available, as in the real-world dataset analyzed in Section~\ref{sec:5}. For this reason, and for simplicity, we assume zero correlation.

To model spatial dependence, we constructed an adjacency matrix $\mathbf{W}$ based on a regular lattice structure. Specifically, we considered a rectangular $10\times (m/10)$ grid where each cell corresponds to a spatial unit. Two units are considered adjacent (i.e., $w_{ij}=1$) if they share a common edge (i.e., they are first-order neighbors in the cardinal directions: up, down, left, or right). Diagonal adjacency is not considered. The resulting matrix $\mathbf{W}$ is symmetric and binary, with $w_{ij}=1$ indicating adjacency and $w_{ij}=0$ otherwise.

For generating a true random effect vector $\bm{u} \in \mathbb{R}^{mk}$, we consider the following cases. For all scenarios, we set $\rho=0.95$ and $\mathbf{\Sigma}=\begin{pmatrix}
1 & 0.3 \\
0.3 & 1
\end{pmatrix}$.

\begin{itemize}

\item {\bf (Scenario 1)}: 
Let $\bm{u} \sim \mathcal{N}_{mk}\!\left(\bm{0},\, (\mathbf{D}-\rho\mathbf{W})^{-1} \otimes \mathbf{\Sigma}\right)$.

\item {\bf (Scenario 2)}: 
Let $\bm{u} \sim \mathcal{N}_{mk}\!\left(\bm{\mu},\, (\mathbf{D}-\rho\mathbf{W})^{-1} \otimes \mathbf{\Sigma}\right)$,
where each component of $\bm{\mu}$ is generated from $(1 - z)\delta_{5} + z\delta_{0}$,
with $z \sim \mathrm{Ber}(\omega)$ and $\delta_x$ denoting a point mass at $x$.
We set $\omega = 0.5$.

\item {\bf (Scenario 3)}: 
The same as Scenario 2, except that we set $\omega = 0.8$.

\item {\bf (Scenario 4)}: 
Let $\bm{u} \sim \mathcal{N}_{mk}\!\left(\bm{0},\, 4\left((\mathbf{D}-\rho\mathbf{W})^{-1} \otimes \mathbf{\Sigma}\right)\right)$.
For each variable (i.e., for the vectors $(u_{1,1}, \ldots, u_{m,1})$ and $(u_{1,2}, \ldots, u_{m,2})$),
the lower 50\% of the elements in terms of absolute magnitude are replaced with zero.

\item {\bf (Scenario 5)}: 
Let $\bm{u} \sim \mathcal{N}_{mk}\!\left(\bm{0},\, 4\left((\mathbf{D}-\rho\mathbf{W})^{-1} \otimes \mathbf{\Sigma}\right)\right)$.
For each variable, the lower 80\% of the elements in terms of absolute magnitude are replaced with zero.

\end{itemize}
Scenario~1 assumes a Gaussian distribution for the random effects, which is the standard specification commonly used in spatial Fay--Herriot (FH) models. 
In Scenarios~2 and 3, the random effects follow a two-component mixture distribution that represents a combination of signal and noise: 
values equal to $0$ correspond to noise, while those equal to $\mu = 5$ represent nonzero signal.  
The choice $\mu = 5$ ensures that the signal is sufficiently distinguishable from noise, as smaller values may lead to weak identifiability.
The sparsity level of the non-signal components is controlled by $\omega$. Specifically, $\omega=0.5$ (Scenario 2) represents a setting in which signal and noise are present in equal proportions, whereas $\omega=0.8$ (Scenario 3) corresponds to a setting in which the proportion of signal is smaller.
These scenarios are commonly used when evaluating the performance of shrinkage priors. 
Scenarios~4 and 5 introduce an alternative form of sparsity that is generated directly from Gaussian random effects \citep[see also][]{tang2023global}.  
After drawing $\bm{u}$ from a spatially correlated normal distribution with inflated variance, we enforce sparsity by thresholding: 
for each variable, the lower 50\% (Scenario~4) or 80\% (Scenario~5) of the components in terms of absolute magnitude are set to zero.  
These scenarios differ from Scenarios~2 and 3 in that sparsity arises from magnitude-based thresholding rather than a discrete mixture distribution, and in that some random effects are replaced with exactly zero, thereby allowing us to assess the robustness of the methods against different forms of sparsity.

We compare the following five methods: (I) The proposed model under the horseshoe prior \eqref{proposal} (denoted by {\bf SpaHS}); (II) The proposed model under the normal-gamma prior \eqref{proposal} (denoted by {\bf SpaGa}) (III) The multivariate spatial FH model \eqref{model-spatial} (denoted by {\bf SpaFH}); (IV) the multivariate global-local model under the horseshoe prior \eqref{Ghosh-prior} proposed by \cite{ghosh2022multivariate} (denoted by {\bf HS}); (V) the mutivariate Fay-Herriot model \eqref{FHmodel} (denoted by {\bf FH}). Note that we assume $a \sim \mathrm{DiscreteUniform}(0.01, 0.02,\dots, 0.99, 1.0)$ and $b\sim \mathrm{Ga}(0.001, 0.001)$ for the SpaGa method.
For each of these models, we run MCMC for 2,000 iterations. The first 500 iterations are discarded as the burn-in period. The small area means $\bm{\theta}_i$ are estimated by using the posterior means. 

We evaluated the performance of estimation for $\bm{\theta}$ using average absolute
deviation (AAD) and average squared deviation (ASD), which are defined as follows:
\begin{align*}
\mathrm{AAD}= \frac{1}{mk} \sum_{i=1}^m \sum_{j=1}^k |\hat{\theta}_{ij}-\theta_{ij}|, \quad \mathrm{ASD}=\frac{1}{mk} \sum_{i=1}^m \sum_{j=1}^k (\hat{\theta}_{ij}-\theta_{ij})^2.
\end{align*}
The empirical coverage rate (CP) and average length (AL) of 95\% credible intervals of $\theta_{ij}$ are also
used to evaluate the interval estimates.
For a given situation, we
report the median of the four measures over the 50 datasets.

\subsection{Simulation results}
\label{sec:4.2}

First, we examine the behavior of the local parameters $\lambda$ in the proposed method SpaHS and the existing method HS using a single simulation run. Here, we consider Scenario~5 with $m=500$ and observation variance corresponding to case~(a). Figure~\ref{fig:scenario5-random-effect} displays the distribution of the true random effect vectors and the posterior means of the estimated local parameters under the two methods. Note that for SpaHS, the local parameters are given by
$\bm{\lambda}^{\mathrm{SpaHS}}
=(\lambda_{1,1}, \lambda_{1,2}, \dots, \lambda_{500,1}, \lambda_{500,2})^{\top}$,
whereas for HS they are given by
$\bm{\lambda}^{\mathrm{HS}}
=(\lambda_1, \dots, \lambda_{500})^{\top}$.
The left panel shows the true two-dimensional random effects $\bm{u}_i = (u_{i,1}, u_{i,2})^{\top}$ for each area, where the color indicates the number of nonzero components in $\bm{u}_i$. Specifically, points at the center correspond to $\bm{u}_i = \mathbf{0}$ (zero nonzero components), points lying on one of the axes correspond to cases where only one component of $\bm{u}_i$ is nonzero (one nonzero component), and points off the axes correspond to cases where both components are nonzero (two nonzero components). From this panel, we observe substantial heterogeneity across areas in the magnitude of the random effects at the component level.
The middle and right panels plot the posterior means of the local parameters for HS and SpaHS, respectively, using the same color coding as in the left panel. Since HS performs shrinkage jointly at the area level, the estimated local parameters tend to be uniformly small (red) or large (blue) within each area, while areas with only one nonzero component in $\bm{u}_i$ (green) exhibit intermediate shrinkage. In contrast, SpaHS allows component-wise shrinkage: red points correspond to small values in both dimensions, blue points to large values in both dimensions, and green points are aligned along the axes, indicating that shrinkage is applied selectively to each component.
These results demonstrate that, in such scenarios, the proposed SpaHS method is capable of more flexible and adaptive shrinkage compared with HS.

\begin{figure}[tbp]
    \centering
    \includegraphics[width=\linewidth]{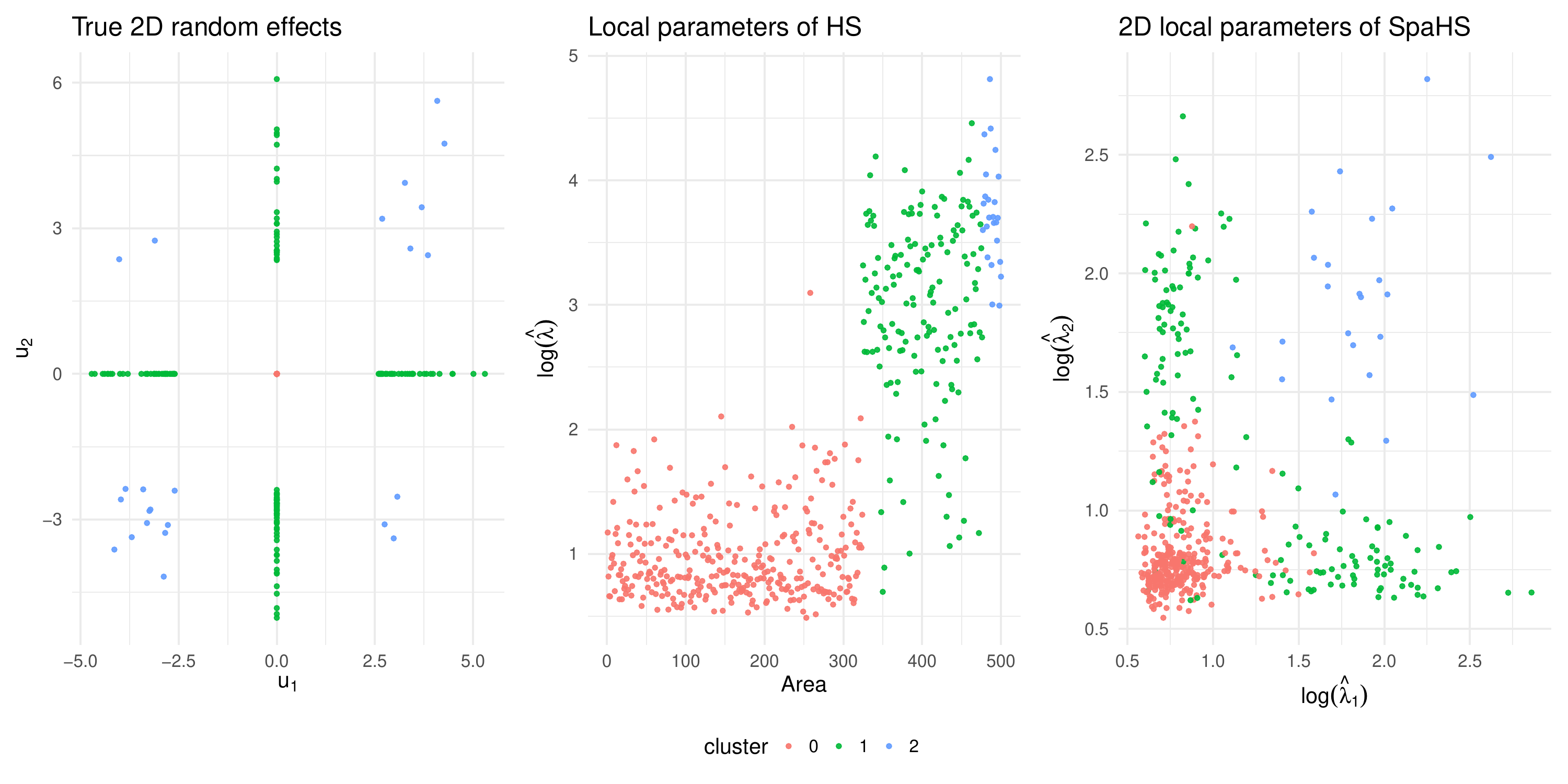}
    \caption{Results from a single simulation run for true random effects (left) and posterior mean estimates of the local parameters under HS (center) and SpaHS (right) for Scenario~5 with $m=500$ and observation variance case~(a). For HS, each area has a single local parameter, and thus the $x$-axis represents the area index. In contrast, SpaHS assigns two local parameters to each area, which are plotted in two dimensions.
}
    \label{fig:scenario5-random-effect}
\end{figure}

Next, we report the simulation results in Figures~\ref{fig:sim1} and \ref{fig:sim2}. Figure~\ref{fig:sim1} corresponds to variance scenario (a), whereas Figure~\ref{fig:sim2} pertains to scenario (b), where the area-specific variances exhibit greater variability.
We first focus on the point estimation results shown in the top two rows of each figure.
In Scenario 1, the SpaFH model performs well across all configurations, which is expected since the data-generating mechanism follows the SpaFH specification. SpaGa achieves the second-best performance in this scenario. This is reasonable since the current version of SpaGa, although less capable of extreme shrinkage than the horseshoe prior, can adaptively control the amount of shrinkage through the parameter $a$, which is estimated from the data via MCMC.
Under Scenario 2, SpaHS and SpaGa perform similarly. In particular, when the sampling error covariance exhibits high area-to-area variability (scenario (b)), SpaHS and SpaGa yield almost identical ASD values, and even the standard FH model without spatial effects performs reasonably well.
In Scenario 3, SpaHS provides the best performance due to the horseshoe prior, which is well suited for sparse structures. In contrast, SpaGa performs slightly worse under these conditions, and the HS model proposed by \cite{ghosh2022multivariate}, which incorporates neither spatial dependence nor element-wise shrinkage, performs comparably to SpaGa.

In the thresholded sparsity scenarios (Scenarios 4 and 5), SpaHS again performs consistently well. When the proportion of nonzero elements is 50\% (Scenario 4), SpaFH also yields good results. Notably, SpaGa -- which performs well in Scenarios 2 and 3 -- shows a clear deterioration in Scenarios 4 and 5. This decline is likely due to the nature of the normal-gamma prior, which induces partial shrinkage toward zero and is therefore less suited to thresholded sparsity structures.

The bottom two rows of Figures~\ref{fig:sim1} and \ref{fig:sim2} summarize the uncertainty quantification results. Although differences across methods are generally small, SpaHS yields shorter credible intervals and coverage probabilities closer to the nominal 95\% level in many scenarios. In Scenarios 4 and 5, where SpaGa performs poorly in terms of point estimation, its credible intervals also noticeably underperform relative to the other methods.

Overall, SpaHS performs particularly well when the underlying signal is sparse. SpaGa also offers stable estimates in Scenarios 1 to 3. From a practical standpoint, one may consider applying both SpaHS and SpaGa models and selecting between them using an appropriate model comparison criterion. We experimented with three different numbers of areas, $m = 100$, $200$, and $500$, and the results were largely consistent across these settings. Because the proposed methods involve $mk$ local parameters, repeated experiments for substantially larger m were not feasible due to computational constraints. Investigating the performance of these methods for larger $m$ remains an interesting direction for future work.

The results of the numerical experiments for cases with weak spatial correlation ($\rho = 0.30$) or a larger dimension of direct estimators ($k = 4$) are provided in the supplementary material. In either case, the results remain largely unchanged from those presented in this section.

\begin{figure}[htbp]
    \centering
    \includegraphics[width=\linewidth]{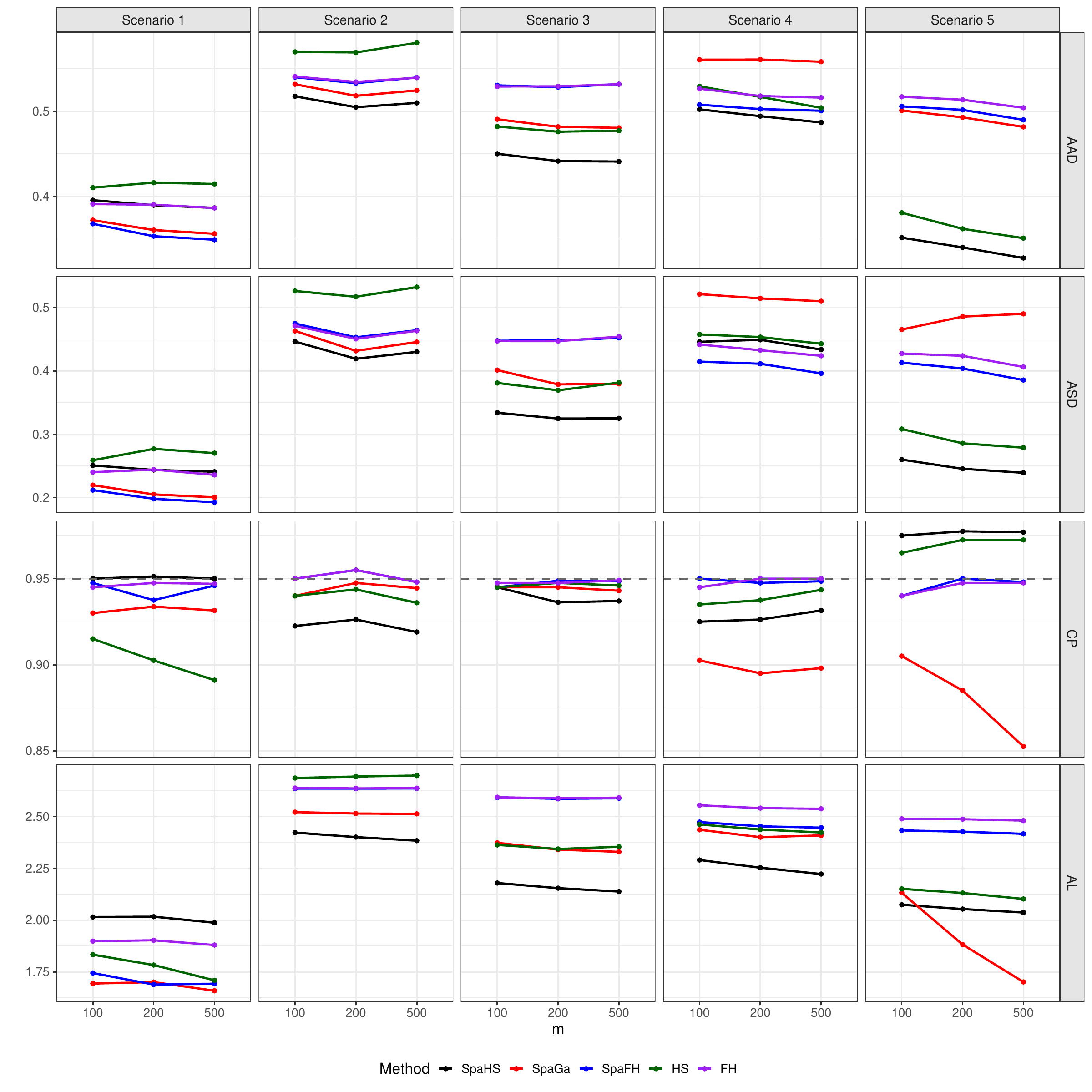}
    \caption{The average absolute
deviation (AAD), average squared deviation (ASD), coverage probability of the 95\% credible interval and average length of the 95\% credible interval under the variance scenario (a). The figure displays the median computed over 50 replications and the horizontal axis is the number of small areas.}
    \label{fig:sim1}
\end{figure}

\begin{figure}[htbp]
    \centering
    \includegraphics[width=\linewidth]{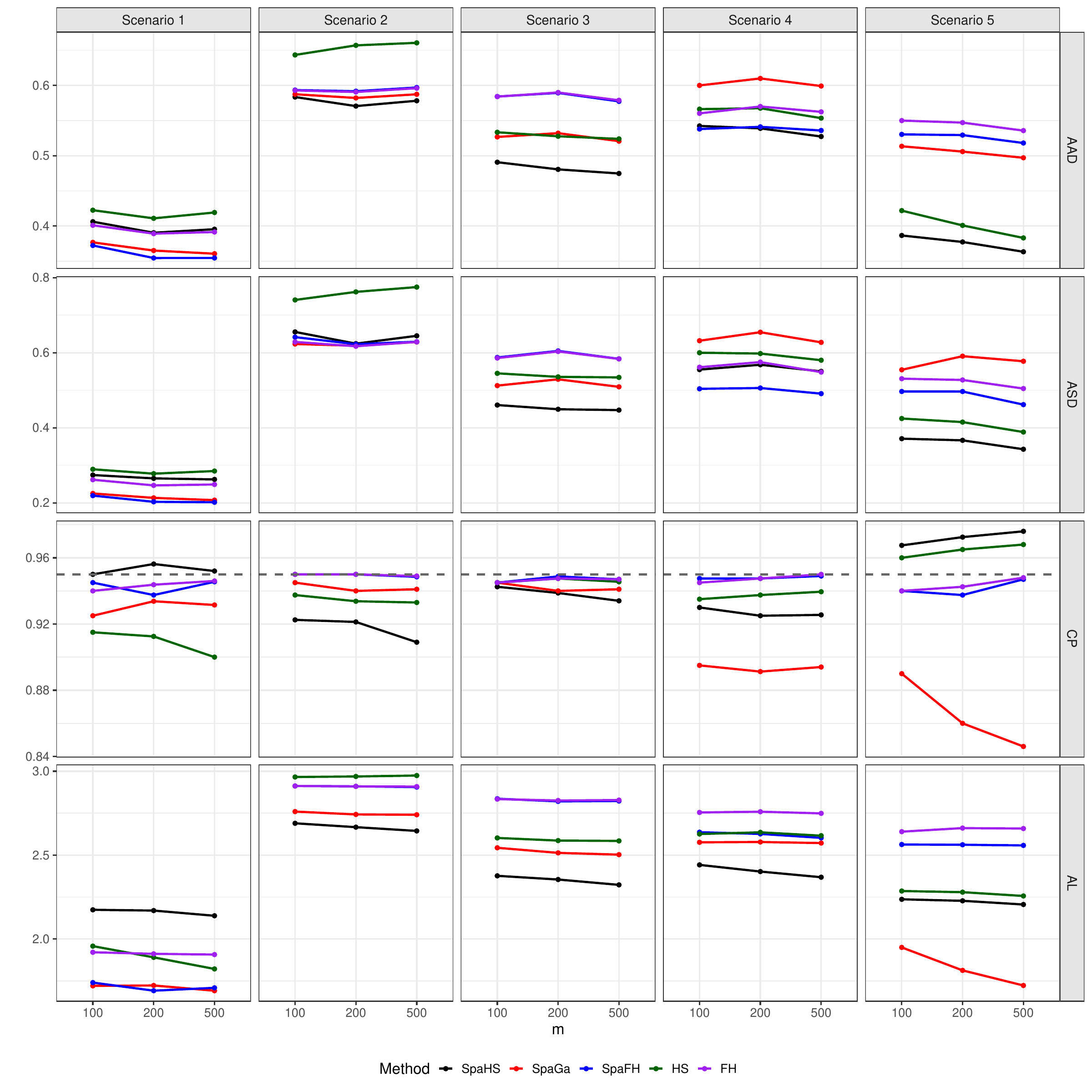}
    \caption{The average absolute
deviation (AAD), average squared deviation (ASD), coverage probability of the 95\% credible interval and average length of the 95\% credible interval under the variance scenario (b). The figure displays the median computed over 50 replications and the horizontal axis is the number of small areas.}
    \label{fig:sim2}
\end{figure}

\section{Real data analysis}
\label{sec:5}

In the context of small area estimation, we analyze county-level data in the West Census Region using two key socioeconomic indicators: median household income and poverty rate. Many counties have unreliable direct estimates due to small sample sizes; thus, jointly modeling these two correlated variables allows us to borrow strength across indicators and improve estimation accuracy.
Moreover, by incorporating spatial information, the model can leverage correlations among neighboring counties, capturing spatial dependencies and identifying areas exhibiting extreme values (hotspots) in income or poverty. This approach provides more reliable and nuanced estimates, which can inform policy decisions and targeted interventions in the region. As a key feature of the proposed method is its ability to accurately identify counties exhibiting extreme values in either income or poverty rate, while simultaneously improving the overall estimation precision across the region.

\subsection{Dataset and setting}
\label{sec:5.1}

We use data from the 5-year American Community Survey (ACS) released in 2022, where the response variables are the estimated median household income and poverty rate, and the percentage of high school graduates is employed as an auxiliary covariate. We focus on the 413 counties located in the West Census Region, excluding those that are geographic enclaves (non-contiguous counties). 
Tang et al. (2023) also analyzed poverty rates in the West Census Region. 
To better satisfy the normality assumption, the median household income and poverty rate were log-transformed, and the corresponding variances were computed using the Delta method. The variances of the log-transformed direct estimates are shown in Figure~\ref{fig:variance-obs}. 
The Moran's I statistics for the log-transformed direct estimates of the median household income and the poverty rate were $0.50$ and $0.35$, respectively (both with $\text{p-value} < 0.001$). These results indicate that both variables exhibit statistically significant positive spatial autocorrelation, suggesting the presence of spatial clustering rather than spatial randomness.
In our study, we handle county-level data, and the sample sizes within each county are sufficiently large. Consequently, the variances of the direct estimates for most areas are relatively moderate. However, we observe that a few regions exhibit comparatively large variances. 
For example, with respect to median income, De Baca County, New Mexico and Esmeralda County, Nevada exhibit relatively high values of the sample variance. In terms of the poverty rate, Eureka County, Nevada and Camas County, Idaho likewise show notably elevated levels.
Such counties tend to exhibit high levels of variance in socioeconomic indicators. 
This can be largely attributed to factors such as extremely small population sizes, the dominance of a limited number of industries, and the susceptibility of local statistics to substantial fluctuations resulting from minor demographic or economic changes. Additionally, survey-based estimates for such counties often involve relatively large sampling errors, further contributing to the observed variance.
Since the ACS data do not provide covariance information between variables, we assume that the covariances among the direct estimates are zero. 

The proposed SpaHS and SpaGa models, as well as the SpaFH model for comparison, are fitted to the data. The prior distributions for SpaHS, SpaGa, and SpaFH were set to be the same as those used in Section 4. For the HS and FH models, we employed the prior distributions specified by \cite{ghosh2022multivariate}. For each model, 20,000 MCMC iterations were performed, discarding the first 10,000 samples as burn-in. The adjacency matrix $\mathbf{W}$ is constructed based on whether counties share a common border. 

Model performance is evaluated using the Deviance Information Criterion (DIC). The DIC proposed by \cite{spiegelhalter2002bayesian} is defined by
$\mathrm{DIC}=D(\bar{\bm{\theta}})+ 2p_D$,
where $p_D=\bar{D}-D(\bar{\bm{\theta}})$ measures the model complexity, $D(\bar{\bm{\theta}})$ is the deviance of the model evaluated at the posterior mean of model parameter $\bm{\theta}$, and $\bar{D}$ is the posterior mean of the deviance. In our models, 
$D(\bm{\theta}) = \sum_{i=1}^m (\bm{y}_i-\bm{\theta}_i)^{\top} \mathbf{V}_i^{-1}(\bm{y}_i-\bm{\theta}_i)$.
Using a posterior sample $\bm{\theta}^{(s)}=(\bm{\theta}_1^{(s)}, \dots, \bm{\theta}_m^{(s)})$ for $s=1,\dots, S$, we obtain
\begin{align*}
\widehat{\mathrm{DIC}}
= \frac{2}{S}\sum_{s=1}^S \sum_{i=1}^m (\bm{y}_i-\bm{\theta}_i^{(s)})^{\top} \mathbf{V}_i^{-1}(\bm{y}_i-\bm{\theta}_i^{(s)}) - \sum_{i=1}^m (\bm{y}_i-\hat{\bm{\theta}}_i)^{\top} \mathbf{V}_i^{-1}(\bm{y}_i-\hat{\bm{\theta}}_i),
\end{align*}
where $\hat{\bm{\theta}}_i=(1/S) \sum_{s=1}^S \bm{\theta}_i^{(s)}$ is the posterior sample mean of $\bm{\theta}_i$.
Although the Deviance Information Criterion (DIC) has been widely used for model comparison in the existing literature, it lacks a clear decision-theoretic justification, and its application to hierarchical models has also faced criticism. Therefore, in addition to DIC, we also conduct model comparison using the posterior predictive loss (PPL) \citep{gelfand1998model}. 
Specifically, let $\bm{y}_{\text{new}}=(\bm{y}_{\text{new},1}^{\top},\dots, \bm{y}_{\text{new},m}^{\top})^{\top}$ be a replicate of the observed data $\bm{y}$ generated from the posterior predictive distribution. Following \citet{gelfand1998model}, the general PPL criterion is defined as a sum of a goodness-of-fit term and a penalty term for model complexity:
\begin{align*}
D_w = \frac{w}{w+1} G + P,
\end{align*}
where $G=\sum_{i=1}^{m} \|\mathrm{E}[\bm{y}_{\text{new},i} \mid \bm{y}] - \bm{y}_i\|^2$, $P=\sum_{i=1}^{m} \mathrm{tr}\left(\mathrm{Var}(\bm{y}_{\text{new},i} \mid \bm{y})\right)$, 
$\bm{y}_i$ is the $i$-th observed data point, and $w > 0$ is a weight parameter balancing the trade-off between bias and variance. Following standard practice in spatial and hierarchical modeling, we adopt the unweighted ($w \to \infty$) sum of the goodness-of-fit and variance penalty terms. For 
each model, we supply the values of the goodness-of-fit term ($G$), the penalty term ($P$) and the sum $D=D_{\infty}$ \citep[see also][]{gelfand2003proper, banerjee2025hierarchical}.
Similar to DIC, smaller values of PPL indicate a better predictive model, and it is rigorously derived within a decision-theoretic framework. For further details regarding PPL, see \citet{gelfand1998model}.

\begin{figure}[htbp]
    \centering
    \includegraphics[width=\linewidth]{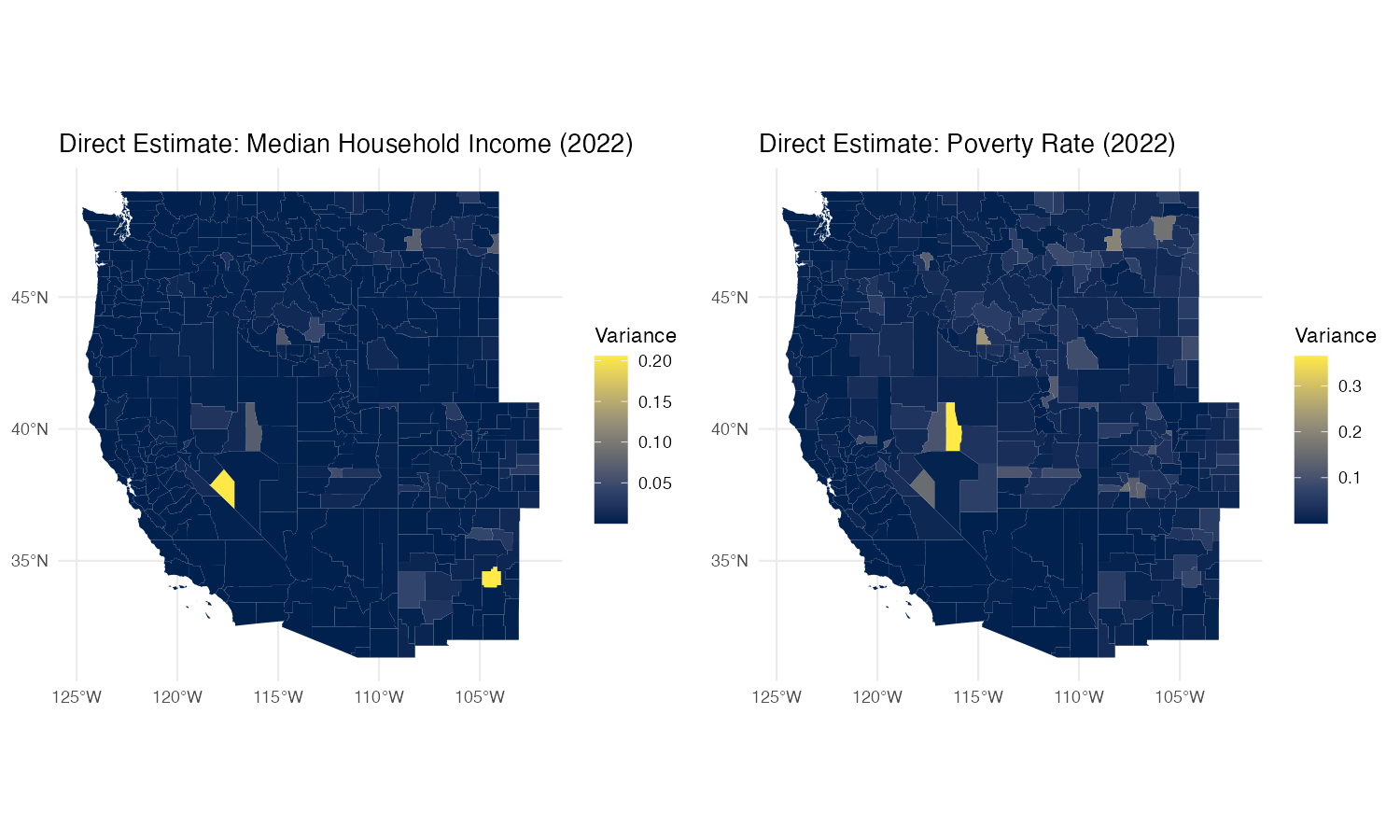}
    \caption{Variance on the log-scale of median household income and poverty rate of direct estimates for 413 counties in West Census Region, based on 5-year period estimates from the 2022 American Community Survey.}
    \label{fig:variance-obs}
\end{figure}

\subsection{Results}
\label{sec:5.2}

Figure~\ref{fig:theta-plot} presents the SpaHS estimates of income and poverty rates (posterior means of $\bm{\theta}$) alongside the corresponding estimates from SpaFH and SpaGa. Relative to SpaFH, the largest discrepancies appear in areas with lower outcome values, where SpaFH tends to slightly overestimate. In contrast, the SpaHS and SpaGa estimates of income are generally similar, while noticeable differences emerge for the poverty rate due to different levels of shrinkage applied by the models.

Figure~\ref{fig:analysis-theta-spatial} shows spatial maps of the SpaHS estimates of $\bm{\theta}$ together with the corresponding direct estimates $\bm{y}_i$. Visual differences in the income estimates are modest. For the poverty rate, however, the largest deviations between the SpaHS small-area means and the direct estimates occur in Petroleum County and Toole County, Montana, located in the northeastern portion of Figure~\ref{fig:analysis-theta-spatial}. This pattern is consistent with the estimated random effects (posterior means of $\bm{u}$) shown in Figure~\ref{fig:analysis-rf-spatial}, which suggest that the random effects exhibit a certain degree of sparsity. Several regions display high or low values (hot or cold spots). The variance of the random effects is larger for the poverty rate than for median income, and the largest posterior mean of SpaHS for the poverty rate is found in Morgan County, Utah, followed by Madison County, Idaho. Both counties show negative random effects and may be considered cold spots. 

\begin{figure}[htbp]
    \centering
    \includegraphics[width=0.8\linewidth]{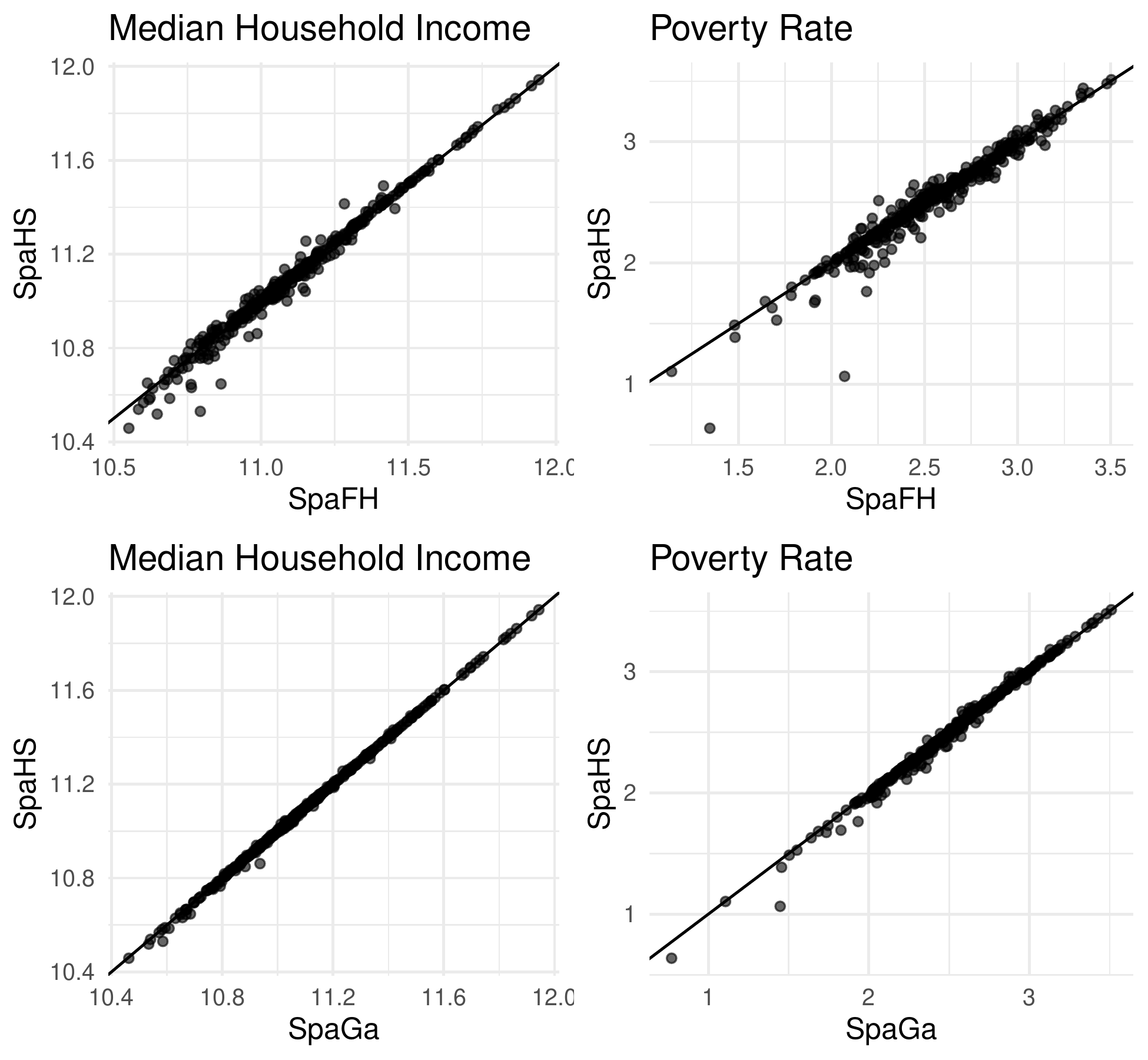}
    \caption{Posterior means of small area mean $\bm{\theta}_i$ for the SpaHS, SpaGa and SpaFH methods.}
    \label{fig:theta-plot}
\end{figure}

\begin{figure}[htbp]
    \centering
   \includegraphics[width=\linewidth]{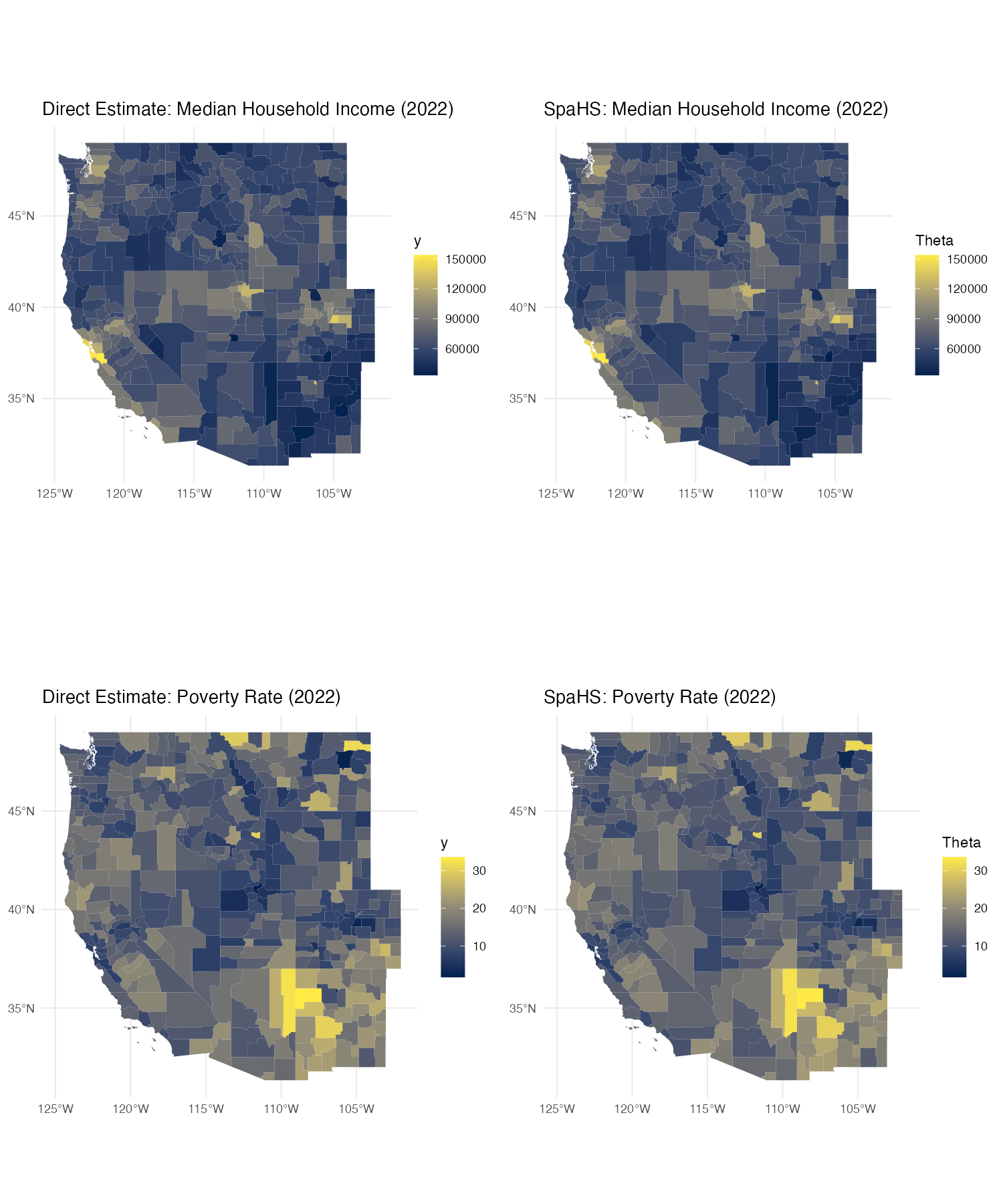}
    \caption{Estimates of median household income and poverty rate for 413 counties in West Census Region, based on 5-year period estimates from the 2022 American Community Survey, including direct estimates and results from the SpaHS, SpaGa and SpaFH methods.}
    \label{fig:analysis-theta-spatial}
\end{figure}

\begin{figure}[htbp]
    \centering
   \includegraphics[width=\linewidth]{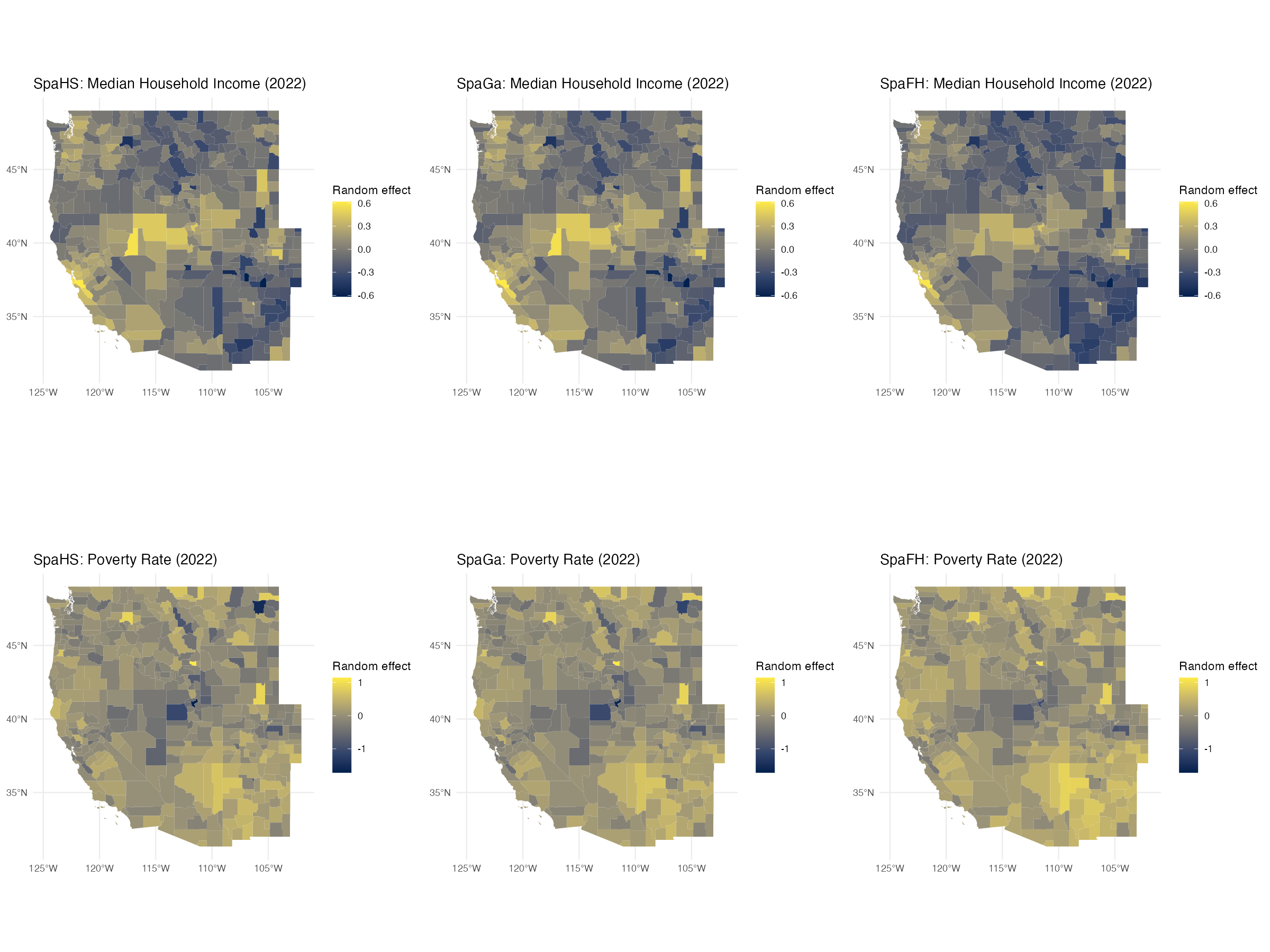}
    \caption{Spatial random effects on the log-scale of median household income and poverty rate for 413 counties in West Census Region, based on 5-year period estimates from the 2022 American Community Survey, including results from the SpaHS, SpaGa and SpaFH methods.}
    \label{fig:analysis-rf-spatial}
\end{figure}

Although these visualizations offer useful insights into the behavior of the models, they do not directly determine which model performs best. Model evaluation based on the DIC values in Table~\ref{tab:DIC} shows that SpaHS provides the best overall fit, followed by SpaGa. Regarding PPL in Table~\ref{tab:PPL}, SpaHS achieves the smallest value followed by SpaGa, just as observed with DIC. However, unlike DIC, HS takes the third smallest value. The HS method of \citep{ghosh2022multivariate} performs less favorably for these data because it does not incorporate spatial dependence and cannot perform element-wise shrinkage. Table~\ref{tab:DIC} also reports the mean effective sample size across the components of $\bm{\theta}$. Because this value is close to the total number of posterior samples (10,000), on average, the MCMC sampling is considered efficient.
Regarding computation time, the proposed methods are slower than the HS and FH approaches, mainly due to the increased number of latent variables. In SpaHS and SpaGa, the number of local parameters $\lambda_{i,j}$ is twice that of HS, and additional matrix operations required for modeling spatial dependence further increase computational cost. Several strategies could mitigate this computational burden, and these will be discussed in the Discussion section. However, addressing them in depth is beyond the scope of the present paper. All computations were performed on a laptop equipped with an Apple M1 Max processor and 32 GB of unified memory, running macOS 15.7.2 (build 24G325). Additional information on the real data analysis is provided in Appendix \ref{sec:appC}.

\begin{table}[ht]
\caption{The Deviance Information Criterion (DIC) for various models, mean effective sample size (ESS) across components of small area mean $\bm{\theta}$ and the CPU time, where the MCMC algorithm was run for 20,000 iterations, discarding the first 10,000 as burn-in.}
\begin{center}
\begin{tabular}{c|rrrrr}
  \hline
 & SpaHS & SpaGa & SpaFH & HS & FH \\ 
  \hline
  DIC & 1625.41 & 1654.36 & 1685.29 & 1737.33 & 1650.43 \\ 
  ESS & 8254.51 & 7620.79 & 9550.68 & 6953.43 & 9973.85 \\ 
  CPU time (min) & 23.71 & 22.85 & 14.64 & 7.49 & 4.59 \\ 
   \hline
\end{tabular}
\end{center}
\label{tab:DIC}
\end{table}

\begin{table}[ht]
\caption{Model comparison via posterior predictive loss (PPL). For each model, the values of the goodness-of-fit term ($G$), the penalty term ($P$), and their sum ($D = G + P$) are presented. Smaller values of $D$ indicate better predictive performance.}
\begin{center}
\begin{tabular}{rrrr}
  \hline
Model & $G$ & $P$ & $D=G+P$ \\ 
  \hline
SpaHS & 2.17 & 21.09 & 23.26 \\ 
  SpaGa & 4.47 & 19.60 & 24.07 \\ 
  SpaFH & 7.97 & 17.71 & 25.68 \\ 
  HS & 5.48 & 19.38 & 24.86 \\ 
  FH & 6.42 & 18.87 & 25.30 \\ 
   \hline
\end{tabular}
\end{center}
\label{tab:PPL}
\end{table}

\section{Discussion}
\label{sec:6}

The proposed model becomes computationally challenging when applied to large-scale data. For instance, as the number of small areas grows, the size of the matrices to be inverted increases accordingly, leading to a higher computational burden. To address this issue, approximate inference methods such as variational Bayes \citep{blei2017variational} or variational autoencoders \citep{kingma2013auto} could be considered as alternatives to MCMC. Indeed, such approaches have already been explored in the contexts of extreme value theory \citep{zhang2023flexible} and small area estimation \citep{wang2025variational}. Incorporating such approximate inference methods into the proposed framework remains an important direction for future research. The proposed prior can potentially be extended to spatially varying sparse coefficient models, where it may induce area-specific shrinkage of regression coefficients \citep{boehm2015spatial}.

Furthermore, regarding model complexity and parsimony, one might question whether the added flexibility of global-local shrinkage priors is justified compared to the standard Fay-Herriot (FH) model. In practice, the true underlying structure of random effects such as sparsity or departure from normality is unknown a priori. While the traditional FH model performs well under strict Gaussian assumptions, it lacks adaptivity and risks structural misspecification, particularly in settings with a large number of areas $m$. The proposed framework serves as a robust, data-driven generalization designed to automatically accommodate sparsity, heavy tails, or localized anomalies without requiring prior knowledge of the data-generating process.

In relation to this, a subtle consideration is the interplay between aggressive component-wise local shrinkage and spatial dependence (borrowing of strength). Rather than reducing valid spatial smoothing, the local scale parameters act as an adaptive safeguard: for typical regions or noise, heavy shrinkage prevents local anomalies from contaminating neighbors, whereas for areas with genuine signals, the model retains sufficient flexibility to borrow strength across space. Thus, the proposed framework achieves a sophisticated balance between spatial smoothing and local protection.

The proposed methods can also be extended to non-Gaussian responses such as binomial and Poisson responses. The binomial distribution is frequently employed to model proportions such as disease prevalence or poverty rates, while the Poisson distribution is commonly used to model count data in a given area. The latter is particularly useful in applied contexts and is often referred to as disease mapping \citep[e.g.,][]{macnab2022bayesian}. Although we do not explore the details in this paper, it is worth noting that the computational complexity increases when moving from the normal distribution to the Poisson case. In fact, it is well known that posterior inference for spatial regression models based on the Poisson likelihood is computationally demanding. On the other hand, by approximating the Poisson model with a negative binomial distribution, it becomes feasible to leverage P\'olya-Gamma data augmentation techniques \citep{polson2013bayesian}, in a manner analogous to the binomial case.

An interesting direction for future research would be to extend the current framework to allow for spatial correlation parameters that vary across regions, as proposed in the generalized multivariate CAR models by \citep{jin2005generalized, jin2007order}. Incorporating such flexibility could potentially capture local spatial heterogeneity more effectively. Extending the model to a spatio-temporal setting also presents another promising avenue, enabling the analysis of dynamic changes over both space and time. However, since computations based on MCMC are expected to be considerably demanding, some methodological or computational refinements will likely be required.

\section*{Acknowledgments}

The authors thank the Editor and an anonymous reviewer for their constructive comments and valuable suggestions. In particular, we would like to express our sincere gratitude to the reviewer for providing insightful comments regarding the presentation and the variations in the numerical experiments, which have significantly improved the quality and clarity of the manuscript.
The authors' research was supported in part by JSPS KAKENHI Grant Numbers 25K07131 and 25K23108 from the Japan Society for the Promotion of Science.
The authors also thank Dr. Xueying Tang of University of Arizona for sharing us R code of the paper \cite{ghosh2022multivariate}.

\def\thesection{Appendix}
\def\thesubsection{A.\arabic{subsection}}
\section{}
\setcounter{equation}{0}
\def\theequation{A.\arabic{equation}}
\def\thelem{A.\arabic{lem}}
\allowdisplaybreaks[1]

\subsection{Proof}
\label{sec:appA}

\begin{proof}[Proof of Proposition \ref{propriety}]
The joint posterior density is
\begin{align*}
&\pi(\tilde{\bm{u}}, \bm{\beta}, \mathbf{\Sigma}, \mathbf{\Lambda}, \rho, \tau^2 \mid \bm{y})\\
&\propto \exp\left[-\frac{1}{2}(\bm{y}-\mathbf{X}\bm{\beta}-\mathbf{\Lambda}\tilde{\bm{u}})^{\top} \mathbf{V}^{-1}(\bm{y}-\mathbf{X}\bm{\beta}-\mathbf{\Lambda}\tilde{\bm{u}})\right]\\
&\quad \times |\tau^2((\mathbf{D}-\rho \mathbf{W})^{-1} \otimes \mathbf{\Sigma})|^{-1/2}  \exp\left[-\frac{1}{2}\tilde{\bm{u}}^{\top} \frac{1}{\tau^2}((\mathbf{D}-\rho \mathbf{W}) \otimes \mathbf{\Sigma}^{-1})\tilde{\bm{u}}\right]\\
&\quad \times \pi(\mathbf{\Sigma})\left(\prod_{i=1}^m \prod_{j=1}^k \pi(\lambda_{i,j}) \right) \pi(\rho) \pi(\tau^2).
\end{align*}
Let $\bm{z}=\bm{y} - \mathbf{\Lambda}\tilde{\bm{u}}$. Then
\begin{align*}
(\bm{y}-\mathbf{X}\bm{\beta}-\mathbf{\Lambda}\tilde{\bm{u}})^{\top} \mathbf{V}^{-1}(\bm{y}-\mathbf{X}\bm{\beta}-\mathbf{\Lambda}\tilde{\bm{u}})
&= (\bm{z}-\mathbf{X}\bm{\beta})^{\top} \mathbf{V}^{-1}(\bm{z}-\mathbf{X}\bm{\beta})\\
&= \bm{z}^{\top} \mathbf{V}^{-1/2} (\mathbf{I}_{mk} - \mathbf{P}_{\mathbf{X}}) \mathbf{V}^{-1/2}\bm{z},
\end{align*}
where $\mathbf{P}_{\mathbf{X}}=\mathbf{X}(\mathbf{X}^{\top}\mathbf{V}^{-1}\mathbf{X})\mathbf{X}^{\top}$. Integrating $\pi(\tilde{\bm{u}}, \bm{\beta}, \mathbf{\Sigma}, \mathbf{\Lambda}, \rho, \tau^2 \mid \bm{y})$ with respect to $\bm{\beta}$, we have
\begin{align*}
&\pi(\tilde{\bm{u}}, \mathbf{\Sigma}, \mathbf{\Lambda}, \rho, \tau^2 \mid \bm{y})\\
&=K \exp\left[-\frac{1}{2}(\bm{y} - \mathbf{\Lambda}\tilde{\bm{u}})^{\top} \mathbf{V}^{-1/2} (\mathbf{I}_{mk} - \mathbf{P}_{\mathbf{X}}) \mathbf{V}^{-1/2}(\bm{y} - \mathbf{\Lambda}\tilde{\bm{u}})\right]\\
&\quad \times |\tau^2((\mathbf{D}-\rho \mathbf{W})^{-1} \otimes \mathbf{\Sigma})|^{-1/2}  \exp\left[-\frac{1}{2}\tilde{\bm{u}}^{\top} \frac{1}{\tau^2}((\mathbf{D}-\rho \mathbf{W}) \otimes \mathbf{\Sigma}^{-1})\tilde{\bm{u}}\right]\\
&\quad \times \pi(\mathbf{\Sigma}) \left(\prod_{i=1}^m \prod_{j=1}^k \pi(\lambda_{i,j}) \right) \pi(\rho) \pi(\tau^2)\\
& \le K |\tau^2((\mathbf{D}-\rho \mathbf{W})^{-1} \otimes \mathbf{\Sigma})|^{-1/2}  \exp\left[-\frac{1}{2}\tilde{\bm{u}}^{\top} \frac{1}{\tau^2}((\mathbf{D}-\rho \mathbf{W}) \otimes \mathbf{\Sigma}^{-1})\tilde{\bm{u}}\right]\\
&\quad \times \pi(\mathbf{\Sigma}) \left(\prod_{i=1}^m \prod_{j=1}^k \pi(\lambda_{i,j}) \right) \pi(\rho) \pi(\tau^2),
\end{align*}
where $K$ is a positive constant. If the priors of $\mathbf{\Sigma}$, $\rho$, $\tau$ and $\lambda_{i,j}$ are proper, then the posterior density is proper.
    
\end{proof}

\subsection{Additional information on posterior computation}
\label{sec:appB}

We provide derivations of full conditional distributions, together with details of the sampling scheme for local parameters via the elliptical slice sampler.

\subsubsection{Derivation of algorithm}
\label{sec:B1}

\begin{itemize}
\item Sampling of regression coefficient vector $\bm{\beta}$.

The full conditional density of $\bm{\beta}$ is given by
\begin{align*}
\pi(\bm{\beta} \mid -)
&\propto \exp\left[-\frac{1}{2}(\bm{y}-\mathbf{X}\bm{\beta}-\mathbf{\Lambda}\tilde{\bm{u}})^{\top} \mathbf{V}^{-1}(\bm{y}-\mathbf{X}\bm{\beta}-\mathbf{\Lambda}\tilde{\bm{u}})\right]\\
&\propto \exp\left[-\frac{1}{2}(\bm{z}-\mathbf{X}\bm{\beta})^{\top} \mathbf{V}^{-1}(\bm{z}-\mathbf{X}\bm{\beta})\right] \quad (\bm{z}:=\bm{y}-\mathbf{\Lambda}\tilde{\bm{u}}),\\
& \propto \exp\left[-\frac{1}{2} (\bm{\beta}-\bm{m})^{\top}(\bm{X}^{\top}\mathbf{V}^{-1}\mathbf{X})(\bm{\beta}-\bm{m})\right],
\end{align*}
where $\bm{y} \in \mathbb{R}^{mk}$, $\mathbf{X}\in \mathbb{R}^{mk \times s}$, $\bm{\beta} \in \mathbb{R}^{s}$, $\tilde{\bm{u}}_i\in \mathbb{R}^{mk}$, $\mathbf{\Lambda}=\mathrm{block diag}(\mathbf{\Lambda}_1,\dots, \mathbf{\Lambda}_m) \in \mathbb{R}^{mk\times mk}$, $\mathbf{V}=\mathrm{block diag}(\mathbf{V}_1,\dots, \mathbf{V}_m) \in \mathbb{R}^{mk\times mk}$ and 
\[\bm{m}=(\mathbf{X}^{\top}\mathbf{V}^{-1}\mathbf{X})^{-1}\mathbf{X}^{\top}\mathbf{V}^{-1}\bm{z}, \quad \bm{z}:=\bm{y}-\mathbf{\Lambda}\tilde{\bm{u}}.\]
Hence, $\bm{\beta}\mid - \sim \mathcal{N}_s(\bm{m}, (\bm{X}^{\top}\mathbf{V}^{-1}\mathbf{X})^{-1})$.

\item Sampling of (transformed) random effect vector $\tilde{\bm{u}}$.

Using \eqref{MCAR-fullcond}, the full conditional density of $\tilde{\bm{u}}$ is given by
\begin{align*}
\pi(\tilde{\bm{u}}_i \mid -)
& \propto \exp\left[-\frac{1}{2}(\bm{y}_i-\mathbf{X}_i\bm{\beta}-\mathbf{\Lambda}_i\tilde{\bm{u}}_i)^{\top} \mathbf{V}_i^{-1}(\bm{y}_i-\mathbf{X}_i\bm{\beta}-\mathbf{\Lambda}_i\tilde{\bm{u}}_i)\right] \\
& \quad \times \exp\left(-\frac{1}{2}(\tilde{\bm{u}}_i- \bm{\mu}_{\mathrm{CAR},i})^{\top} \mathbf{\Omega}_{\mathrm{CAR},i}^{-1} (\tilde{\bm{u}}_i- \bm{\mu}_{\mathrm{CAR},i})\right)\\
&\propto \exp\left(-\frac{1}{2} \tilde{\bm{u}}_i^{\top} (\mathbf{\Lambda}_i \mathbf{V}_i^{-1}\mathbf{\Lambda}_i +\mathbf{\Omega}_{\mathrm{CAR},i}^{-1} )\tilde{\bm{u}}_i + \tilde{\bm{u}}_i^{\top} (\mathbf{\Lambda}_i \mathbf{V}_i^{-1}\bm{\xi}_i + \mathbf{\Omega}_{\mathrm{CAR},i}^{-1} \bm{\mu}_{\mathrm{CAR},i})\right) \\
& \qquad (\bm{\xi}_i:=\bm{y}_i-\mathbf{X}_i\bm{\beta}).
\end{align*}
where $\bm{\mu}_{\mathrm{CAR},i}=\rho\sum_{j\sim i}  \tilde{\bm{u}}_j/w_{i+}$ and $\mathbf{\Omega}_{\mathrm{CAR},i}=(\tau^2/w_{i+}) \mathbf{\Sigma}$.
Hence, $\tilde{\bm{u}}_i \mid - \sim \mathcal{N}_k(\mathbf{A}^{-1}\bm{b}, \mathbf{A}^{-1})$, where
\[\mathbf{A}=\mathbf{\Lambda}_i\mathbf{V}_i^{-1}\mathbf{\Lambda}_i +\mathbf{\Omega}_{\mathrm{CAR},i}^{-1}, \quad \bm{b}=\mathbf{\Lambda}_i \mathbf{V}_i^{-1}\bm{\xi}_i + \mathbf{\Omega}_{\mathrm{CAR},i}^{-1} \bm{\mu}_{\mathrm{CAR},i}.\]

\item Sampling of covariance matrix $\mathbf{\Sigma}$.

Let $\tilde{\mathbf{U}}:=(\bm{u}_1,\dots, \bm{u}_m) \in \mathbb{R}^{k\times m}$. We can also express as the matrix-normal distribution
\[\tilde{\mathbf{U}}\sim \mathcal{MN}_{k\times m} (\mathbf{O}, \mathbf{\Sigma}, \tau^2(\mathbf{D}-\rho \mathbf{W})^{-1})\]
whose density function is 
\begin{align}
f(\tilde{\mathbf{U}}) 
&\propto |\tau^2(\mathbf{D}-\rho \mathbf{W})^{-1}|^{-k/2} |\mathbf{\Sigma}|^{-m/2}\exp\left(-\frac{1}{2}\mathrm{tr}\left[ \tilde{\mathbf{U}}\frac{1}{\tau^2}(\mathbf{D}-\rho \mathbf{W}) \tilde{\mathbf{U}}^{\top}\mathbf{\Sigma}^{-1}\right]\right). \label{MatNorm}
\end{align}
Using \eqref{MatNorm}, the full conditional density of $\mathbf{\Sigma}$ is given by
\begin{align*}
\pi(\mathbf{\Sigma} \mid -)
&\propto  |\mathbf{\Sigma}|^{-m/2} \exp\left(-\frac{1}{2}\mathrm{tr}\left[ \tilde{\mathbf{U}}(\mathbf{D}-\rho \mathbf{W}) \tilde{\mathbf{U}}^{\top}\mathbf{\Sigma}^{-1}\right]\right)\\
& \quad \times |\mathbf{\Sigma}|^{-(\nu_0+k+1)/2} \exp\{-\mathrm{tr}(\mathbf{B} \mathbf{\Sigma}^{-1})/2\}\\
&= |\mathbf{\Sigma}|^{-(m+\nu_0+k+1)/2} \exp\left[-\frac{1}{2}\mathrm{tr}([\mathbf{B}+\tilde{\mathbf{U}}\frac{1}{\tau^2}(\mathbf{D}-\rho \mathbf{W}) \tilde{\mathbf{U}}^{\top}] \mathbf{\Sigma}^{-1})\right].
\end{align*}
Hence, $\mathbf{\Sigma}\mid - \sim \mathrm{IW}(m+\nu_0, \mathbf{B}+\tilde{\mathbf{U}}\frac{1}{\tau^2}(\mathbf{D}-\rho \mathbf{W}) \tilde{\mathbf{U}}^{\top})$.

\item Sampling of spatial correlation parameter $\rho$. 

The full conditional density of $\rho$ is given by
\begin{align*}
\pi(\rho \mid -) \propto  |\mathbf{D}-\rho \mathbf{W}|^{k/2} \times \exp\left[-\frac{1}{2}\tilde{\bm{u}}^{\top} \frac{1}{\tau^2} ((\mathbf{D}-\rho \mathbf{W}) \otimes \mathbf{\Sigma}^{-1})\tilde{\bm{u}}\right],
\end{align*}
where we use the fact $|\tau^2((\mathbf{D}-\rho \mathbf{W})^{-1} \otimes \mathbf{\Sigma})| =(\tau^2)^{mk} |(\mathbf{D}-\rho \mathbf{W})^{-1}|^k |\mathbf{\Sigma}|^m$. 

\item Sampling of global/local parameters
\begin{itemize}
\item The full conditional distribution of $\mathbf{\Lambda}$.
\begin{align*}
\pi(\bm{\lambda}_i \mid -) &
\propto \exp\left[-\frac{1}{2}(\bm{y}_i-\mathbf{X}_i\bm{\beta}-\mathbf{\Lambda}_i\tilde{\bm{u}}_i)^{\top} \mathbf{V}_i^{-1}(\bm{y}_i-\mathbf{X}_i\bm{\beta}-\mathbf{\Lambda}_i\tilde{\bm{u}}_i)\right] \prod_{j=1}^k\pi(\lambda_{i,j}) \notag \\
&= \exp\left[-\frac{1}{2}\left(\frac{\bm{y}_i-\mathbf{X}_i\bm{\beta}}{\tilde{\bm{u}}_i}-\bm{\lambda}_i\right)^{\top} \tilde{\mathbf{U}}_i\mathbf{V}_i^{-1}\tilde{\mathbf{U}}_i\left(\frac{\bm{y}_i-\mathbf{X}_i\bm{\beta}}{\tilde{\bm{u}}_i}-\bm{\lambda}_i\right)\right] \prod_{j=1}^k\pi(\lambda_{i,j}),
\end{align*}
where the fraction $(\bm{y}_i-\mathbf{X}_i\bm{\beta})/\tilde{\bm{u}}_i$ is applied elementwise and $\tilde{\mathbf{U}}_i=\mathrm{diag}(\tilde{v}_{i1},\dots,\tilde{v}_{ik})$.
Let
\[\bm{\mu}_{\bm{\lambda}_i}:=\frac{\bm{y}_i-\mathbf{X}_i\bm{\beta}}{\tilde{\bm{u}}_i}, \quad \mathbf{\Omega}_{\bm{\lambda}_i}=(\tilde{\mathbf{U}}_i\mathbf{V}_i^{-1}\tilde{\mathbf{U}}_i)^{-1}.\]

\item Sampling of $\tau^2$ (horseshoe prior)
\begin{align*}
\pi(\tau^2 \mid -) \propto (\tau^2)^{-mk/2}\exp\left[-\frac{1}{2}\tilde{\bm{u}}^{\top} \frac{1}{\tau^2}((\mathbf{D}-\rho \mathbf{W}) \otimes \mathbf{\Sigma}^{-1})\tilde{\bm{u}}\right] \pi(\tau^2)
\end{align*}
When $\tau\sim \mathrm{C}_+(0,1)$, using 
\[\tau\sim \mathrm{C}_+(0,1) \iff \tau^2\mid \psi \sim \mathrm{IG}\left(\frac{1}{2}, \frac{1}{\psi}\right), \quad \psi\sim \mathrm{IG}\left(\frac{1}{2},1\right),\]
we have
\begin{align*}
\pi(\tau^2\mid \psi, \mathrm{others})
&\propto (\tau^2)^{-\frac{mk+1}{2}-1} \exp\left[-\frac{1}{\tau^2} \left(\frac{1}{2}\tilde{\bm{u}}^{\top} ((\mathbf{D}-\rho \mathbf{W}) \otimes \mathbf{\Sigma}^{-1})\tilde{\bm{u}}  + \frac{1}{\psi}\right)\right]\\
&\sim \mathrm{IG}\left(\frac{mk+1}{2},\frac{1}{2}\tilde{\bm{u}}^{\top} ((\mathbf{D}-\rho \mathbf{W}) \otimes \mathbf{\Sigma}^{-1})\tilde{\bm{u}}  + \frac{1}{\psi} \right),\\
\pi(\psi\mid \tau^2) &\propto \psi^{-1/2-1}\exp\left(-\frac{1}{\psi}\left(\frac{1}{\tau^2} + 1\right)\right)\\
&\sim \mathrm{IG}\left(\frac{1}{2},\frac{1}{\tau^2} + 1 \right)
\end{align*}

\item Sampling of $a$ and $b$ (normal-gamma prior)

Under the prior $a\sim U(0,1)$ and $b\sim \mathrm{Ga}(c,d)$, the full conditional densities for $a$ and $b$ are given by
\begin{align*}
\pi(b\mid -) &  \propto \left(\prod_{i=1}^m \prod_{j=1}^k b^a\exp(-b \lambda_{i,j}^2) \right) \times b^{c-1} \exp(-db)\\
&= b^{(amk + c)-1} \exp\left(-b\left( \sum_{i=1}^m \sum_{j=1}^k \lambda_{i,j}^2 + d\right)\right)\\
&\sim \mathrm{Ga}\left(amk + c, \sum_{i=1}^m \sum_{j=1}^k \lambda_{i,j}^2 + d \right),\\
\pi(a\mid -)&\propto \prod_{i=1}^m \prod_{j=1}^k \frac{b^a}{\Gamma(a)} \lambda_{i,j}^{2a-1}= \left(\frac{b^a}{\Gamma(a)}\right)^{mk} \left( \prod_{i=1}^m \prod_{j=1}^k \lambda_{i,j} \right)^{2a-1}.
\end{align*}

\end{itemize}
\end{itemize}

\subsubsection{Sampling of local parameter}
\label{sec:appB2}

We provide an overview of the elliptical slice sampler used for sampling $\bm{\lambda}_1, \dots, \bm{\lambda}_m$. For further details, the reader is referred to \cite{murray2010elliptical} and \cite{hahn2019efficient}. We define 
\[\pi^*(\bm{\lambda}_i)=\prod_{j=1}^k \left\{\pi(\lambda_{i,j}) \frac{1}{1+\exp(-\eta \lambda_{i,j})}\right\},\]
where $\eta$ is a large constant (e.g., $\eta=100$). For $i=1,\dots,m$, let $\bm{\Delta}_i = \bm{\lambda}_{i}-\bm{\mu}_{\bm{\lambda}_{i}}$. Sampling is performed as follows.
\begin{itemize}
\item[1.] Draw $\bm{\zeta}_i\sim \mathcal{N}_k(\bm{0}, \mathbf{\Omega}_{\bm{\lambda}_{i}})$\\
Set $\bm{\zeta}_{0i}=\bm{\Delta}_i \sin \theta + \bm{\zeta}_i\cos \theta$\\
Set $\bm{u}_{1i}=\bm{\Delta}_i \cos\theta -\bm{\zeta}_i\sin \theta$.
\item[2.] Draw $\ell$ from $\mathcal{U}(0, \pi^*(\bm{\mu}_{\bm{\lambda}_{i}} + \bm{\zeta}_{0i}\sin \theta + \bm{\zeta}_{1i}\cos \theta))$.\\
Initialize $a=0$ and $b=2\pi$.
\begin{itemize}
\item[2.1] Sample $\theta'$ from $\mathcal{U}(a,b)$.
\item[2.2] If $\pi^*(\bm{\mu}_{\bm{\lambda}_{i}} + \bm{\zeta}_{0i}\sin \theta' + \bm{\zeta}_{1i}\cos \theta')>\ell$\\
Then: Set $\theta \leftarrow \theta'$. Go to step 3.\\
Else: If $\theta'<\theta$, set $a\leftarrow \theta'$, else set $b\leftarrow \theta'$. Go to step 2.1.
\end{itemize}
\item[3.] Return $\bm{\Delta}_i=\bm{\zeta}_{0i} \sin \theta + \bm{\zeta}_{1i}\cos \theta$ and $\bm{\lambda}_{i}=\bm{\mu}_{\bm{\lambda}_i}+\bm{\Delta}_i$.
\end{itemize}
If we assume the horseshoe prior, we set $\pi(\lambda_{i,j})=(2/\pi)(1+\lambda_{i,j}^2)^{-1}$. When we use the normal-gamma type prior $\lambda_{i,j}^2 \sim \mathrm{Ga}(a,b)$, we set $\pi(\lambda_{i,j})=\frac{2 b^a}{\Gamma(a)} \lambda_{i,j}^{2a-1}\exp(-b \lambda_{i,j}^2)$.

\subsection{Additional information on real data analysis}
\label{sec:appC}

We present additional results from the data analysis described in Section~\ref{sec:5}.

We conduct an exploratory investigation of the shrinkage behavior of random effects under different models.
As a benchmark, we first consider the Fay--Herriot (FH) model, which is the most basic among the competing methods.
Let $\bm{e} = \bm{y} - \mathbf{X}\hat{\bm{\beta}}$ 
denote the marginal residuals under the FH model, where $\hat{\bm{\beta}}$ is the posterior mean of $\bm{\beta}$. The FH model is used as the baseline for no reason other than the fact that it is the most basic and widely used model in small area estimation. In the data analysis, the estimated regression coefficients were found to be almost identical across all the considered methods, and thus no noticeable differences were observed among the methods in terms of regression coefficient estimation.
Note that $\bm{e}$ is not the usual regression residual; rather, it is obtained after marginalizing out the random effects. Thus, the analysis should be regarded as exploratory. 
From this perspective, $\bm{e}$ can be regarded as pseudo-data, and estimating random effects under sparse models may be interpreted as applying shrinkage to this pseudo-data.
To examine the degree and structure of shrinkage, we analyze the sparsity patterns in $\bm{e}$ and their relationship with the estimated random effects.

Figure~\ref{fig:marginal-residual} presents a scatter plot of the two-dimensional residual vectors $\bm{e}$.
The color of each point represents a similarity measure between the two components of $\bm{e}_i=(e_{i1}, e_{i2})$ (after normalization in each dimension).
Specifically, we adopt a Bray--Curtis type similarity \citep{bray1957ordination} defined by
\[
\delta_i
=
\frac{\bigl||e_{i1}| - |e_{i2}|\bigr|}
{|e_{i1}| + |e_{i2}|},
\qquad
\text{Similarity}_i = 1 - \delta_i
\]
for $i=1,\dots,m$. This measure takes values close to $1$ when the absolute values of the two components are similar, and values close to $0$ when they differ substantially.
Points lying near the coordinate axes correspond to cases where one component is close to zero while the other is not, resulting in low similarity (blue).
In contrast, points along the diagonal lines exhibit high similarity (red), indicating that the magnitudes of the two components are comparable.
Based on this similarity measure, we divide the observations into groups with low similarity (one component sparse) and high similarity (both components sparse or both large), and investigate how different methods shrink the random effects in each case.

We first consider the case where the similarity measure is at most $0.25$, corresponding to observations whose components differ substantially.
This setting represents situations in which sparsity is uneven across dimensions.
Figure~\ref{fig:residual-rf-comparison-1} displays scatter plots of $\bm{e}$ against the point estimates of the random effects obtained by SpaHS, SpaFH, and HS, shown separately for each dimension.
The color indicates the value of the second component of $\bm{e}_i$ (poverty rate), which is shared across the panels.
For observations colored in blue, where the second component is small, the first component (median household income) is close to zero.
In this case, all methods shrink the estimated random effects toward zero in the first dimension.
However, for the second dimension, both SpaFH and HS also shrink the estimates toward zero, even when the residual signal is present.
This suggests that, since these methods perform shrinkage jointly across dimensions, a near-zero random effect in median household income leads to excessive shrinkage in poverty rate as well.
In contrast, SpaHS mitigates this effect, exhibiting weaker shrinkage in the second dimension.
A similar phenomenon is observed for the orange-colored points, where the poverty rate component is close to zero.
Although all methods appropriately shrink the poverty rate effects, HS and SpaFH also excessively shrink the median household income effects, whereas SpaHS better preserves the nonzero signal.

Next, we examine the case where the similarity is at least $0.75$, corresponding to observations whose components have similar magnitudes.
The results are shown in Figure~\ref{fig:residual-rf-comparison-2}.
These observations represent cases in which both components are either sparse or large.
In this setting, HS and SpaHS yield similar shrinkage behavior, as both dimensions are shrunk or retained simultaneously.
In particular, when the residuals are close to zero, the estimates are shrunk toward zero.
For moderate residual values, the estimated random effects exhibit asymmetric shrinkage: for negative residuals, the estimates lie above the line $y=x$, while for positive residuals they lie below it.
SpaFH, on the other hand, exhibits a noticeably different pattern from HS and SpaHS in this regime.

Although this analysis is exploratory in nature, the results consistently indicate that the proposed SpaHS method is capable of performing component-specific shrinkage of random effects.
Through the real data analysis, we observe that SpaHS can adaptively shrink each component of the random-effect vector, thereby alleviating the over-shrinkage problem inherent in models that impose common shrinkage across dimensions.

\begin{figure}[htbp]
    \centering
    \includegraphics[width=0.6\linewidth]{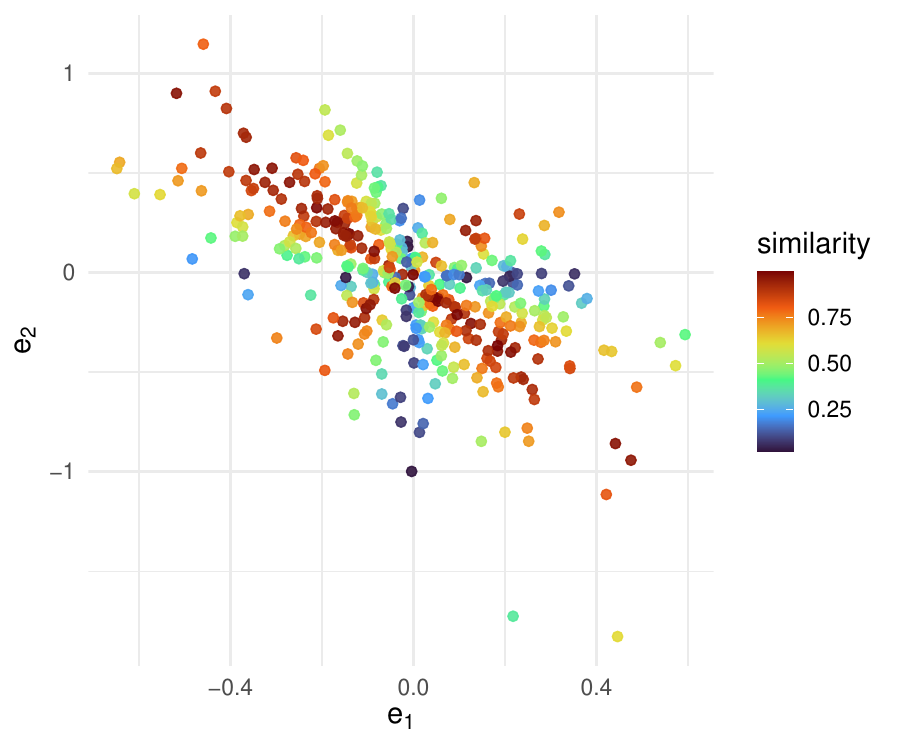}
    \caption{
Scatter plot of the two-dimensional marginal residual
$\bm{e}_i = (e_{i1}, e_{i2})$ obtained under the Fay--Herriot (FH) model.
Each point corresponds to a small area.
The color indicates the similarity between the absolute values of the two components,
defined by a Bray--Curtis type similarity
$1 - \delta_i$, where
$\delta_i = \bigl||e_{i1}| - |e_{i2}|\bigr| / (|e_{i1}| + |e_{i2}|)$.
}
  \label{fig:marginal-residual}
\end{figure}

\begin{figure}[htbp]
    \centering
    \includegraphics[width=\linewidth]{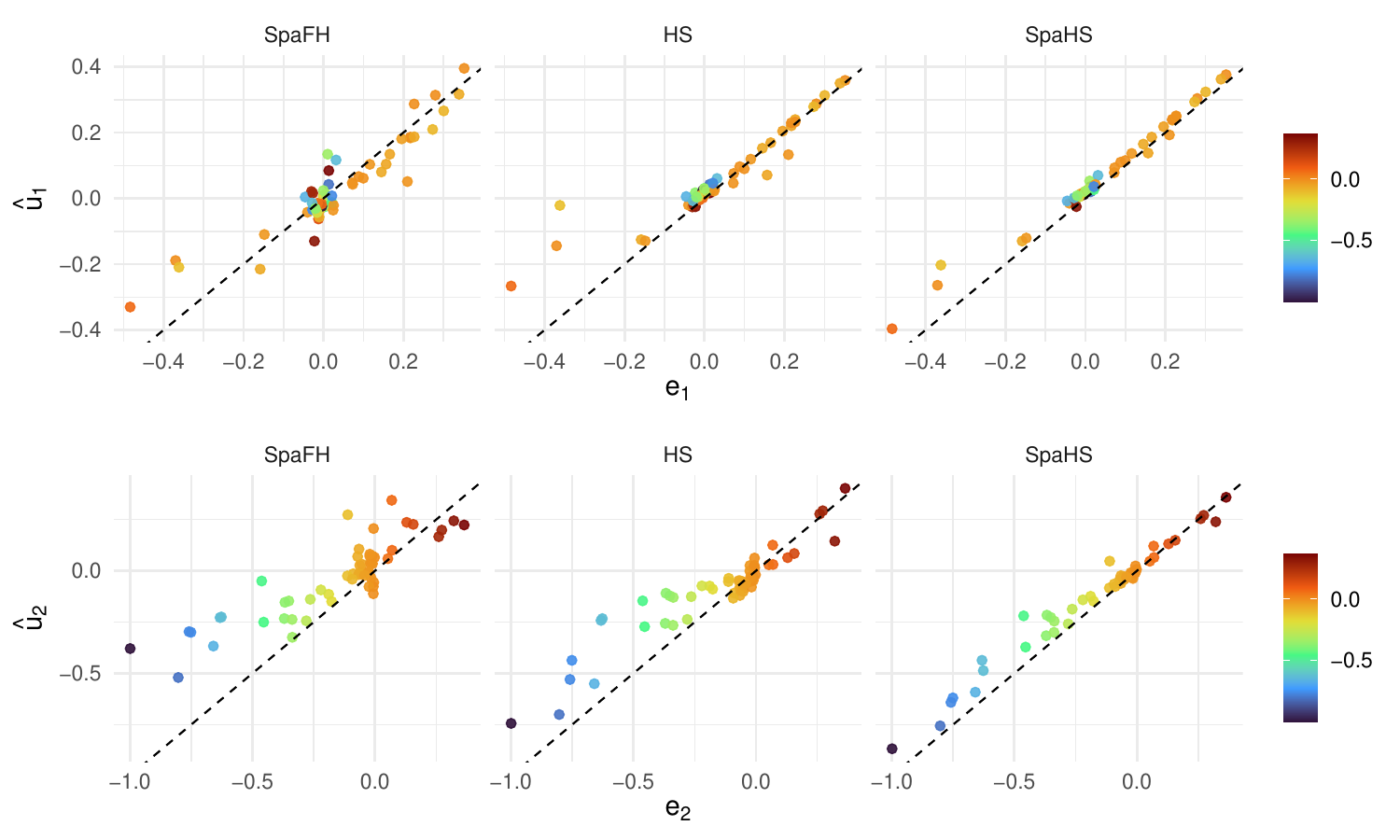}
    \caption{
Scatter plots of the FH marginal residual $\bm{e}_i=(e_{i1}, e_{i2})$ versus the estimated random effects
under SpaHS, SpaFH, and HS, for areas with low similarity
($1-\delta_i \le 0.25$).
These observations correspond to situations where sparsity is uneven across dimensions,
that is, one component is close to zero while the other is nonzero.
The upper and lower panels display the results for the first component
(median household income) and the second component (poverty rate), respectively.
Points are colored according to the value of the second residual component $e_{i2}$,
with blue indicating values close to zero and orange indicating larger magnitudes.
}
\label{fig:residual-rf-comparison-1}
\end{figure}

\begin{figure}[htbp]
    \centering
    \includegraphics[width=\linewidth]{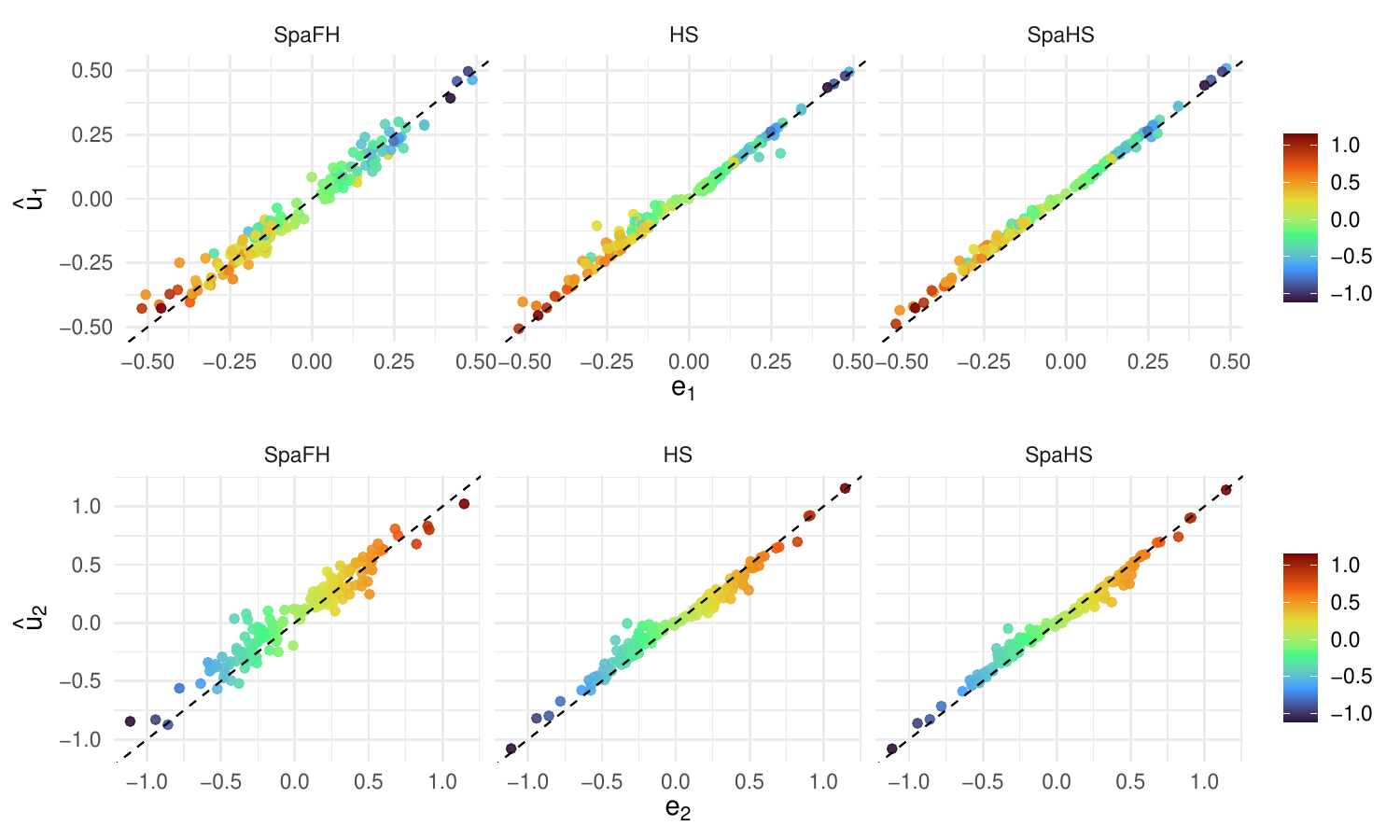}
    \caption{
Scatter plots of the FH marginal residual $\bm{e}_i=(e_{i1}, e_{i2})$ versus the estimated random effects
under SpaHS, SpaFH, and HS, for areas with high similarity
($1-\delta_i \ge 0.75$).
These observations correspond to cases where both components of $\bm{e}_i$ are either small
or large, indicating similar sparsity patterns across dimensions.
As in Figure~\ref{fig:residual-rf-comparison-1},
the upper and lower panels show the first and second components, respectively,
and points are colored according to the value of $e_{i2}$.
}
\label{fig:residual-rf-comparison-2}
\end{figure}

\vspace{5mm}
\bibliography{refs}

\begin{thebibliography}{}

\bibitem[\protect\citeauthoryear{Banerjee, Gelfand, and Carlin}{Banerjee
  et~al.}{2025}]{banerjee2025hierarchical}
Banerjee, S., A.~E. Gelfand, and B.~P. Carlin (2025).
\newblock {\em Hierarchical modeling and analysis for spatial data}.
\newblock CRC press.

\bibitem[\protect\citeauthoryear{Barnard, McCulloch, and Meng}{Barnard
  et~al.}{2000}]{barnard2000modeling}
Barnard, J., R.~McCulloch, and X.-L. Meng (2000).
\newblock Modeling covariance matrices in terms of standard deviations and
  correlations, with application to shrinkage.
\newblock {\em Statistica Sinica\/}, 1281--1311.

\bibitem[\protect\citeauthoryear{Battese, Harter, and Fuller}{Battese
  et~al.}{1988}]{Battese01031988}
Battese, G.~E., R.~M. Harter, and W.~A. Fuller (1988).
\newblock An error-components model for prediction of county crop areas using
  survey and satellite data.
\newblock {\em Journal of the American Statistical Association\/}~{\em
  83\/}(401), 28--36.

\bibitem[\protect\citeauthoryear{Benavent and Morales}{Benavent and
  Morales}{2016}]{benavent2016multivariate}
Benavent, R. and D.~Morales (2016).
\newblock Multivariate fay--herriot models for small area estimation.
\newblock {\em Computational Statistics \& Data Analysis\/}~{\em 94}, 372--390.

\bibitem[\protect\citeauthoryear{Bhattacharya, Pati, Pillai, and
  Dunson}{Bhattacharya et~al.}{2015}]{bhattacharya2015dirichlet}
Bhattacharya, A., D.~Pati, N.~S. Pillai, and D.~B. Dunson (2015).
\newblock Dirichlet--laplace priors for optimal shrinkage.
\newblock {\em Journal of the American Statistical Association\/}~{\em
  110\/}(512), 1479--1490.

\bibitem[\protect\citeauthoryear{Blei, Kucukelbir, and McAuliffe}{Blei
  et~al.}{2017}]{blei2017variational}
Blei, D.~M., A.~Kucukelbir, and J.~D. McAuliffe (2017).
\newblock Variational inference: A review for statisticians.
\newblock {\em Journal of the American statistical Association\/}~{\em
  112\/}(518), 859--877.

\bibitem[\protect\citeauthoryear{Boehm~Vock, Reich, Fuentes, and
  Dominici}{Boehm~Vock et~al.}{2015}]{boehm2015spatial}
Boehm~Vock, L.~F., B.~J. Reich, M.~Fuentes, and F.~Dominici (2015).
\newblock Spatial variable selection methods for investigating acute health
  effects of fine particulate matter components.
\newblock {\em Biometrics\/}~{\em 71\/}(1), 167--177.

\bibitem[\protect\citeauthoryear{Bray and Curtis}{Bray and
  Curtis}{1957}]{bray1957ordination}
Bray, J.~R. and J.~T. Curtis (1957).
\newblock An ordination of the upland forest communities of southern wisconsin.
\newblock {\em Ecological Monographs\/}~{\em 27}, 325--349.

\bibitem[\protect\citeauthoryear{Carlin and Banerjee}{Carlin and
  Banerjee}{2003}]{carlin2003hierarchical}
Carlin, B.~P. and S.~Banerjee (2003).
\newblock Hierarchical multivariate car models for spatio-temporally correlated
  survival data.
\newblock In {\em Bayesian Statistics 7: Proceedings of the Seventh Valencia
  International Meeting}. Oxford University Press.

\bibitem[\protect\citeauthoryear{Carvalho, Polson, and Scott}{Carvalho
  et~al.}{2010}]{carvalho2010horseshoe}
Carvalho, C.~M., N.~G. Polson, and J.~G. Scott (2010).
\newblock The horseshoe estimator for sparse signals.
\newblock {\em Biometrika\/}, 465--480.

\bibitem[\protect\citeauthoryear{Chakraborty, Datta, and Mandal}{Chakraborty
  et~al.}{2016}]{chakraborty2016two}
Chakraborty, A., G.~S. Datta, and A.~Mandal (2016).
\newblock A two-component normal mixture alternative to the fay-herriot model.
\newblock {\em Statistics in Transition New Series\/}~{\em 17\/}(1), 67--90.

\bibitem[\protect\citeauthoryear{Chung and Datta}{Chung and
  Datta}{2022}]{chung2022bayesian}
Chung, H.~C. and G.~S. Datta (2022).
\newblock Bayesian spatial models for estimating means of sampled and
  non-sampled small areas.
\newblock {\em Survey Methodology\/}~{\em 48\/}(2), 463--489.

\bibitem[\protect\citeauthoryear{Datta, Day, and Maiti}{Datta
  et~al.}{1998}]{datta1998multivariate}
Datta, G.~S., B.~Day, and T.~Maiti (1998).
\newblock Multivariate bayesian small area estimation: an application to survey
  and satellite data.
\newblock {\em Sankhy{\=a}: The Indian Journal of Statistics, Series A\/},
  344--362.

\bibitem[\protect\citeauthoryear{Datta, Hall, and Mandal}{Datta
  et~al.}{2011}]{datta2011model}
Datta, G.~S., P.~Hall, and A.~Mandal (2011).
\newblock Model selection by testing for the presence of small-area effects,
  and application to area-level data.
\newblock {\em Journal of the American Statistical Association\/}~{\em
  106\/}(493), 362--374.

\bibitem[\protect\citeauthoryear{Datta, Lahiri, Maiti, and Lu}{Datta
  et~al.}{1999}]{datta1999hierarchical}
Datta, G.~S., P.~Lahiri, T.~Maiti, and K.~L. Lu (1999).
\newblock Hierarchical bayes estimation of unemployment rates for the states of
  the us.
\newblock {\em Journal of the American Statistical Association\/}~{\em
  94\/}(448), 1074--1082.

\bibitem[\protect\citeauthoryear{Datta and Mandal}{Datta and
  Mandal}{2015}]{datta2015small}
Datta, G.~S. and A.~Mandal (2015).
\newblock Small area estimation with uncertain random effects.
\newblock {\em Journal of the American Statistical Association\/}~{\em
  110\/}(512), 1735--1744.

\bibitem[\protect\citeauthoryear{Datta, Rao, and Smith}{Datta
  et~al.}{2005}]{datta2005measuring}
Datta, G.~S., J.~Rao, and D.~D. Smith (2005).
\newblock On measuring the variability of small area estimators under a basic
  area level model.
\newblock {\em Biometrika\/}~{\em 92\/}(1), 183--196.

\bibitem[\protect\citeauthoryear{Fay}{Fay}{1987}]{fay1987application}
Fay, R.~E. (1987).
\newblock Application of multivariate regression to small domain estimation.
\newblock {\em Small area statistics\/}, 91--102.

\bibitem[\protect\citeauthoryear{Fay and Herriot}{Fay and
  Herriot}{1979}]{fay1979estimates}
Fay, R.~E. and R.~A. Herriot (1979).
\newblock Estimates of income for small places: an application of james-stein
  procedures to census data.
\newblock {\em Journal of the American Statistical Association\/}~{\em
  74\/}(366a), 269--277.

\bibitem[\protect\citeauthoryear{Gelfand and Ghosh}{Gelfand and
  Ghosh}{1998}]{gelfand1998model}
Gelfand, A.~E. and S.~K. Ghosh (1998).
\newblock Model choice: a minimum posterior predictive loss approach.
\newblock {\em Biometrika\/}~{\em 85\/}(1), 1--11.

\bibitem[\protect\citeauthoryear{Gelfand and Vounatsou}{Gelfand and
  Vounatsou}{2003}]{gelfand2003proper}
Gelfand, A.~E. and P.~Vounatsou (2003).
\newblock Proper multivariate conditional autoregressive models for spatial
  data analysis.
\newblock {\em Biostatistics\/}~{\em 4\/}(1), 11--15.

\bibitem[\protect\citeauthoryear{Ghosh, Ghosh, Maples, and Tang}{Ghosh
  et~al.}{2022}]{ghosh2022multivariate}
Ghosh, T., M.~Ghosh, J.~J. Maples, and X.~Tang (2022).
\newblock Multivariate global-local priors for small area estimation.
\newblock {\em Stats\/}~{\em 5\/}(3), 673--688.

\bibitem[\protect\citeauthoryear{Griffin and Brown}{Griffin and
  Brown}{2010}]{griffin2010inference}
Griffin, J.~E. and P.~J. Brown (2010).
\newblock Inference with normal-gamma prior distributions in regression
  problems.
\newblock {\em Bayesian Analysis\/}~{\em 5\/}(1), 171--188.

\bibitem[\protect\citeauthoryear{Hahn, He, and Lopes}{Hahn
  et~al.}{2019}]{hahn2019efficient}
Hahn, P.~R., J.~He, and H.~F. Lopes (2019).
\newblock Efficient sampling for gaussian linear regression with arbitrary
  priors.
\newblock {\em Journal of Computational and Graphical Statistics\/}~{\em
  28\/}(1), 142--154.

\bibitem[\protect\citeauthoryear{Hamura, Irie, and Sugasawa}{Hamura
  et~al.}{2025}]{hamura2025outlier}
Hamura, Y., K.~Irie, and S.~Sugasawa (2025).
\newblock Outlier-robust bayesian multivariate analysis with correlation-intact
  sandwich mixture.
\newblock {\em arXiv preprint arXiv:2508.18004\/}.

\bibitem[\protect\citeauthoryear{Hoff}{Hoff}{2009}]{hoff2009first}
Hoff, P.~D. (2009).
\newblock {\em A First Course in Bayesian Statistical Methods}, Volume 580.
\newblock Springer.

\bibitem[\protect\citeauthoryear{Ito and Kubokawa}{Ito and
  Kubokawa}{2020}]{ito2020robust}
Ito, T. and T.~Kubokawa (2020).
\newblock Robust estimation of mean squared error matrix of small area
  estimators in a multivariate fay--herriot model.
\newblock {\em Japanese Journal of Statistics and Data Science\/}~{\em 3\/}(1),
  39--61.

\bibitem[\protect\citeauthoryear{Ito and Kubokawa}{Ito and
  Kubokawa}{2021}]{ITO202112}
Ito, T. and T.~Kubokawa (2021).
\newblock Corrected empirical bayes confidence region in a multivariate
  fay--herriot model.
\newblock {\em Journal of Statistical Planning and Inference\/}~{\em 211},
  12--32.

\bibitem[\protect\citeauthoryear{Jin, Banerjee, and Carlin}{Jin
  et~al.}{2007}]{jin2007order}
Jin, X., S.~Banerjee, and B.~P. Carlin (2007).
\newblock Order-free co-regionalized areal data models with application to
  multiple-disease mapping.
\newblock {\em Journal of the Royal Statistical Society Series B: Statistical
  Methodology\/}~{\em 69\/}(5), 817--838.

\bibitem[\protect\citeauthoryear{Jin, Carlin, and Banerjee}{Jin
  et~al.}{2005}]{jin2005generalized}
Jin, X., B.~P. Carlin, and S.~Banerjee (2005).
\newblock Generalized hierarchical multivariate car models for areal data.
\newblock {\em Biometrics\/}~{\em 61\/}(4), 950--961.

\bibitem[\protect\citeauthoryear{Kingma and Welling}{Kingma and
  Welling}{2013}]{kingma2013auto}
Kingma, D.~P. and M.~Welling (2013).
\newblock Auto-encoding variational bayes.
\newblock {\em arXiv preprint arXiv:1312.6114\/}.

\bibitem[\protect\citeauthoryear{Lewandowski, Kurowicka, and Joe}{Lewandowski
  et~al.}{2009}]{lewandowski2009generating}
Lewandowski, D., D.~Kurowicka, and H.~Joe (2009).
\newblock Generating random correlation matrices based on vines and extended
  onion method.
\newblock {\em Journal of Multivariate Analysis\/}~{\em 100\/}(9), 1989--2001.

\bibitem[\protect\citeauthoryear{MacNab}{MacNab}{2022}]{macnab2022bayesian}
MacNab, Y.~C. (2022).
\newblock Bayesian disease mapping: Past, present, and future.
\newblock {\em Spatial Statistics\/}~{\em 50}, 100593.

\bibitem[\protect\citeauthoryear{Mardia}{Mardia}{1988}]{mardia1988multi}
Mardia, K. (1988).
\newblock Multi-dimensional multivariate gaussian markov random fields with
  application to image processing.
\newblock {\em Journal of Multivariate Analysis\/}~{\em 24\/}(2), 265--284.

\bibitem[\protect\citeauthoryear{Murray, Adams, and MacKay}{Murray
  et~al.}{2010}]{murray2010elliptical}
Murray, I., R.~Adams, and D.~MacKay (2010).
\newblock Elliptical slice sampling.
\newblock In {\em Proceedings of the Thirteenth International Conference on
  Artificial Intelligence and Statistics}, pp.\  541--548. JMLR Workshop and
  Conference Proceedings.

\bibitem[\protect\citeauthoryear{Onizuka, Iwashige, and Hashimoto}{Onizuka
  et~al.}{2024}]{onizuka2024bayesian}
Onizuka, T., F.~Iwashige, and S.~Hashimoto (2024).
\newblock Bayesian boundary trend filtering.
\newblock {\em Computational Statistics \& Data Analysis\/}~{\em 191}, 107889.

\bibitem[\protect\citeauthoryear{Park and Casella}{Park and
  Casella}{2008}]{park2008bayesian}
Park, T. and G.~Casella (2008).
\newblock The bayesian lasso.
\newblock {\em Journal of the American Statistical Association\/}~{\em
  103\/}(482), 681--686.

\bibitem[\protect\citeauthoryear{Polson, Scott, and Windle}{Polson
  et~al.}{2013}]{polson2013bayesian}
Polson, N.~G., J.~G. Scott, and J.~Windle (2013).
\newblock Bayesian inference for logistic models using p{\'o}lya--gamma latent
  variables.
\newblock {\em Journal of the American statistical Association\/}~{\em
  108\/}(504), 1339--1349.

\bibitem[\protect\citeauthoryear{Porter, Holan, Wikle, and Cressie}{Porter
  et~al.}{2014}]{porter2014spatial}
Porter, A.~T., S.~H. Holan, C.~K. Wikle, and N.~Cressie (2014).
\newblock Spatial fay--herriot models for small area estimation with functional
  covariates.
\newblock {\em Spatial Statistics\/}~{\em 10}, 27--42.

\bibitem[\protect\citeauthoryear{Porter, Wikle, and Holan}{Porter
  et~al.}{2015}]{porter2015small}
Porter, A.~T., C.~K. Wikle, and S.~H. Holan (2015).
\newblock Small area estimation via multivariate fay--herriot models with
  latent spatial dependence.
\newblock {\em Australian \& New Zealand Journal of Statistics\/}~{\em
  57\/}(1), 15--29.

\bibitem[\protect\citeauthoryear{Rao and Molina}{Rao and
  Molina}{2015}]{rao2015small}
Rao, J.~N. and I.~Molina (2015).
\newblock {\em Small area estimation}.
\newblock John Wiley \& Sons.

\bibitem[\protect\citeauthoryear{Ray, Pati, and Bhattacharya}{Ray
  et~al.}{2020}]{ray2020efficient}
Ray, P., D.~Pati, and A.~Bhattacharya (2020).
\newblock Efficient bayesian shape-restricted function estimation with
  constrained gaussian process priors.
\newblock {\em Statistics and Computing\/}~{\em 30\/}(4), 839--853.

\bibitem[\protect\citeauthoryear{Spiegelhalter, Best, Carlin, and Van
  Der~Linde}{Spiegelhalter et~al.}{2002}]{spiegelhalter2002bayesian}
Spiegelhalter, D.~J., N.~G. Best, B.~P. Carlin, and A.~Van Der~Linde (2002).
\newblock Bayesian measures of model complexity and fit.
\newblock {\em Journal of the Royal Statistical Society: Series B (Statistical
  Methodology)\/}~{\em 64\/}(4), 583--639.

\bibitem[\protect\citeauthoryear{Sugasawa and Kubokawa}{Sugasawa and
  Kubokawa}{2020}]{sugasawa2020small}
Sugasawa, S. and T.~Kubokawa (2020).
\newblock Small area estimation with mixed models: a review.
\newblock {\em Japanese Journal of Statistics and Data Science\/}~{\em 3\/}(2),
  693--720.

\bibitem[\protect\citeauthoryear{Tang and Ghosh}{Tang and
  Ghosh}{2023}]{tang2023global}
Tang, X. and M.~Ghosh (2023).
\newblock Global-local priors for spatial small area estimation.
\newblock {\em Calcutta Statistical Association Bulletin\/}~{\em 75\/}(2),
  141--154.

\bibitem[\protect\citeauthoryear{Tang, Ghosh, Ha, and Sedransk}{Tang
  et~al.}{2018}]{tang2018modeling}
Tang, X., M.~Ghosh, N.~S. Ha, and J.~Sedransk (2018).
\newblock Modeling random effects using global--local shrinkage priors in small
  area estimation.
\newblock {\em Journal of the American Statistical Association\/}~{\em
  113\/}(524), 1476--1489.

\bibitem[\protect\citeauthoryear{Tang and Zhang}{Tang and
  Zhang}{2024}]{tang2024hierarchical}
Tang, X. and L.~Zhang (2024).
\newblock A hierarchical gamma prior for modeling random effects in small area
  estimation.
\newblock {\em Survey Methodology\/}~{\em 50\/}(2), 287--301.

\bibitem[\protect\citeauthoryear{Wang, Parker, and Holan}{Wang
  et~al.}{2025}]{wang2025variational}
Wang, Z., P.~A. Parker, and S.~H. Holan (2025).
\newblock Variational autoencoded multivariate spatial fay--herriot models.
\newblock {\em Spatial Statistics\/}, 100929.

\bibitem[\protect\citeauthoryear{Zhang, Ma, Wikle, and Huser}{Zhang
  et~al.}{2023}]{zhang2023flexible}
Zhang, L., X.~Ma, C.~K. Wikle, and R.~Huser (2023).
\newblock Flexible and efficient spatial extremes emulation via variational
  autoencoders.
\newblock {\em arXiv preprint arXiv:2307.08079\/}.

\end{thebibliography}
\bibliographystyle{chicago}

\newpage

\newpage
\setcounter{page}{1}
\setcounter{equation}{0}
\renewcommand{\theequation}{S\arabic{equation}}
\setcounter{section}{0}
\renewcommand{\thesection}{S\arabic{section}}
\setcounter{subsection}{0}
\renewcommand{\thesubsection}{S\arabic{subsection}}
\setcounter{table}{0}
\renewcommand{\thetable}{S\arabic{table}}
\setcounter{figure}{0}
\renewcommand{\thefigure}{S\arabic{figure}}

\begin{center}
{\LARGE\bf Supplementary Materials for ``Global-local shrinkage priors for modeling random
effects in multivariate spatial small area estimation"}
\end{center}

\vspace{0.5cm}
\begin{center}
{\large 
Shushi Nishina$^{1}$, Takahiro Onizuka$^{2}$ and Shintaro Hashimoto$^{1}$
}
\end{center}

\medskip
\noindent
$^{1}$Department of Mathematics, Hiroshima University\\
$^{2}$Graduate School of Social Sciences, Chiba University

\vspace{1cm}
In this supplementary material, we provide the results of additional numerical experiments as well as additional information regarding the real data analysis presented in the main text.

\section{Additional simulation results}
\label{sec:supp-1}

In this section, we report the results of additional numerical experiments. The data generation scenarios 1--5 and the observation variance scenarios (a) and (b) remain the same as in the main text. While the main text considered a large spatial correlation parameter ($\rho = 0.95$), the first experiment investigates the case where $\rho = 0.30$. Furthermore, while the dimension $k$ of the direct estimator $\bm{y}_i$ was set to $k = 2$ in the main text, we also report experiments for $k = 4$. The prior distribution settings are identical to those in the main text, and the tables display the medians over 50 replications.

\begin{itemize}
\item Under variance scenario (a) for $m = 100, 200, 500$ with $(k, \rho) = (2, 0.30)$, the results are shown in Tables~\ref{tab:results_m100_k2_030_a}, \ref{tab:results_m200_k2_030_a}, and \ref{tab:results_m500_k2_030_a}.
\item Under variance scenario (b) for $m = 100, 200, 500$ with $(k, \rho) = (2, 0.30)$, the results are shown in Tables~\ref{tab:results_m100_k2_030_b}, \ref{tab:results_m200_k2_030_b}, and \ref{tab:results_m500_k2_030_b}.
\item Under variance scenario (a) for $m = 100, 200, 500$ with $(k, \rho) = (4, 0.95)$, the results are shown in Tables~\ref{tab:results_m100_k4_095_a}, \ref{tab:results_m200_k4_095_a}, and \ref{tab:results_m500_k4_095_a}.
\item Under variance scenario (b) for $m = 100, 200, 500$ with $(k, \rho) = (4, 0.95)$, the results are shown in Tables~\ref{tab:results_m100_k4_095_b}, \ref{tab:results_m200_k4_095_b}, and \ref{tab:results_m500_k4_095_b}.
\item Under variance scenario (a) for $m = 100, 200, 500$ with $(k, \rho) = (4, 0.30)$, the results are shown in Tables~\ref{tab:results_m100_k4_030_a}, \ref{tab:results_m200_k4_030_a}, and \ref{tab:results_m500_k4_030_a}.
\item Under variance scenario (b) for $m = 100, 200, 500$ with $(k, \rho) = (4, 0.30)$, the results are shown in Tables~\ref{tab:results_m100_k4_030_b}, \ref{tab:results_m200_k4_030_b}, and \ref{tab:results_m500_k4_030_b}.
\end{itemize}

Overall, these results indicate that the findings do not differ significantly from those in the main text. In particular, under $\rho = 0.30$ (weak spatial correlation) in Scenario 1 (where data are generated from the spatial FH model), unlike the case of $\rho = 0.95$, the non-spatial FH model tended to exhibit better performance than the spatial FH model in terms of AAD and ASD.

\begin{table}[htbp]
\centering
\caption{Simulation results for $m=100$, $k=2$ and $\rho=0.30$ under variance scenario (a)}
\label{tab:results_m100_k2_030_a}
\begin{tabular}{llcccc}
\toprule
Scenario & Method & AAD & ASD & CP & AL \\
\midrule
\textbf{Scenario 1} & SpaHS  & 0.366 & 0.213 & 0.955 & 1.960 \\
                  & SpaGa  & \textbf{0.344} & \textbf{0.186} & 0.925 & 1.561 \\
                  & SpaFH  & 0.348 & \textbf{0.186} & 0.933 & 1.608 \\
                  & HS     & 0.357 & 0.203 & 0.925 & 1.731 \\
                  & FH     & 0.345 & 0.187 & 0.955 & 1.732 \\
\midrule
\textbf{Scenario 2} & SpaHS  & \textbf{0.495} & \textbf{0.410} & 0.935 & 2.400 \\
                  & SpaGa  & 0.519 & 0.440 & 0.950 & 2.506 \\
                  & SpaFH  & 0.542 & 0.458 & 0.950 & 2.632 \\
                  & HS     & 0.569 & 0.514 & 0.945 & 2.688 \\
                  & FH     & 0.537 & 0.458 & 0.950 & 2.634 \\
\midrule
\textbf{Scenario 3} & SpaHS  & \textbf{0.423} & \textbf{0.294} & 0.950 & 2.143 \\
                  & SpaGa  & 0.456 & 0.340 & 0.955 & 2.292 \\
                  & SpaFH  & 0.522 & 0.437 & 0.950 & 2.586 \\
                  & HS     & 0.449 & 0.342 & 0.955 & 2.315 \\
                  & FH     & 0.519 & 0.432 & 0.950 & 2.588 \\
\midrule
\textbf{Scenario 4} & SpaHS  & 0.507 & 0.453 & 0.920 & 2.199 \\
                  & SpaGa  & 0.534 & 0.458 & 0.897 & 2.338 \\
                  & SpaFH  & \textbf{0.496} & \textbf{0.399} & 0.948 & 2.435 \\
                  & HS     & 0.500 & 0.434 & 0.935 & 2.311 \\
                  & FH     & 0.499 & \textbf{0.399} & 0.950 & 2.449 \\
\midrule
\textbf{Scenario 5} & SpaHS  & \textbf{0.342} & \textbf{0.260} & 0.980 & 2.057 \\
                  & SpaGa  & 0.504 & 0.460 & 0.890 & 2.122 \\
                  & SpaFH  & 0.475 & 0.368 & 0.945 & 2.341 \\
                  & HS     & 0.361 & 0.287 & 0.970 & 2.071 \\
                  & FH     & 0.478 & 0.371 & 0.945 & 2.349 \\
\bottomrule
\end{tabular}
\end{table}

\begin{table}[htbp]
\centering
\caption{Simulation results for $m=200$, $k=2$ and $\rho=0.30$ under variance scenario (a)}
\label{tab:results_m200_k2_030_a}
\begin{tabular}{llcccc}
\toprule
Scenario & Method & AAD & ASD & CP & AL \\
\midrule
\textbf{Scenario 1} & SpaHS  & 0.356 & 0.209 & 0.954 & 1.938 \\
                  & SpaGa  & 0.340 & 0.182 & 0.920 & 1.520 \\
                  & SpaFH  & \textbf{0.337} & 0.179 & 0.941 & 1.588 \\
                  & HS     & 0.351 & 0.201 & 0.917 & 1.646 \\
                  & FH     & \textbf{0.337} & \textbf{0.178} & 0.948 & 1.667 \\
\midrule
\textbf{Scenario 2} & SpaHS  & \textbf{0.495} & \textbf{0.400} & 0.940 & 2.375 \\
                  & SpaGa  & 0.516 & 0.434 & 0.950 & 2.508 \\
                  & SpaFH  & 0.539 & 0.457 & 0.950 & 2.634 \\
                  & HS     & 0.575 & 0.520 & 0.940 & 2.683 \\
                  & FH     & 0.538 & 0.455 & 0.950 & 2.633 \\
\midrule
\textbf{Scenario 3} & SpaHS  & \textbf{0.420} & \textbf{0.297} & 0.944 & 2.115 \\
                  & SpaGa  & 0.462 & 0.347 & 0.950 & 2.267 \\
                  & SpaFH  & 0.529 & 0.445 & 0.948 & 2.582 \\
                  & HS     & 0.452 & 0.351 & 0.953 & 2.303 \\
                  & FH     & 0.529 & 0.444 & 0.946 & 2.582 \\
\midrule
\textbf{Scenario 4} & SpaHS  & 0.503 & 0.461 & 0.915 & 2.164 \\
                  & SpaGa  & 0.537 & 0.466 & 0.897 & 2.316 \\
                  & SpaFH  & 0.500 & \textbf{0.402} & 0.948 & 2.429 \\
                  & HS     & \textbf{0.495} & 0.428 & 0.930 & 2.310 \\
                  & FH     & 0.500 & 0.404 & 0.945 & 2.439 \\
\midrule
\textbf{Scenario 5} & SpaHS  & \textbf{0.335} & \textbf{0.269} & 0.970 & 2.022 \\
                  & SpaGa  & 0.493 & 0.467 & 0.879 & 2.118 \\
                  & SpaFH  & 0.470 & 0.362 & 0.944 & 2.328 \\
                  & HS     & 0.355 & 0.294 & 0.963 & 2.015 \\
                  & FH     & 0.474 & 0.365 & 0.945 & 2.342 \\
\bottomrule
\end{tabular}
\end{table}

\begin{table}[htbp]
\centering
\caption{Simulation results for $m=500$, $k=2$ and $\rho=0.30$ under variance scenario (a)}
\label{tab:results_m500_k2_030_a}
\begin{tabular}{llcccc}
\toprule
Scenario & Method & AAD & ASD & CP & AL \\
\midrule
\textbf{Scenario 1} & SpaHS  & 0.355 & 0.206 & 0.954 & 1.924 \\
                  & SpaGa  & 0.334 & 0.177 & 0.921 & 1.509 \\
                  & SpaFH  & \textbf{0.330} & \textbf{0.172} & 0.944 & 1.599 \\
                  & HS     & 0.352 & 0.197 & 0.901 & 1.546 \\
                  & FH     & \textbf{0.330} & 0.174 & 0.949 & 1.636 \\
\midrule
\textbf{Scenario 2} & SpaHS  & \textbf{0.492} & \textbf{0.407} & 0.932 & 2.361 \\
                  & SpaGa  & 0.522 & 0.440 & 0.946 & 2.503 \\
                  & SpaFH  & 0.540 & 0.465 & 0.948 & 2.631 \\
                  & HS     & 0.583 & 0.537 & 0.935 & 2.697 \\
                  & FH     & 0.540 & 0.466 & 0.949 & 2.632 \\
\midrule
\textbf{Scenario 3} & SpaHS  & \textbf{0.419} & \textbf{0.301} & 0.941 & 2.093 \\
                  & SpaGa  & 0.452 & 0.345 & 0.945 & 2.227 \\
                  & SpaFH  & 0.529 & 0.451 & 0.948 & 2.582 \\
                  & HS     & 0.452 & 0.352 & 0.952 & 2.298 \\
                  & FH     & 0.530 & 0.449 & 0.949 & 2.583 \\
\midrule
\textbf{Scenario 4} & SpaHS  & 0.500 & 0.461 & 0.911 & 2.118 \\
                  & SpaGa  & 0.535 & 0.457 & 0.891 & 2.306 \\
                  & SpaFH  & 0.495 & \textbf{0.390} & 0.948 & 2.418 \\
                  & HS     & \textbf{0.489} & 0.421 & 0.932 & 2.294 \\
                  & FH     & 0.496 & 0.392 & 0.948 & 2.428 \\
\midrule
\textbf{Scenario 5} & SpaHS  & \textbf{0.329} & \textbf{0.263} & 0.971 & 2.018 \\
                  & SpaGa  & 0.512 & 0.464 & 0.875 & 2.133 \\
                  & SpaFH  & 0.472 & 0.365 & 0.943 & 2.329 \\
                  & HS     & 0.346 & 0.292 & 0.965 & 2.015 \\
                  & FH     & 0.476 & 0.368 & 0.944 & 2.342 \\
\bottomrule
\end{tabular}
\end{table}

\begin{table}[htbp]
\centering
\caption{Simulation results for $m=100$, $k=2$ and $\rho=0.30$ under variance scenario (b)}
\label{tab:results_m100_k2_030_b}
\begin{tabular}{llcccc}
\toprule
Scenario & Method & AAD & ASD & CP & AL \\
\midrule
\textbf{Scenario 1} & SpaHS  & 0.374 & 0.236 & 0.955 & 2.106 \\
                  & SpaGa  & 0.352 & 0.197 & 0.923 & 1.574 \\
                  & SpaFH  & 0.353 & 0.196 & 0.927 & 1.589 \\
                  & HS     & 0.366 & 0.217 & 0.930 & 1.830 \\
                  & FH     & \textbf{0.349} & \textbf{0.191} & 0.950 & 1.734 \\
\midrule
\textbf{Scenario 2} & SpaHS  & \textbf{0.565} & 0.620 & 0.938 & 2.670 \\
                  & SpaGa  & 0.568 & \textbf{0.599} & 0.945 & 2.732 \\
                  & SpaFH  & 0.585 & 0.609 & 0.948 & 2.912 \\
                  & HS     & 0.637 & 0.731 & 0.930 & 2.962 \\
                  & FH     & 0.583 & 0.604 & 0.945 & 2.910 \\
\midrule
\textbf{Scenario 3} & SpaHS  & \textbf{0.467} & \textbf{0.410} & 0.940 & 2.340 \\
                  & SpaGa  & 0.519 & 0.508 & 0.942 & 2.474 \\
                  & SpaFH  & 0.586 & 0.588 & 0.945 & 2.816 \\
                  & HS     & 0.505 & 0.500 & 0.952 & 2.527 \\
                  & FH     & 0.586 & 0.582 & 0.950 & 2.815 \\
\midrule
\textbf{Scenario 4} & SpaHS  & 0.542 & 0.549 & 0.920 & 2.349 \\
                  & SpaGa  & 0.570 & 0.556 & 0.895 & 2.456 \\
                  & SpaFH  & \textbf{0.531} & 0.488 & 0.945 & 2.581 \\
                  & HS     & \textbf{0.531} & 0.518 & 0.930 & 2.487 \\
                  & FH     & 0.532 & \textbf{0.484} & 0.945 & 2.603 \\
\midrule
\textbf{Scenario 5} & SpaHS  & \textbf{0.372} & \textbf{0.354} & 0.965 & 2.195 \\
                  & SpaGa  & 0.482 & 0.520 & 0.880 & 1.830 \\
                  & SpaFH  & 0.496 & 0.437 & 0.940 & 2.447 \\
                  & HS     & 0.394 & 0.384 & 0.960 & 2.184 \\
                  & FH     & 0.502 & 0.440 & 0.942 & 2.464 \\
\bottomrule
\end{tabular}
\end{table}

\begin{table}[htbp]
\centering
\caption{Simulation results for $m=200$, $k=2$ and $\rho=0.30$ under variance scenario (b)}
\label{tab:results_m200_k2_030_b}
\begin{tabular}{llcccc}
\toprule
Scenario & Method & AAD & ASD & CP & AL \\
\midrule
\textbf{Scenario 1} & SpaHS  & 0.367 & 0.233 & 0.958 & 2.080 \\
                  & SpaGa  & 0.341 & 0.186 & 0.922 & 1.546 \\
                  & SpaFH  & 0.340 & \textbf{0.183} & 0.943 & 1.607 \\
                  & HS     & 0.357 & 0.207 & 0.924 & 1.746 \\
                  & FH     & \textbf{0.339} & \textbf{0.183} & 0.951 & 1.688 \\
\midrule
\textbf{Scenario 2} & SpaHS  & \textbf{0.561} & 0.627 & 0.930 & 2.640 \\
                  & SpaGa  & 0.586 & \textbf{0.621} & 0.945 & 2.736 \\
                  & SpaFH  & 0.597 & 0.639 & 0.945 & 2.909 \\
                  & HS     & 0.659 & 0.770 & 0.936 & 2.970 \\
                  & FH     & 0.594 & 0.633 & 0.948 & 2.908 \\
\midrule
\textbf{Scenario 3} & SpaHS  & \textbf{0.459} & \textbf{0.412} & 0.943 & 2.316 \\
                  & SpaGa  & 0.497 & 0.474 & 0.945 & 2.438 \\
                  & SpaFH  & 0.578 & 0.584 & 0.950 & 2.824 \\
                  & HS     & 0.501 & 0.494 & 0.951 & 2.541 \\
                  & FH     & 0.577 & 0.581 & 0.950 & 2.821 \\
\midrule
\textbf{Scenario 4} & SpaHS  & 0.530 & 0.563 & 0.915 & 2.293 \\
                  & SpaGa  & 0.565 & 0.543 & 0.890 & 2.443 \\
                  & SpaFH  & 0.527 & \textbf{0.481} & 0.943 & 2.570 \\
                  & HS     & \textbf{0.525} & 0.535 & 0.930 & 2.433 \\
                  & FH     & 0.528 & 0.484 & 0.943 & 2.581 \\
\midrule
\textbf{Scenario 5} & SpaHS  & \textbf{0.360} & \textbf{0.337} & 0.968 & 2.177 \\
                  & SpaGa  & 0.485 & 0.520 & 0.874 & 1.824 \\
                  & SpaFH  & 0.489 & 0.421 & 0.946 & 2.457 \\
                  & HS     & 0.374 & 0.363 & 0.963 & 2.151 \\
                  & FH     & 0.490 & 0.427 & 0.943 & 2.464 \\
\bottomrule
\end{tabular}
\end{table}

\begin{table}[htbp]
\centering
\caption{Simulation results for $m=500$, $k=2$ and $\rho=0.30$ under variance scenario (b)}
\label{tab:results_m500_k2_030_b}
\begin{tabular}{llcccc}
\toprule
Scenario & Method & AAD & ASD & CP & AL \\
\midrule
\textbf{Scenario 1} & SpaHS  & 0.366 & 0.226 & 0.954 & 2.068 \\
                  & SpaGa  & 0.339 & 0.182 & 0.921 & 1.523 \\
                  & SpaFH  & \textbf{0.333} & \textbf{0.176} & 0.944 & 1.602 \\
                  & HS     & 0.355 & 0.206 & 0.909 & 1.649 \\
                  & FH     & 0.335 & 0.179 & 0.949 & 1.640 \\
\midrule
\textbf{Scenario 2} & SpaHS  & \textbf{0.560} & \textbf{0.619} & 0.928 & 2.623 \\
                  & SpaGa  & 0.585 & 0.626 & 0.946 & 2.727 \\
                  & SpaFH  & 0.595 & 0.624 & 0.950 & 2.902 \\
                  & HS     & 0.661 & 0.780 & 0.933 & 2.970 \\
                  & FH     & 0.594 & 0.623 & 0.950 & 2.904 \\
\midrule
\textbf{Scenario 3} & SpaHS  & \textbf{0.455} & \textbf{0.411} & 0.942 & 2.278 \\
                  & SpaGa  & 0.494 & 0.464 & 0.948 & 2.393 \\
                  & SpaFH  & 0.574 & 0.575 & 0.948 & 2.814 \\
                  & HS     & 0.498 & 0.496 & 0.951 & 2.530 \\
                  & FH     & 0.574 & 0.575 & 0.950 & 2.817 \\
\midrule
\textbf{Scenario 4} & SpaHS  & 0.520 & 0.534 & 0.911 & 2.263 \\
                  & SpaGa  & 0.560 & 0.531 & 0.899 & 2.455 \\
                  & SpaFH  & 0.521 & \textbf{0.465} & 0.949 & 2.571 \\
                  & HS     & \textbf{0.517} & 0.504 & 0.932 & 2.446 \\
                  & FH     & 0.524 & 0.469 & 0.950 & 2.587 \\
\midrule
\textbf{Scenario 5} & SpaHS  & \textbf{0.355} & \textbf{0.339} & 0.968 & 2.162 \\
                  & SpaGa  & 0.484 & 0.525 & 0.862 & 1.681 \\
                  & SpaFH  & 0.489 & 0.421 & 0.942 & 2.441 \\
                  & HS     & 0.367 & 0.356 & 0.960 & 2.128 \\
                  & FH     & 0.491 & 0.424 & 0.943 & 2.461 \\
\bottomrule
\end{tabular}
\end{table}

\begin{table}[htbp]
\centering
\caption{Simulation results for $m=100$, $k=4$ and $\rho=0.95$ under variance scenario (a)}
\label{tab:results_m100_k4_095_a}
\begin{tabular}{llcccc}
\toprule
Scenario & Method & AAD & ASD & CP & AL \\
\midrule
\textbf{Scenario 1} & SpaHS  & 0.391 & 0.249 & 0.951 & 2.006 \\
                  & SpaGa  & 0.368 & 0.213 & 0.930 & 1.688 \\
                  & SpaFH  & \textbf{0.363} & \textbf{0.207} & 0.929 & 1.651 \\
                  & HS     & 0.414 & 0.269 & 0.906 & 1.782 \\
                  & FH     & 0.390 & 0.242 & 0.945 & 1.850 \\
\midrule
\textbf{Scenario 2} & SpaHS  & \textbf{0.518} & \textbf{0.443} & 0.930 & 2.438 \\
                  & SpaGa  & 0.524 & 0.447 & 0.950 & 2.547 \\
                  & SpaFH  & 0.541 & 0.468 & 0.950 & 2.632 \\
                  & HS     & 0.567 & 0.504 & 0.945 & 2.665 \\
                  & FH     & 0.541 & 0.466 & 0.948 & 2.634 \\
\midrule
\textbf{Scenario 3} & SpaHS  & \textbf{0.451} & \textbf{0.338} & 0.940 & 2.173 \\
                  & SpaGa  & 0.502 & 0.405 & 0.943 & 2.394 \\
                  & SpaFH  & 0.533 & 0.458 & 0.948 & 2.582 \\
                  & HS     & 0.506 & 0.423 & 0.940 & 2.430 \\
                  & FH     & 0.533 & 0.454 & 0.950 & 2.583 \\
\midrule
\textbf{Scenario 4} & SpaHS  & 0.505 & 0.446 & 0.924 & 2.251 \\
                  & SpaGa  & 0.542 & 0.485 & 0.912 & 2.464 \\
                  & SpaFH  & \textbf{0.497} & \textbf{0.396} & 0.948 & 2.438 \\
                  & HS     & 0.530 & 0.455 & 0.936 & 2.454 \\
                  & FH     & 0.513 & 0.423 & 0.949 & 2.524 \\
\midrule
\textbf{Scenario 5} & SpaHS  & \textbf{0.341} & \textbf{0.242} & 0.976 & 2.050 \\
                  & SpaGa  & 0.493 & 0.492 & 0.890 & 2.085 \\
                  & SpaFH  & 0.499 & 0.399 & 0.945 & 2.408 \\
                  & HS     & 0.401 & 0.330 & 0.963 & 2.177 \\
                  & FH     & 0.508 & 0.414 & 0.948 & 2.473 \\
\bottomrule
\end{tabular}
\end{table}

\begin{table}[htbp]
\centering
\caption{Simulation results for $m=200$, $k=4$ and $\rho=0.95$ under variance scenario (a)}
\label{tab:results_m200_k4_095_a}
\begin{tabular}{llcccc}
\toprule
Scenario & Method & AAD & ASD & CP & AL \\
\midrule
\textbf{Scenario 1} & SpaHS  & 0.386 & 0.242 & 0.951 & 1.989 \\
                  & SpaGa  & 0.359 & 0.203 & 0.927 & 1.644 \\
                  & SpaFH  & \textbf{0.352} & \textbf{0.193} & 0.935 & 1.634 \\
                  & HS     & 0.402 & 0.258 & 0.900 & 1.721 \\
                  & FH     & 0.382 & 0.231 & 0.945 & 1.824 \\
\midrule
\textbf{Scenario 2} & SpaHS  & \textbf{0.516} & \textbf{0.436} & 0.924 & 2.415 \\
                  & SpaGa  & 0.524 & 0.446 & 0.947 & 2.542 \\
                  & SpaFH  & 0.539 & 0.470 & 0.945 & 2.634 \\
                  & HS     & 0.562 & 0.505 & 0.939 & 2.668 \\
                  & FH     & 0.540 & 0.470 & 0.946 & 2.635 \\
\midrule
\textbf{Scenario 3} & SpaHS  & \textbf{0.442} & \textbf{0.320} & 0.935 & 2.148 \\
                  & SpaGa  & 0.481 & 0.378 & 0.939 & 2.338 \\
                  & SpaFH  & 0.529 & 0.447 & 0.948 & 2.587 \\
                  & HS     & 0.507 & 0.424 & 0.944 & 2.434 \\
                  & FH     & 0.529 & 0.450 & 0.948 & 2.588 \\
\midrule
\textbf{Scenario 4} & SpaHS  & 0.506 & 0.454 & 0.915 & 2.208 \\
                  & SpaGa  & 0.556 & 0.507 & 0.909 & 2.463 \\
                  & SpaFH  & \textbf{0.503} & \textbf{0.405} & 0.946 & 2.430 \\
                  & HS     & 0.527 & 0.461 & 0.934 & 2.448 \\
                  & FH     & 0.520 & 0.433 & 0.946 & 2.524 \\
\midrule
\textbf{Scenario 5} & SpaHS  & \textbf{0.335} & \textbf{0.246} & 0.976 & 2.032 \\
                  & SpaGa  & 0.474 & 0.491 & 0.838 & 1.642 \\
                  & SpaFH  & 0.493 & 0.390 & 0.946 & 2.395 \\
                  & HS     & 0.401 & 0.332 & 0.960 & 2.157 \\
                  & FH     & 0.502 & 0.404 & 0.948 & 2.457 \\
\bottomrule
\end{tabular}
\end{table}

\begin{table}[htbp]
\centering
\caption{Simulation results for $m=500$, $k=4$ and $\rho=0.95$ under variance scenario (a)}
\label{tab:results_m500_k4_095_a}
\begin{tabular}{llcccc}
\toprule
Scenario & Method & AAD & ASD & CP & AL \\
\midrule
\textbf{Scenario 1} & SpaHS  & 0.384 & 0.239 & 0.952 & 1.979 \\
                  & SpaGa  & 0.352 & 0.195 & 0.930 & 1.634 \\
                  & SpaFH  & \textbf{0.345} & \textbf{0.186} & 0.942 & 1.656 \\
                  & HS     & 0.405 & 0.260 & 0.892 & 1.682 \\
                  & FH     & 0.378 & 0.227 & 0.945 & 1.839 \\
\midrule
\textbf{Scenario 2} & SpaHS  & \textbf{0.510} & \textbf{0.430} & 0.920 & 2.398 \\
                  & SpaGa  & 0.521 & 0.440 & 0.948 & 2.543 \\
                  & SpaFH  & 0.537 & 0.461 & 0.949 & 2.634 \\
                  & HS     & 0.558 & 0.495 & 0.943 & 2.669 \\
                  & FH     & 0.536 & 0.460 & 0.949 & 2.635 \\
\midrule
\textbf{Scenario 3} & SpaHS  & \textbf{0.440} & \textbf{0.324} & 0.933 & 2.130 \\
                  & SpaGa  & 0.479 & 0.377 & 0.936 & 2.310 \\
                  & SpaFH  & 0.531 & 0.450 & 0.948 & 2.586 \\
                  & HS     & 0.502 & 0.414 & 0.943 & 2.435 \\
                  & FH     & 0.530 & 0.449 & 0.949 & 2.587 \\
\midrule
\textbf{Scenario 4} & SpaHS  & \textbf{0.495} & 0.446 & 0.918 & 2.182 \\
                  & SpaGa  & 0.554 & 0.503 & 0.909 & 2.456 \\
                  & SpaFH  & 0.499 & \textbf{0.398} & 0.947 & 2.426 \\
                  & HS     & 0.515 & 0.447 & 0.938 & 2.441 \\
                  & FH     & 0.516 & 0.427 & 0.949 & 2.525 \\
\midrule
\textbf{Scenario 5} & SpaHS  & \textbf{0.331} & \textbf{0.242} & 0.976 & 2.022 \\
                  & SpaGa  & 0.461 & 0.474 & 0.807 & 1.587 \\
                  & SpaFH  & 0.491 & 0.389 & 0.946 & 2.402 \\
                  & HS     & 0.392 & 0.323 & 0.963 & 2.168 \\
                  & FH     & 0.503 & 0.407 & 0.947 & 2.471 \\
\bottomrule
\end{tabular}
\end{table}

\begin{table}[htbp]
\centering
\caption{Simulation results for $m=100$, $k=4$ and $\rho=0.95$ under variance scenario (b)}
\label{tab:results_m100_k4_095_b}
\begin{tabular}{llcccc}
\toprule
Scenario & Method & AAD & ASD & CP & AL \\
\midrule
\textbf{Scenario 1} & SpaHS  & 0.397 & 0.270 & 0.953 & 2.155 \\
                  & SpaGa  & 0.367 & 0.214 & 0.930 & 1.715 \\
                  & SpaFH  & \textbf{0.362} & \textbf{0.208} & 0.931 & 1.676 \\
                  & HS     & 0.412 & 0.272 & 0.912 & 1.889 \\
                  & FH     & 0.393 & 0.249 & 0.940 & 1.904 \\
\midrule
\textbf{Scenario 2} & SpaHS  & \textbf{0.581} & 0.644 & 0.922 & 2.698 \\
                  & SpaGa  & \textbf{0.581} & \textbf{0.612} & 0.945 & 2.785 \\
                  & SpaFH  & 0.595 & 0.632 & 0.948 & 2.908 \\
                  & HS     & 0.633 & 0.716 & 0.938 & 2.936 \\
                  & FH     & 0.595 & 0.633 & 0.948 & 2.908 \\
\midrule
\textbf{Scenario 3} & SpaHS  & \textbf{0.484} & \textbf{0.448} & 0.938 & 2.361 \\
                  & SpaGa  & 0.539 & 0.529 & 0.940 & 2.570 \\
                  & SpaFH  & 0.580 & 0.594 & 0.943 & 2.818 \\
                  & HS     & 0.553 & 0.583 & 0.940 & 2.640 \\
                  & FH     & 0.579 & 0.591 & 0.944 & 2.818 \\
\midrule
\textbf{Scenario 4} & SpaHS  & 0.545 & 0.562 & 0.921 & 2.403 \\
                  & SpaGa  & 0.600 & 0.640 & 0.912 & 2.655 \\
                  & SpaFH  & \textbf{0.544} & \textbf{0.504} & 0.943 & 2.589 \\
                  & HS     & 0.579 & 0.603 & 0.930 & 2.607 \\
                  & FH     & 0.564 & 0.549 & 0.945 & 2.717 \\
\midrule
\textbf{Scenario 5} & SpaHS  & \textbf{0.377} & \textbf{0.347} & 0.973 & 2.216 \\
                  & SpaGa  & 0.508 & 0.570 & 0.873 & 1.817 \\
                  & SpaFH  & 0.523 & 0.483 & 0.943 & 2.536 \\
                  & HS     & 0.432 & 0.457 & 0.955 & 2.300 \\
                  & FH     & 0.535 & 0.511 & 0.943 & 2.629 \\
\bottomrule
\end{tabular}
\end{table}

\begin{table}[htbp]
\centering
\caption{Simulation results for $m=200$, $k=4$ and $\rho=0.95$ under variance scenario (b)}
\label{tab:results_m200_k4_095_b}
\begin{tabular}{llcccc}
\toprule
Scenario & Method & AAD & ASD & CP & AL \\
\midrule
\textbf{Scenario 1} & SpaHS  & 0.394 & 0.266 & 0.953 & 2.137 \\
                  & SpaGa  & 0.364 & 0.214 & 0.928 & 1.688 \\
                  & SpaFH  & \textbf{0.356} & \textbf{0.204} & 0.935 & 1.675 \\
                  & HS     & 0.414 & 0.282 & 0.905 & 1.831 \\
                  & FH     & 0.392 & 0.248 & 0.939 & 1.876 \\
\midrule
\textbf{Scenario 2} & SpaHS  & \textbf{0.581} & 0.641 & 0.917 & 2.671 \\
                  & SpaGa  & 0.588 & \textbf{0.627} & 0.944 & 2.781 \\
                  & SpaFH  & 0.600 & 0.640 & 0.946 & 2.909 \\
                  & HS     & 0.639 & 0.737 & 0.934 & 2.938 \\
                  & FH     & 0.598 & 0.637 & 0.946 & 2.907 \\
\midrule
\textbf{Scenario 3} & SpaHS  & \textbf{0.480} & \textbf{0.449} & 0.935 & 2.330 \\
                  & SpaGa  & 0.526 & 0.502 & 0.938 & 2.506 \\
                  & SpaFH  & 0.576 & 0.587 & 0.949 & 2.825 \\
                  & HS     & 0.547 & 0.567 & 0.940 & 2.644 \\
                  & FH     & 0.579 & 0.587 & 0.948 & 2.825 \\
\midrule
\textbf{Scenario 4} & SpaHS  & 0.538 & 0.573 & 0.919 & 2.353 \\
                  & SpaGa  & 0.597 & 0.632 & 0.911 & 2.656 \\
                  & SpaFH  & \textbf{0.536} & \textbf{0.492} & 0.945 & 2.597 \\
                  & HS     & 0.564 & 0.588 & 0.933 & 2.619 \\
                  & FH     & 0.559 & 0.550 & 0.947 & 2.723 \\
\midrule
\textbf{Scenario 5} & SpaHS  & \textbf{0.374} & \textbf{0.354} & 0.973 & 2.206 \\
                  & SpaGa  & 0.503 & 0.593 & 0.833 & 1.676 \\
                  & SpaFH  & 0.525 & 0.477 & 0.945 & 2.550 \\
                  & HS     & 0.431 & 0.441 & 0.959 & 2.315 \\
                  & FH     & 0.538 & 0.507 & 0.946 & 2.640 \\
\bottomrule
\end{tabular}
\end{table}

\begin{table}[htbp]
\centering
\caption{Simulation results for $m=500$, $k=4$ and $\rho=0.95$ under variance scenario (b)}
\label{tab:results_m500_k4_095_b}
\begin{tabular}{llcccc}
\toprule
Scenario & Method & AAD & ASD & CP & AL \\
\midrule
\textbf{Scenario 1} & SpaHS  & 0.392 & 0.260 & 0.953 & 2.126 \\
                  & SpaGa  & 0.356 & 0.204 & 0.930 & 1.655 \\
                  & SpaFH  & \textbf{0.349} & \textbf{0.195} & 0.942 & 1.668 \\
                  & HS     & 0.410 & 0.273 & 0.901 & 1.785 \\
                  & FH     & 0.386 & 0.241 & 0.946 & 1.875 \\
\midrule
\textbf{Scenario 2} & SpaHS  & \textbf{0.575} & 0.632 & 0.916 & 2.661 \\
                  & SpaGa  & 0.583 & \textbf{0.610} & 0.946 & 2.783 \\
                  & SpaFH  & 0.594 & 0.623 & 0.949 & 2.906 \\
                  & HS     & 0.637 & 0.721 & 0.938 & 2.936 \\
                  & FH     & 0.592 & 0.621 & 0.949 & 2.909 \\
\midrule
\textbf{Scenario 3} & SpaHS  & \textbf{0.482} & \textbf{0.455} & 0.932 & 2.314 \\
                  & SpaGa  & 0.519 & 0.495 & 0.936 & 2.486 \\
                  & SpaFH  & 0.580 & 0.587 & 0.948 & 2.821 \\
                  & HS     & 0.551 & 0.576 & 0.941 & 2.656 \\
                  & FH     & 0.581 & 0.588 & 0.949 & 2.825 \\
\midrule
\textbf{Scenario 4} & SpaHS  & 0.532 & 0.553 & 0.921 & 2.333 \\
                  & SpaGa  & 0.599 & 0.629 & 0.906 & 2.639 \\
                  & SpaFH  & \textbf{0.530} & \textbf{0.482} & 0.947 & 2.588 \\
                  & HS     & 0.555 & 0.571 & 0.936 & 2.617 \\
                  & FH     & 0.555 & 0.535 & 0.949 & 2.730 \\
\midrule
\textbf{Scenario 5} & SpaHS  & \textbf{0.365} & \textbf{0.344} & 0.973 & 2.190 \\
                  & SpaGa  & 0.482 & 0.586 & 0.799 & 1.610 \\
                  & SpaFH  & 0.519 & 0.468 & 0.946 & 2.544 \\
                  & HS     & 0.418 & 0.432 & 0.960 & 2.319 \\
                  & FH     & 0.533 & 0.499 & 0.946 & 2.643 \\
\bottomrule
\end{tabular}
\end{table}

\begin{table}[htbp]
\centering
\caption{Simulation results for $m=100$, $k=4$ and $\rho=0.30$ under variance scenario (a)}
\label{tab:results_m100_k4_030_a}
\begin{tabular}{llcccc}
\toprule
Scenario & Method & AAD & ASD & CP & AL \\
\midrule
\textbf{Scenario 1} & SpaHS  & 0.362 & 0.211 & 0.955 & 1.946 \\
                  & SpaGa  & 0.342 & 0.187 & 0.921 & 1.526 \\
                  & SpaFH  & 0.342 & \textbf{0.184} & 0.917 & 1.527 \\
                  & HS     & 0.358 & 0.203 & 0.924 & 1.654 \\
                  & FH     & \textbf{0.340} & \textbf{0.184} & 0.950 & 1.682 \\
\midrule
\textbf{Scenario 2} & SpaHS  & \textbf{0.503} & \textbf{0.420} & 0.943 & 2.414 \\
                  & SpaGa  & 0.519 & 0.438 & 0.953 & 2.536 \\
                  & SpaFH  & 0.542 & 0.471 & 0.950 & 2.632 \\
                  & HS     & 0.563 & 0.508 & 0.943 & 2.664 \\
                  & FH     & 0.541 & 0.471 & 0.948 & 2.631 \\
\midrule
\textbf{Scenario 3} & SpaHS  & \textbf{0.421} & \textbf{0.298} & 0.948 & 2.129 \\
                  & SpaGa  & 0.496 & 0.401 & 0.944 & 2.460 \\
                  & SpaFH  & 0.531 & 0.444 & 0.948 & 2.578 \\
                  & HS     & 0.491 & 0.399 & 0.951 & 2.392 \\
                  & FH     & 0.530 & 0.441 & 0.948 & 2.578 \\
\midrule
\textbf{Scenario 4} & SpaHS  & 0.511 & 0.458 & 0.911 & 2.151 \\
                  & SpaGa  & 0.534 & 0.454 & 0.906 & 2.352 \\
                  & SpaFH  & 0.500 & \textbf{0.396} & 0.945 & 2.393 \\
                  & HS     & 0.505 & 0.435 & 0.927 & 2.310 \\
                  & FH     & \textbf{0.501} & 0.401 & 0.945 & 2.412 \\
\midrule
\textbf{Scenario 5} & SpaHS  & \textbf{0.346} & \textbf{0.266} & 0.970 & 2.027 \\
                  & SpaGa  & 0.516 & 0.461 & 0.884 & 2.228 \\
                  & SpaFH  & 0.480 & 0.372 & 0.938 & 2.316 \\
                  & HS     & 0.392 & 0.323 & 0.956 & 2.069 \\
                  & FH     & 0.479 & 0.371 & 0.940 & 2.333 \\
\bottomrule
\end{tabular}
\end{table}

\begin{table}[htbp]
\centering
\caption{Simulation results for $m=200$, $k=4$ and $\rho=0.30$ under variance scenario (a)}
\label{tab:results_m200_k4_030_a}
\begin{tabular}{llcccc}
\toprule
Scenario & Method & AAD & ASD & CP & AL \\
\midrule
\textbf{Scenario 1} & SpaHS  & 0.357 & 0.209 & 0.954 & 1.925 \\
                  & SpaGa  & 0.335 & 0.178 & 0.917 & 1.473 \\
                  & SpaFH  & 0.333 & 0.177 & 0.923 & 1.481 \\
                  & HS     & 0.348 & 0.194 & 0.914 & 1.575 \\
                  & FH     & \textbf{0.332} & \textbf{0.174} & 0.946 & 1.609 \\
\midrule
\textbf{Scenario 2} & SpaHS  & \textbf{0.497} & \textbf{0.414} & 0.936 & 2.395 \\
                  & SpaGa  & 0.518 & 0.436 & 0.949 & 2.532 \\
                  & SpaFH  & 0.541 & 0.470 & 0.947 & 2.631 \\
                  & HS     & 0.563 & 0.511 & 0.941 & 2.668 \\
                  & FH     & 0.541 & 0.469 & 0.948 & 2.631 \\
\midrule
\textbf{Scenario 3} & SpaHS  & \textbf{0.423} & \textbf{0.303} & 0.942 & 2.105 \\
                  & SpaGa  & 0.462 & 0.360 & 0.943 & 2.263 \\
                  & SpaFH  & 0.527 & 0.450 & 0.946 & 2.582 \\
                  & HS     & 0.487 & 0.399 & 0.944 & 2.383 \\
                  & FH     & 0.527 & 0.450 & 0.946 & 2.582 \\
\midrule
\textbf{Scenario 4} & SpaHS  & 0.507 & 0.470 & 0.899 & 2.092 \\
                  & SpaGa  & 0.529 & 0.445 & 0.907 & 2.346 \\
                  & SpaFH  & \textbf{0.493} & \textbf{0.387} & 0.947 & 2.401 \\
                  & HS     & 0.499 & 0.426 & 0.929 & 2.294 \\
                  & FH     & 0.497 & 0.388 & 0.949 & 2.411 \\
\midrule
\textbf{Scenario 5} & SpaHS  & \textbf{0.339} & \textbf{0.265} & 0.969 & 2.003 \\
                  & SpaGa  & 0.517 & 0.460 & 0.881 & 2.214 \\
                  & SpaFH  & 0.475 & 0.370 & 0.942 & 2.312 \\
                  & HS     & 0.383 & 0.321 & 0.956 & 2.049 \\
                  & FH     & 0.478 & 0.372 & 0.942 & 2.324 \\
\bottomrule
\end{tabular}
\end{table}

\begin{table}[htbp]
\centering
\caption{Simulation results for $m=500$, $k=4$ and $\rho=0.30$ under variance scenario (a)}
\label{tab:results_m500_k4_030_a}
\begin{tabular}{llcccc}
\toprule
Scenario & Method & AAD & ASD & CP & AL \\
\midrule
\textbf{Scenario 1} & SpaHS  & 0.353 & 0.203 & 0.955 & 1.911 \\
                  & SpaGa  & 0.332 & 0.175 & 0.913 & 1.451 \\
                  & SpaFH  & \textbf{0.326} & \textbf{0.168} & 0.933 & 1.526 \\
                  & HS     & 0.346 & 0.189 & 0.897 & 1.501 \\
                  & FH     & 0.327 & 0.169 & 0.943 & 1.585 \\
\midrule
\textbf{Scenario 2} & SpaHS  & \textbf{0.493} & \textbf{0.409} & 0.935 & 2.378 \\
                  & SpaGa  & 0.516 & 0.433 & 0.950 & 2.532 \\
                  & SpaFH  & 0.537 & 0.461 & 0.949 & 2.631 \\
                  & HS     & 0.559 & 0.497 & 0.943 & 2.668 \\
                  & FH     & 0.536 & 0.460 & 0.950 & 2.632 \\
\midrule
\textbf{Scenario 3} & SpaHS  & \textbf{0.417} & \textbf{0.297} & 0.941 & 2.084 \\
                  & SpaGa  & 0.449 & 0.339 & 0.946 & 2.220 \\
                  & SpaFH  & 0.527 & 0.443 & 0.949 & 2.581 \\
                  & HS     & 0.483 & 0.391 & 0.948 & 2.385 \\
                  & FH     & 0.526 & 0.442 & 0.949 & 2.582 \\
\midrule
\textbf{Scenario 4} & SpaHS  & 0.506 & 0.466 & 0.897 & 2.049 \\
                  & SpaGa  & 0.524 & 0.440 & 0.903 & 2.339 \\
                  & SpaFH  & \textbf{0.489} & \textbf{0.380} & 0.949 & 2.396 \\
                  & HS     & 0.492 & 0.419 & 0.930 & 2.293 \\
                  & FH     & 0.491 & 0.384 & 0.949 & 2.407 \\
\midrule
\textbf{Scenario 5} & SpaHS  & \textbf{0.329} & \textbf{0.261} & 0.971 & 1.996 \\
                  & SpaGa  & 0.515 & 0.450 & 0.877 & 2.217 \\
                  & SpaFH  & 0.468 & 0.358 & 0.946 & 2.315 \\
                  & HS     & 0.372 & 0.315 & 0.958 & 2.045 \\
                  & FH     & 0.471 & 0.360 & 0.947 & 2.324 \\
\bottomrule
\end{tabular}
\end{table}

\begin{table}[htbp]
\centering
\caption{Simulation results for $m=100$, $k=4$ and $\rho=0.30$ under variance scenario (b)}
\label{tab:results_m100_k4_030_b}
\begin{tabular}{llcccc}
\toprule
Scenario & Method & AAD & ASD & CP & AL \\
\midrule
\textbf{Scenario 1} & SpaHS  & 0.369 & 0.236 & 0.958 & 2.087 \\
                  & SpaGa  & \textbf{0.344} & 0.191 & 0.920 & 1.539 \\
                  & SpaFH  & 0.347 & 0.193 & 0.915 & 1.532 \\
                  & HS     & 0.358 & 0.212 & 0.927 & 1.749 \\
                  & FH     & \textbf{0.344} & \textbf{0.189} & 0.944 & 1.679 \\
\midrule
\textbf{Scenario 2} & SpaHS  & \textbf{0.562} & 0.618 & 0.934 & 2.664 \\
                  & SpaGa  & 0.573 & \textbf{0.593} & 0.950 & 2.774 \\
                  & SpaFH  & 0.596 & 0.624 & 0.949 & 2.902 \\
                  & HS     & 0.637 & 0.711 & 0.938 & 2.931 \\
                  & FH     & 0.597 & 0.623 & 0.949 & 2.901 \\
\midrule
\textbf{Scenario 3} & SpaHS  & \textbf{0.463} & \textbf{0.423} & 0.945 & 2.309 \\
                  & SpaGa  & 0.514 & 0.503 & 0.943 & 2.486 \\
                  & SpaFH  & 0.575 & 0.589 & 0.948 & 2.815 \\
                  & HS     & 0.536 & 0.559 & 0.943 & 2.607 \\
                  & FH     & 0.573 & 0.584 & 0.949 & 2.814 \\
\midrule
\textbf{Scenario 4} & SpaHS  & 0.544 & 0.547 & 0.907 & 2.278 \\
                  & SpaGa  & 0.557 & 0.533 & 0.907 & 2.515 \\
                  & SpaFH  & \textbf{0.532} & \textbf{0.472} & 0.941 & 2.555 \\
                  & HS     & 0.545 & 0.526 & 0.920 & 2.439 \\
                  & FH     & 0.536 & 0.480 & 0.940 & 2.569 \\
\midrule
\textbf{Scenario 5} & SpaHS  & \textbf{0.369} & \textbf{0.345} & 0.966 & 2.172 \\
                  & SpaGa  & 0.503 & 0.522 & 0.875 & 2.192 \\
                  & SpaFH  & 0.496 & 0.439 & 0.935 & 2.428 \\
                  & HS     & 0.408 & 0.410 & 0.951 & 2.152 \\
                  & FH     & 0.499 & 0.443 & 0.938 & 2.450 \\
\bottomrule
\end{tabular}
\end{table}

\begin{table}[htbp]
\centering
\caption{Simulation results for $m=200$, $k=4$ and $\rho=0.30$ under variance scenario (b)}
\label{tab:results_m200_k4_030_b}
\begin{tabular}{llcccc}
\toprule
Scenario & Method & AAD & ASD & CP & AL \\
\midrule
\textbf{Scenario 1} & SpaHS  & 0.366 & 0.230 & 0.959 & 2.073 \\
                  & SpaGa  & 0.337 & 0.180 & 0.922 & 1.509 \\
                  & SpaFH  & 0.338 & 0.179 & 0.927 & 1.538 \\
                  & HS     & 0.354 & 0.201 & 0.919 & 1.676 \\
                  & FH     & \textbf{0.335} & \textbf{0.178} & 0.949 & 1.651 \\
\midrule
\textbf{Scenario 2} & SpaHS  & \textbf{0.568} & 0.625 & 0.931 & 2.649 \\
                  & SpaGa  & 0.581 & \textbf{0.605} & 0.949 & 2.772 \\
                  & SpaFH  & 0.600 & 0.630 & 0.948 & 2.906 \\
                  & HS     & 0.643 & 0.732 & 0.934 & 2.936 \\
                  & FH     & 0.599 & 0.629 & 0.948 & 2.904 \\
\midrule
\textbf{Scenario 3} & SpaHS  & \textbf{0.457} & \textbf{0.428} & 0.943 & 2.289 \\
                  & SpaGa  & 0.497 & 0.476 & 0.948 & 2.427 \\
                  & SpaFH  & 0.573 & 0.583 & 0.948 & 2.819 \\
                  & HS     & 0.528 & 0.552 & 0.948 & 2.601 \\
                  & FH     & 0.572 & 0.582 & 0.948 & 2.816 \\
\midrule
\textbf{Scenario 4} & SpaHS  & 0.532 & 0.546 & 0.904 & 2.231 \\
                  & SpaGa  & 0.557 & 0.527 & 0.907 & 2.504 \\
                  & SpaFH  & \textbf{0.526} & \textbf{0.467} & 0.948 & 2.550 \\
                  & HS     & 0.533 & 0.518 & 0.928 & 2.435 \\
                  & FH     & 0.528 & 0.472 & 0.949 & 2.556 \\
\midrule
\textbf{Scenario 5} & SpaHS  & \textbf{0.360} & \textbf{0.341} & 0.966 & 2.154 \\
                  & SpaGa  & 0.493 & 0.512 & 0.870 & 2.160 \\
                  & SpaFH  & 0.491 & 0.424 & 0.940 & 2.418 \\
                  & HS     & 0.402 & 0.387 & 0.953 & 2.159 \\
                  & FH     & 0.493 & 0.429 & 0.943 & 2.427 \\
\bottomrule
\end{tabular}
\end{table}

\begin{table}[htbp]
\centering
\caption{Simulation results for $m=500$, $k=4$ and $\rho=0.30$ under variance scenario (b)}
\label{tab:results_m500_k4_030_b}
\begin{tabular}{llcccc}
\toprule
Scenario & Method & AAD & ASD & CP & AL \\
\midrule
\textbf{Scenario 1} & SpaHS  & 0.360 & 0.222 & 0.956 & 2.055 \\
                  & SpaGa  & 0.335 & 0.179 & 0.912 & 1.471 \\
                  & SpaFH  & \textbf{0.328} & \textbf{0.172} & 0.939 & 1.555 \\
                  & HS     & 0.347 & 0.194 & 0.907 & 1.597 \\
                  & FH     & \textbf{0.328} & \textbf{0.172} & 0.948 & 1.610 \\
\midrule
\textbf{Scenario 2} & SpaHS  & \textbf{0.559} & 0.619 & 0.930 & 2.632 \\
                  & SpaGa  & 0.575 & \textbf{0.602} & 0.949 & 2.767 \\
                  & SpaFH  & 0.592 & 0.624 & 0.950 & 2.900 \\
                  & HS     & 0.639 & 0.725 & 0.937 & 2.935 \\
                  & FH     & 0.592 & 0.622 & 0.949 & 2.903 \\
\midrule
\textbf{Scenario 3} & SpaHS  & \textbf{0.452} & \textbf{0.413} & 0.941 & 2.274 \\
                  & SpaGa  & 0.489 & 0.455 & 0.945 & 2.392 \\
                  & SpaFH  & 0.578 & 0.580 & 0.948 & 2.811 \\
                  & HS     & 0.529 & 0.544 & 0.947 & 2.618 \\
                  & FH     & 0.577 & 0.579 & 0.948 & 2.815 \\
\midrule
\textbf{Scenario 4} & SpaHS  & 0.531 & 0.548 & 0.895 & 2.184 \\
                  & SpaGa  & 0.554 & 0.521 & 0.906 & 2.485 \\
                  & SpaFH  & \textbf{0.519} & \textbf{0.459} & 0.947 & 2.538 \\
                  & HS     & 0.522 & 0.510 & 0.926 & 2.425 \\
                  & FH     & 0.522 & 0.464 & 0.948 & 2.554 \\
\midrule
\textbf{Scenario 5} & SpaHS  & \textbf{0.350} & \textbf{0.333} & 0.969 & 2.138 \\
                  & SpaGa  & 0.470 & 0.531 & 0.835 & 1.551 \\
                  & SpaFH  & 0.483 & 0.410 & 0.944 & 2.422 \\
                  & HS     & 0.388 & 0.371 & 0.954 & 2.151 \\
                  & FH     & 0.486 & 0.414 & 0.944 & 2.439 \\
\bottomrule
\end{tabular}
\end{table}

\section{Additional information on real data analysis}
\label{sec:supp-2}


\subsection{Sensitivity analysis for the spatial dependence parameter $\rho$}
\label{app:sensitivity_rho}

In our proposed spatial multivariate Fay-Herriot model, the spatial dependence parameter $\rho$ is updated using a grid-based sampling method, as its full conditional distribution does not belong to a standard probability distribution. To evaluate the robustness of the posterior inference against the choice of the discrete grid for $\rho$, we conducted a sensitivity analysis using various grid configurations. 

Specifically, we compared the following four grid patterns for $\rho \in [0, 1)$:
\begin{itemize}
    \item \textbf{Grid 1 (Default):} A non-uniform grid with finer increments near the upper boundary ($\rho \in [0, 0.8]$ by $0.05$, $[0.82, 0.90]$ by $0.02$, and $[0.91, 0.99]$ by $0.01$; total 31 points), which is used in the main analysis.
    \item \textbf{Grid 2 (Uniform 0.01):} A dense, uniform grid with a constant step size of $0.01$ (total 99 points).
    \item \textbf{Grid 3 (Coarse 0.05):} A coarse, uniform grid with a step size of $0.05$ (total 19 points).
    \item \textbf{Grid 4 (Boundary-Fine):} A specialized grid with ultra-fine intervals near the boundary ($\rho \in [0, 0.70]$ by $0.1$, $[0.72, 0.90]$ by $0.02$, and $[0.905, 0.995]$ by $0.005$; total 37 points).
\end{itemize}

Each model was fitted using $10,000$ MCMC iterations with the first $5,000$ iterations discarded as burn-in. Table~\ref{tab:rho_sensitivity} summarizes the results, including the CPU execution time, posterior summary statistics of $\rho$ (mean, median, and 95\% credible intervals), and the Deviance Information Criterion (DIC). Furthermore, the Pearson correlation coefficients between the posterior means of the small-area estimates $\bm{\theta}_i$ under the alternative grids and those under the default grid (Grid 1) were computed to assess point-estimation stability.

\begin{table}[htbp]
\centering
\caption{Sensitivity analysis results for different grid choices of $\rho$ based on the real data analysis.}
\label{tab:rho_sensitivity}
\small
\begin{tabular}{lcccc}
\hline
Grid Specification & Grid 1 & Grid 2  & Grid 3  & Grid 4  \\ \hline
Grid Size              & 31              & 99                    & 19                   & 37                     \\
CPU Time (sec)         & 642.63          & 899.41                & 600.24               & 669.78                 \\
$\rho$ Mean (Median)   & 0.9043 (0.910)  & 0.8971 (0.900)        & 0.8928 (0.900)       & 0.9095 (0.915)         \\
95\% Credible Interval & [0.820, 0.960]  & [0.810, 0.960]        & [0.800, 0.950]       & [0.820, 0.965]         \\
DIC                    & 1562.04         & 1560.07               & 1562.68              & 1561.15                \\
$\bm{\theta}$ Correlation (Income) & --       & 0.99999               & 0.99996              & 0.99999                \\
$\bm{\theta}$ Correlation (Poverty)& --       & 0.99998               & 0.99993              & 0.99999                \\ \hline
\end{tabular}
\end{table}

As shown in Table~\ref{tab:rho_sensitivity}, the posterior estimates of $\rho$ remain extremely stable across all grid configurations, with posterior means ranging between $0.89$ and $0.91$, and medians around $0.90$ to $0.92$. The DIC values and the point estimates of the small-area means $\bm{\theta}$ show negligible discrepancies (correlations exceed $0.9999$ for both variables in all cases). While the uniform fine grid (Grid 2) incurs a higher computational cost due to its larger grid size, the default grid (Grid 1) strikes an optimal balance between numerical accuracy and computational efficiency. Overall, these findings confirm that the choice of the grid for $\rho$ has virtually no impact on the substantive conclusions of the paper.

\subsection{Sensitivity analysis for the degrees of freedom parameter $\nu_0$ in the inverse-Wishart prior}
\label{app:sensitivity_nu0}

In our proposed model, the error covariance matrix $\mathbf{\Sigma}$ is assigned an inverse-Wishart prior, $\text{IW}(\nu_0, \mathbf{B})$, where $\nu_0$ represents the degrees of freedom parameter. To evaluate the robustness of the posterior inference against the choice of $\nu_0$, we conducted a sensitivity analysis using different specifications for the prior degrees of freedom.

Specifically, we compared the following three specifications for $\nu_0$ (where $k=2$ is the dimension of the response variables):
\begin{itemize}
    \item \textbf{Specification 1 ($\nu_0 = 2$):} The baseline configuration where $\nu_0 = k$, representing a diffuse prior centered at the dimension size.
    \item \textbf{Specification 2 ($\nu_0 = 4$):} A moderately informative prior configuration with $\nu_0 = k + 2$.
    \item \textbf{Specification 3 ($\nu_0 = 7$):} A more informative prior configuration with $\nu_0 = k + 5$.
\end{itemize}

Each model was fitted using $10,000$ MCMC iterations with the first $5,000$ iterations discarded as burn-in. Table~\ref{tab:nu0_sensitivity} summarizes the results, including the CPU execution time, posterior summary statistics of the spatial dependence parameter $\rho$ (mean and median), the posterior mean of the inter-variable correlation $\Sigma_{12}$, and the Deviance Information Criterion (DIC). Furthermore, to visualize the impact on the covariance structure, Figure~\ref{fig:nu0_sensitivity} displays the posterior densities of the inter-variable correlation ($\Sigma_{12}$) under the different $\nu_0$ specifications.

\begin{table}[htbp]
\centering
\caption{Sensitivity analysis results for different choices of $\nu_0$ in the inverse-Wishart prior based on the real data analysis.}
\label{tab:nu0_sensitivity}
\small
\begin{tabular}{lccc}
\hline
Specification & $\nu_0 = 2$ (Baseline) & $\nu_0 = 4$ & $\nu_0 = 7$ \\ \hline
CPU Time (sec)          & 670.69           & 660.70           & 642.38           \\
$\rho$ Mean (Median)    & 0.9066 (0.910)   & 0.9035 (0.910)   & 0.9076 (0.910)   \\
$\Sigma_{12}$ Mean      & -0.5013          & -0.5109          & -0.5231          \\
DIC                     & 1561.09          & 1560.23          & 1560.07          \\ \hline
\end{tabular}
\end{table}

\begin{figure}[htbp]
\centering
\includegraphics[width=0.85\textwidth]{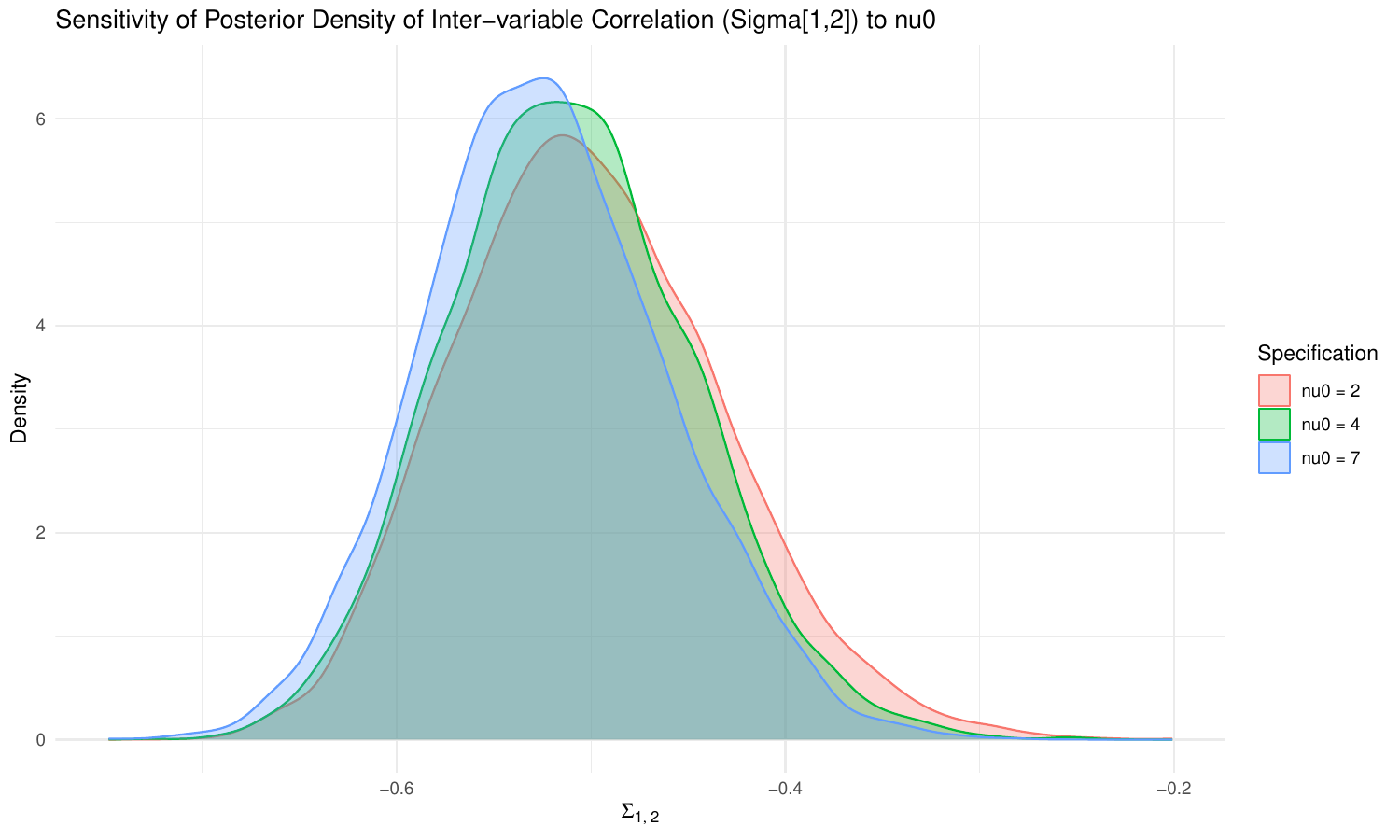}
\caption{Posterior densities of the inter-variable correlation ($\Sigma_{12}$) under different specifications of the degrees of freedom parameter $\nu_0$.}
\label{fig:nu0_sensitivity}
\end{figure}

As shown in Table~\ref{tab:nu0_sensitivity} and Figure~\ref{fig:nu0_sensitivity}, the posterior estimates of both the spatial dependence parameter $\rho$ and the inter-variable correlation $\Sigma_{12}$ remain remarkably stable across all tested values of $\nu_0$. The posterior means of $\rho$ stay virtually identical around $0.90$--$0.91$, and the posterior distributions of $\Sigma_{12}$ exhibit a high degree of overlap regardless of the prior informativeness. Moreover, the DIC values show negligible differences. Overall, these findings confirm that the choice of the degrees of freedom parameter $\nu_0$ in the inverse-Wishart prior has minimal impact on the posterior inference, ensuring the robustness of our proposed framework.

\end{document}